%% file: main.tex
\PassOptionsToPackage{unicode}{hyperref}
\PassOptionsToPackage{hyphens}{url}
\PassOptionsToPackage{dvipsnames,svgnames,x11names}{xcolor}
\documentclass[12pt]{article}

\usepackage{amsmath,amssymb,amsthm}
\usepackage{threeparttable}
\usepackage{mathtools}
\usepackage{iftex}
\ifPDFTeX
  \usepackage[T1]{fontenc}
  \usepackage[utf8]{inputenc}
  \usepackage{textcomp}
\else
  \usepackage{unicode-math}
  \defaultfontfeatures{Scale=MatchLowercase}
  \defaultfontfeatures[\rmfamily]{Ligatures=TeX,Scale=1}
\fi
\usepackage{lmodern}

\IfFileExists{upquote.sty}{\usepackage{upquote}}{}
\IfFileExists{microtype.sty}{%
  \usepackage[]{microtype}
  \UseMicrotypeSet[protrusion]{basicmath}
}{}

\makeatletter
\@ifundefined{KOMAClassName}{%
  \IfFileExists{parskip.sty}{%
    \usepackage{parskip}
  }{%
    \setlength{\parindent}{0pt}
    \setlength{\parskip}{6pt plus 2pt minus 1pt}}
}{%
  \KOMAoptions{parskip=half}}
\makeatother

\usepackage{xcolor}
\colorlet{blue}{black}
\colorlet{Blue}{black}
\colorlet{red}{black}

\usepackage[normalem]{ulem}
\usepackage{cancel}

\makeatletter
\ifx\paragraph\undefined\else
  \let\oldparagraph\paragraph
  \renewcommand{\paragraph}{%
    \@ifstar
      \xxxParagraphStar
      \xxxParagraphNoStar
  }
  \newcommand{\xxxParagraphStar}[1]{\oldparagraph*{#1}\mbox{}}
  \newcommand{\xxxParagraphNoStar}[1]{\oldparagraph{#1}\mbox{}}
\fi
\ifx\subparagraph\undefined\else
  \let\oldsubparagraph\subparagraph
  \renewcommand{\subparagraph}{%
    \@ifstar
      \xxxSubParagraphStar
      \xxxSubParagraphNoStar
  }
  \newcommand{\xxxSubParagraphStar}[1]{\oldsubparagraph*{#1}\mbox{}}
  \newcommand{\xxxSubParagraphNoStar}[1]{\oldsubparagraph{#1}\mbox{}}
\fi
\makeatother

\usepackage{longtable,booktabs,array}
\usepackage{calc}
\usepackage{etoolbox}
\makeatletter
\patchcmd\longtable{\par}{\if@noskipsec\mbox{}\fi\par}{}{}
\makeatother
\IfFileExists{footnotehyper.sty}{\usepackage{footnotehyper}}{\usepackage{footnote}}
\makesavenoteenv{longtable}
\usepackage{graphicx}
\makeatletter
\def\maxwidth{\ifdim\Gin@nat@width>\linewidth\linewidth\else\Gin@nat@width\fi}
\def\maxheight{\ifdim\Gin@nat@height>\textheight\textheight\else\Gin@nat@height\fi}
\makeatother
\setkeys{Gin}{width=\maxwidth,height=\maxheight,keepaspectratio}
\makeatletter
\def\fps@figure{htbp}
\makeatother

\makeatletter
\@ifpackageloaded{caption}{}{\usepackage{caption}}
\AtBeginDocument{%
\ifdefined\contentsname
  \renewcommand*\contentsname{Table of contents}
\else
  \newcommand\contentsname{Table of contents}
\fi
\ifdefined\listfigurename
  \renewcommand*\listfigurename{List of Figures}
\else
  \newcommand\listfigurename{List of Figures}
\fi
\ifdefined\listtablename
  \renewcommand*\listtablename{List of Tables}
\else
  \newcommand\listtablename{List of Tables}
\fi
\ifdefined\figurename
  \renewcommand*\figurename{Figure}
\else
  \newcommand\figurename{Figure}
\fi
\ifdefined\tablename
  \renewcommand*\tablename{Table}
\else
  \newcommand\tablename{Table}
\fi
}
\@ifpackageloaded{float}{}{\usepackage{float}}
\floatstyle{ruled}
\@ifundefined{c@chapter}{\newfloat{codelisting}{h}{lop}}{\newfloat{codelisting}{h}{lop}[chapter]}
\floatname{codelisting}{Listing}

\makeatother

\makeatletter
\@ifpackageloaded{caption}{}{\usepackage{caption}}
\@ifpackageloaded{subcaption}{}{\usepackage{subcaption}}
\makeatother
\usepackage{bbm}        
\usepackage{multirow}   
\usepackage{makecell}   
\usepackage{tabularx}   
\usepackage{lscape}     
\usepackage{algorithm}
\usepackage{algpseudocode}
\usepackage[many]{tcolorbox}
\newtcolorbox{breakablealgorithm}[3][]{
    enhanced,
    breakable,
    blanker,
    boxsep=0pt,
    left=0pt,
    right=0pt,
    top=0pt,
    bottom=0pt,
    before upper={%
      \refstepcounter{algorithm}%
      \label{#3}%
      \addcontentsline{loa}{algorithm}{\protect\numberline{\thealgorithm}{#2}}%
      \hrule height 0.8pt depth 0pt\kern2pt
      \noindent\textbf{Algorithm~\thealgorithm}\ #2\par
      \kern2pt\hrule\kern2pt
    },
    after upper={\kern2pt\hrule},
    #1
}

\ifLuaTeX
  \usepackage{selnolig}
\fi

\theoremstyle{plain}

\newtheorem{theorem}{Theorem}[section]
\newtheorem{lemma}[theorem]{Lemma}

\theoremstyle{plain}
\newtheorem{definition}[theorem]{Definition}
\newtheorem{assumption}{Condition}
\newtheorem{proposition}[theorem]{Proposition}

\newtheorem{remark}{Remark}

\input{command}

\providecommand{\bthe}{\btheta}
\providecommand{\hbthe}{\hbtheta}
\providecommand{\cbthe}{\check{\btheta}}

\providecommand{\hq}{\widehat q}
\providecommand{\hv}{\widehat v}

\usepackage[]{natbib}
\usepackage{bibunits}
\defaultbibliographystyle{abbrvnat}
\defaultbibliography{citation}

\IfFileExists{xurl.sty}{\usepackage{xurl}}{}
\usepackage{hyperref}
\usepackage{bookmark}

\hypersetup{
    pdftitle={High-Dimensional Assisted Learning for Vertically Distributed Data with Blockwise Missingness},
  pdfauthor={Yuwen Long; Shuyuan Wu; Yin Xia},
  pdfkeywords={vertical federated learning, block-missing data, high-dimensional inference, debiased estimation, data perturbation},
  colorlinks=true,
  linkcolor={black},
  filecolor={black},
  citecolor={black},
  urlcolor={black},
  pdfcreator={LaTeX}
}

\newcommand{\anon}{1}

\newcommand{\spacingset}[1]{\renewcommand{\baselinestretch}{#1}\small\normalsize}

\begin{document}

\spacingset{1}


\title{\bf {High-Dimensional Assisted Learning for\\ Vertically Distributed Data with\\ Blockwise Missingness}}

\if1\anon
  \author{Yuwen Long$^{1}$, Shuyuan Wu$^{2}$, and Yin Xia$^{1}$\\[0.5em]
    \normalsize $^{1}$Department of Statistics and Data Science, Fudan University\\
    \normalsize $^{2}$School of Statistics and Data Science, Shanghai University of Finance and Economics\\[0.4em]}
\else
  \author{}
\fi

\date{}
\maketitle

\begin{bibunit}[abbrvnat]

\begin{abstract}
In multi-institutional studies, different parties hold distinct feature blocks for partially overlapping sets of individuals. Responses may also be missing for some records. In such settings, we propose Assisted Learning with Block-Missing Data (ALB) for sparse high-dimensional linear estimation and coordinatewise inference without pooling records or relying on a coordinating server. ALB minimizes a regularized available-case quadratic loss using cyclic block updates. Each cycle communicates $O(n)$ scalars through sample-level linear summaries, regardless of data dimension $p$, and the iterates converge geometrically to the centralized solution. We derive estimation rates that separate statistical and optimization errors. For inference on a target coefficient, ALB estimates the corresponding precision column and uses a sample-level variance estimator that accounts for dependence among moments computed from overlapping samples. Under sparsity and overlap conditions, the studentized estimator is asymptotically standard normal at the $\sqrt n$ rate, even when there are no complete cases. We also study one-time perturbed covariate and response releases that reduce direct disclosure by replacing unperturbed sample-level quantities with noisy versions. Simulations and an analysis of multimodal Alzheimer's Disease Neuroimaging Initiative data indicate that ALB approximates its centralized benchmark and improves upon complete-case Lasso by incorporating partially observed records.
\end{abstract}

\noindent%

\vfill

\spacingset{1.2} 


\section{Introduction}
\subsection{Motivation}
Modern studies often combine covariates held by different institutions. In health care, for example, clinical, imaging, and biomarker data may reside at separate institutions. Similar arrangements arise in finance, online platforms, and public agencies. Privacy, governance, and ownership constraints can preclude pooling individual records, yielding \emph{vertically distributed data}, in which different parties hold distinct feature blocks for overlapping individuals \citep{liu2024vertical,yin2026vertical}. We therefore seek a decentralized procedure that exchanges only task-specific summaries and requires no coordinating server.

The Alzheimer's Disease Neuroimaging Initiative (ADNI) illustrates a more general setting \citep{mueller2005alzheimer}: many visits include only some biomarker and imaging modalities, and responses may also be missing. Consequently, the sample IDs available across feature blocks overlap only partially. After alignment, this \emph{partially collated} structure appears as blockwise missingness, shown in Figure \ref{fig:diagram}. Complete-case analysis may therefore discard most records and much of the available information.

This setting presents several challenges. First, even for individuals observed by all parties, no single party can evaluate a loss involving all feature blocks from its local data. Second, the covariance between two feature blocks can be computed only from IDs observed by both users, while a feature--response product can be computed only from IDs with both that feature block and the response. Because these quantities use different samples, combining them need not produce a positive semidefinite covariance matrix. Third, high dimensionality calls for algorithms whose communication costs do not grow with $p$. Finally, even when raw records remain local, exchanged summaries may disclose individual-level information. We therefore seek to use every ID available for each covariance or feature--response calculation without pooling individual records.

\begin{figure}[t]
        \centering
        \includegraphics[width=1\linewidth]{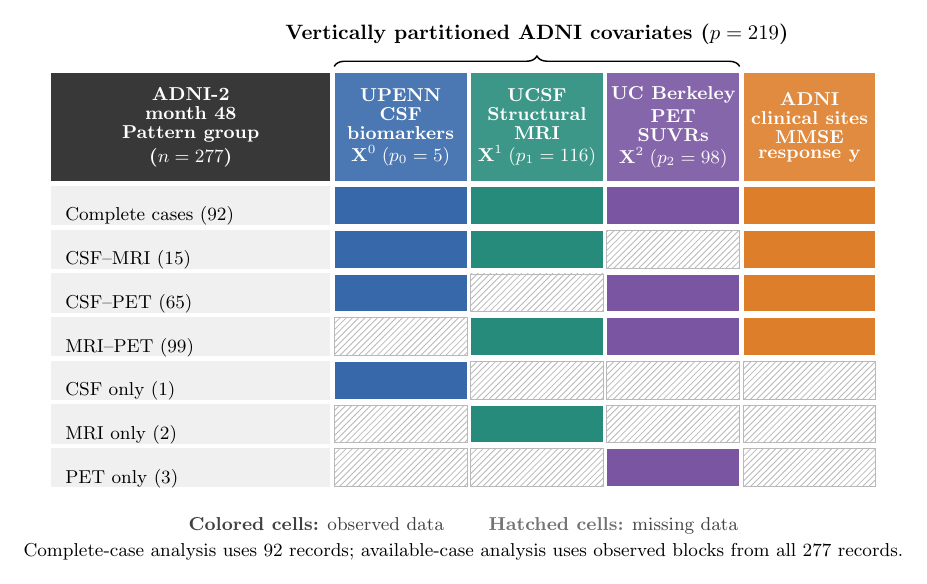}
        \caption{Block-missing pattern in the ADNI-2 month-48 data ($n=277$, $p=219$). The covariates comprise CSF biomarkers ($p_0=5$), structural MRI measures ($p_1=116$), and PET SUVRs ($p_2=98$); MMSE is observed for the 92 complete and 179 pairwise-overlap records.}
        \label{fig:diagram}
    \end{figure}

\subsection{Related Work}

Related work spans three areas corresponding to these challenges: vertically and assisted learning, methods for block-missing data, and high-dimensional inference. The first concerns joint learning when each party retains its own feature block. Vertical federated learning (VFL) has been developed for generalized linear models, trees and forests, regularized regression, clustering, survival and quantile regression, and neural networks \citep{hardy2017private,fang2021large,liu2020federated,huang2022coresets,dai2020verticox,fan2023residual,wang2023unified}. Within VFL, communication costs can be reduced through less frequent synchronization or compressed intermediate quantities \citep{liu2022fedbcd,castiglia2022compressed}, and high-dimensional methods have considered feature selection and screening \citep{castiglia2023less,yin2026vertical}. Most VFL procedures, however, assume aligned samples and coordination by an active party or server \citep{liu2024vertical}. In a related line of research, assisted learning uses decentralized task-specific summaries in a small number of rounds, including reciprocal, gradient-based, and collaborator-screening procedures \citep{xian2020assisted,wang2022parallel,diao2022gal,zhang2026additive}, but does not address block missingness.

Separately, block-missing methods exploit incomplete records through multi-source modeling, structured matrix completion or imputation, and direct use of observed blocks for source selection, sparse prediction, and semi-supervised inference \citep{yuan2012multi,xiang2013multi,cai2016structured,xue2021integrating,xue2025statistical,xiang2014bilevel,yu2020optimal,song2024semi}. These methods show that partial records can improve estimation and inference, but assume centralized computation.

Bridging these areas, recent VFL methods handle incomplete overlap through representation completion and pseudo-labels \citep{kang2022fedcvt}, semi-supervised limited-round training \citep{sun2023communication}, knowledge distillation \citep{ren2022improving}, parameter sharing and task sampling \citep{valdeira2025vertical}, or Gaussian-copula privatization \citep{tang2025data}. These approaches improve prediction, model availability, or privatization, yet do not support decentralized high-dimensional coefficient estimation and inference with unequal overlaps.

For coefficient-level inference, debiasing methods provide valid inference after sparse estimation in centralized high-dimensional models \citep{javanmard2014confidence,zhang2014confidence,van2014asymptotically}. Their distributed extensions aggregate or iteratively refine local estimators and test statistics \citep{lee2017communication,battey2018distributed,cai2022individual,sun2024optimal,jordan2019communication,fan2023communication,chen2022first,li2022statistical,tu2023distributed}, but largely assume row-partitioned designs in which every site observes the same covariates.

Taken together, existing methods do not provide communication-efficient and decentralized high-dimensional coefficient estimation and inference with block-missing data. The proposed Assisted Learning with Block-Missing Data (ALB) method fills this gap through decentralized updates that target the centralized available-case estimator and a variance calculation that retains the dependence induced by \textcolor{blue}{overlapping} IDs. This yields coefficient-level inference even when the complete-case overlap is empty.

\subsection{Contributions} 
The methodological and theoretical contributions are threefold:
\begin{itemize}
    \item \textbf{Decentralized optimization of the available-case objective.} Different covariance and response--covariate blocks are computed from different sets of sample IDs, so no user holds the full objective. ALB computes the cross-user terms from sample-level linear predictions and applies cyclic updates without a \textcolor{blue}{centralized} server. For a fixed number of users, each update cycle communicates $O(n)$ scalars independently of $p$. We establish geometric convergence to the centralized available-case estimator and \textcolor{red}{combine the optimization and centralized high-dimensional statistical errors to derive an estimation-error bound for ALB}.

    \item \textcolor{blue}{\textbf{Decentralized coefficient-level inference under unequal overlaps.} Because the covariance and \textcolor{red}{response--covariate} blocks are computed from partially overlapping samples and the feature blocks are held by different users, we estimate the required precision column and score variance in a decentralized manner while incorporating all available samples. We establish coordinatewise asymptotic normality and consistent variance estimation even when no complete cases are available.}

    \item \textbf{One-time perturbation of transmitted summaries.} 
    The responses and sample-level linear predictions transmitted between users may reveal individual values.
    We add noise 
    to each released response and covariate row once and reuse the perturbed values in all subsequent messages, while each user computes its within-user covariance from the original local data. Under bounded noise scales, {\color{blue} we show that} the perturbed cross-user covariances and response--covariate averages remain unbiased and retain the max-norm concentration rates of their unperturbed counterparts.
\end{itemize}

The remainder of the paper is organized as follows. Section \ref{sec: methodology} presents ALB estimation and inference, Section \ref{sec: theory} gives the theory, and Section \ref{sec: perturbed estimator} introduces perturbation. Sections \ref{sec: simulation} and \ref{sec: ADNI} report simulations and the ADNI application, and Section \ref{sec: conclusion} concludes. The Supplementary Material contains proofs, an optional local differential privacy calibration, additional simulations, and implementation details.

\section{Distributed Estimation and Inference}
\label{sec: methodology}
This section introduces the block-missing data structure and problem formulation, followed by the decentralized ALB estimator and its debiasing procedure. Throughout, for vectors, $\norm{\cdot}_1$, $\norm{\cdot}_2$, and $\norm{\cdot}_{\infty}$ denote the usual norms, and $\norm{\cdot}_0$ counts nonzero entries. For a matrix $\mathbf A=(a_{ij})$, let $\norm{\mathbf A}_{\max}=\max_{i,j}|a_{ij}|$ and $\norm{\mathbf A}_{\infty}=\max_i\sum_j|a_{ij}|$. We use $\#\mathcal A$ for set cardinality, $\mathrm I_p$ for the $p\times p$ identity matrix, and $\operatorname{tr}(\mathbf A)$ for matrix trace.

\subsection{Problem Setup}
\label{subsec: problem set up}
We consider high-dimensional linear estimation and inference with vertically distributed, block-missing data. Users $B_0,\ldots,B_K$ hold disjoint feature blocks for partially overlapping ID sets, responses may also be missing, and the total dimension $p$ may exceed the number $n$ of distinct IDs. Our goal is to estimate the full coefficient vector and conduct coordinatewise inference without pooling raw feature blocks.

Formally, let $\mathcal I=\{1,\ldots,n\}$ be the union of all sample IDs. User $B_k$ observes the subset $\mathcal I_k\subseteq\mathcal I$ and holds a $p_k$-dimensional covariate block $\bx_i^k$ for each $i\in\mathcal I_k$. Let $n_k=\#\mathcal I_k$ and $\bX^k\in\mathbb R^{n_k\times p_k}$ denote its sample size and design matrix, respectively, and let $p=\sum_{k=0}^Kp_k$. The response is observed on an ID set $\mathcal I^y\subseteq\mathcal I$. For each ID $i$, define $r_i^k=\mathbf1\{i\in\mathcal I_k\}$ and $r_i^y=\mathbf1\{i\in\mathcal I^y\}$. Thus, $\bx_i^k$ is observed when $r_i^k=1$, $y_i$ is observed when $r_i^y=1$, and $r_i^k=0$ means that the entire block held by $B_k$ is missing for ID $i$.

Two types of overlap determine the sample-level quantities. For users $B_k$ and $B_l$, $k\neq l$, let $\mathcal I_{kl}=\mathcal I_k\cap\mathcal I_l$ and $n_{kl}=\#\mathcal I_{kl}$. The matrix $\bX^{k,l}$ contains the rows of $\bX^k$ for these shared IDs, ordered to match $\bX^{l,k}$. For response-covariate moments at user $B_k$, let $\mathcal I_k^y=\mathcal I_k\cap\mathcal I^y$ and $n_k^y=\#\mathcal I_k^y$. Align $\by_k\in\mathbb R^{n_k}$ with $\bX^k$ by setting entries outside $\mathcal I_k^y$ to zero, so $\bX^{k\top}\by_k$ sums only over IDs with both $\bx_i^k$ and $y_i$ observed.

The complete-case set is $\mathcal I_{\mathrm{comp}}=\mathcal I^y\cap\bigcap_{k=0}^K\mathcal I_k$. It may be small or empty even when the user-level counts $n_k$, $n_{kl}$, and $n_k^y$ are substantial. For the theoretical rates, we also use coordinate-level effective sample sizes. Specifically, $N_j^x$ counts IDs with covariate $j$, $N_{jt}^{xx}$ counts IDs with both covariates $j$ and $t$, and $N_j^{xy}$ counts IDs with both covariate $j$ and the response. Under blockwise missingness, if $j$ belongs to $B_k$ and $t$ to $B_l$, these counts reduce to $N_j^x=n_k$, $N_j^{xy}=n_k^y$, and $N_{jt}^{xx}=n_k$ when $k=l$ or $n_{kl}$ when $k\neq l$.

Three types of moments are estimated from different ID sets: $\mathcal I_k$ for within-block moments, $\mathcal I_{kl}$ for cross-block moments, and $\mathcal I_k^y$ for response-covariate moments. Consequently, the assembled covariance need not be positive semidefinite, and no user can form the full empirical loss locally. Complete-case analysis wastes the larger pairwise overlaps, whereas transmitting cross-block covariances requires $p_kp_l$ entries. Together with $p\gtrsim n$, these challenges motivate a regularized available-case target and its decentralized computation. We first define the ideal full-data estimand and then its feasible counterpart.

To define the statistical target, suppose first that the complete covariate vector $\bx=(\bx^{0\top},\ldots,\bx^{K\top})^\top\in\mathbb R^p$ and the response $y$ were jointly available. \textcolor{blue}{Throughout, we assume that the covariates and response are centered, so that $\E(\bx)=\mathbf0$ and $\E(y)=0$}. The ideal population coefficient minimizes the mean squared prediction loss:
    \begin{equation}
    \label{oracle loss 1}
        \bbeta^*=\argmin_{\bbeta\in\mathbb R^p}\cL^*(\bbeta),
        \qquad
        \cL^*(\bbeta)=\frac12\E\{(y-\bx^\top\bbeta)^2\},
    \end{equation}
Here $\bbeta=(\bbeta_0^\top,\ldots,\bbeta_K^\top)^\top$, where $\bbeta_k\in\mathbb R^{p_k}$ corresponds to user $B_k$. Define $\Sigma_{kl}=\E(\bx^k\bx^{l\top})$ and $\mbC_k=\E(\bx^k y)$, and assemble $\Sigma=(\Sigma_{kl})_{k,l=0}^K$ and $\mbC=(\mbC_0^\top,\ldots,\mbC_K^\top)^\top$. If $\Sigma$ is nonsingular, the ideal coefficient satisfies
    \begin{equation}
    \label{true estimator}
        \bbeta^*=\Sigma^{-1}\mbC.
    \end{equation}
We take $\bbeta^*$ as the statistical estimand. In the high-dimensional setting, we assume that it is sparse and write $S=\operatorname{supp}(\bbeta^*)$ and $s=\norm{\bbeta^*}_0$.

In the observed block-missing data, neither the complete vector $\bx$ nor its empirical covariance can be formed. Following DISCOM \citep{yu2020optimal}, we therefore estimate each population moment using all IDs on which its variables are jointly observed. The within-user and cross-user covariance estimators are $\tSigma_k=n_k^{-1}\bX^{k\top}\bX^k$ and $\tSigma_{kl}=n_{kl}^{-1}\bX^{k,l\top}\bX^{l,k}$ for $k\neq l$, while $\tmbC_k=(n_k^y)^{-1}\bX^{k\top}\by_k$ estimates the response-covariate block. Set $\tSigma_{kk}=\tSigma_k$, $\tSigma=(\tSigma_{kl})_{k,l=0}^K$, and $\tmbC=(\tmbC_0^\top,\ldots,\tmbC_K^\top)^\top$.

Because these blocks are estimated from different samples, $\tSigma$ may not be positive semidefinite. Write $\tSigma_I=\operatorname{diag}(\tSigma_0,\ldots,\tSigma_K)$ for the within-user blocks and $\tSigma_C=\tSigma-\tSigma_I$ for the cross-user blocks. To obtain a positive-definite matrix, DISCOM replaces $\tSigma$ by $\hSigma=a_1\tSigma_I+a_2\tSigma_C+(1-a_1)c_{tr}\mathrm I_p/p$, where $c_{tr}=\sum_{k=0}^K\operatorname{tr}(\tSigma_k)$ and $a_1,a_2$ are chosen so that $\hSigma$ is positive definite. Because $\bbeta^*$ is sparse and $p$ may exceed $n$, we add an $\ell_1$ penalty and, for $\lambda\geq0$, define
    \begin{equation}
    \label{Qn}
    \begin{aligned}
        \cQ_n(\bbeta)
        &=\cL_n(\bbeta)+\lambda\norm{\bbeta}_1\\
        &=\frac12\bbeta^\top\hSigma\bbeta-\tmbC^\top\bbeta
          +\lambda\norm{\bbeta}_1,
    \end{aligned}
    \end{equation}
where $\cL_n(\bbeta)=\bbeta^\top\hSigma\bbeta/2-\tmbC^\top\bbeta$ is the smooth quadratic loss. 

Then the centralized available-case estimator is $\cbbeta=\argmin_{\bbeta\in\mathbb R^p}\cQ_n(\bbeta)$. If all available-case moments could be assembled centrally, $\cbbeta$ could be obtained by minimizing $\cQ_n$ directly. 

 In the decentralized setting considered here, however, these moments remain distributed across users. Our goal is therefore to recover the same estimator without centralizing them as described in the next subsection.
    
\subsection{Assisted Learning with Block-Missing Data}
\label{subsec: ALB}
To recover $\cbbeta$ without assembling the cross-user moments, ALB replaces each $p_k\times p_l$ cross-user matrix with an $n_{kl}$-dimensional vector of sample-level linear predictions. These summaries permit cyclic block minimization of $\cQ_n$ without assembling any cross-user covariance matrix. Under the conditions of Section \ref{subsec: algorithmic convergence}, the resulting iterates converge geometrically to $\cbbeta$ while communicating $O(n)$ scalars per complete update cycle for a fixed number of {users}, regardless of $p$. 

Specifically, the objective admits the following block decomposition:
    \begin{equation}
    \label{eq:decomposed-discom-loss}
    \begin{aligned}
        \cQ_n(\bbeta)
        &=\frac{a_1}{2}\sum_{k=0}^K\bbeta_k^\top\tSigma_k\bbeta_k
          +\frac{(1-a_1)c_{tr}}{2p}\sum_{k=0}^K\norm{\bbeta_k}_2^2\\
        &\quad+\frac{a_2}{2}\sum_{k\neq l}
          \frac{\bbeta_k^\top\bX^{k,l\top}\bX^{l,k}\bbeta_l}{n_{kl}}
          -\sum_{k=0}^K\frac{\by_k^\top\bX^k\bbeta_k}{n_k^y}
          +\lambda\sum_{k=0}^K\norm{\bbeta_k}_1.
    \end{aligned}
    \end{equation}
To update $\bbeta_k$, user $B_k$ acts as the receiver. For every $l\neq k$, sender $B_l$ computes $\bu_{l\to k}(\bbeta_l)=\bX^{l,k}\bbeta_l\in\mathbb R^{n_{kl}}$ and sends it to $B_k$. The receiver combines this vector with its overlapping rows $\bX^{k,l}$ to evaluate $\bbeta_l^\top\bX^{l,k\top}\bX^{k,l}\bbeta_k
=\bu_{l\to k}(\bbeta_l)^\top\bX^{k,l}\bbeta_k$. When $B_l$ is updated, the roles reverse and $B_k$ sends $\bu_{k\to l}(\bbeta_k)$ to $B_l$. Hence, neither user transmits raw covariates or the $p_k\times p_l$ cross-covariance matrix. Communication for this pair consists only of an $n_{kl}$-vector in each required direction. The complete procedure is summarized in Algorithm \ref{alg:optimization}.
\begin{algorithm}[!t]
\caption{ALB: Cyclic Block Optimization for $K+1$ Users}
\label{alg:optimization}

{\renewcommand{\baselinestretch}{0.90}\selectfont
\begin{algorithmic}[1]
\Statex \textbf{Input:}
$(a_1,a_2,\lambda)$, maximum sweeps $e_{\max}$,
loss tolerance $\tau_{\mathrm{ALB}}$, and initial estimates
$\{\hbbeta_k^{(0)}\}_{k=0}^K$.
\Statex \textbf{Output:}
$\hbbeta=(\hbbeta_0^\top,\ldots,\hbbeta_K^\top)^\top$.
\State
$e\gets0$ and
$q_{\mathrm{prev}}\gets\cQ_n(\hbbeta^{(0)})$.
\While{$e<e_{\max}$}
    \For{$k=1,\ldots,K$}
        \State $B_0$ sends
        $\bu_{0\to k}(\hbbeta_0^{(e)})$ to $B_k$.
        \State $B_l$ sends
        $\bu_{l\to k}(\hbbeta_l^{(e+1)})$ for $1\leq l<k$
        and $\bu_{l\to k}(\hbbeta_l^{(e)})$ for
        $k<l\leq K$ to $B_k$.
        \State $B_k$ minimizes
        \eqref{distributed Qn k not 0} to obtain
        $\hbbeta_k^{(e+1)}$.
    \EndFor
    \State Each $B_k$ sends
    $\bu_{k\to0}(\hbbeta_k^{(e+1)})$ to $B_0$,
    $k=1,\ldots,K$.
    \State $B_0$ minimizes
    \eqref{distributed Qn k 0} to obtain
    $\hbbeta_0^{(e+1)}$.
    \State
    $e\gets e+1$,
    $q_e\gets\cQ_n(\hbbeta^{(e)})$, and
    $\delta_e\gets
    |q_e-q_{\mathrm{prev}}|/
    \max\{1,|q_{\mathrm{prev}}|\}$.
    \If{$\delta_e\leq\tau_{\mathrm{ALB}}$}
        \State \textbf{break}
    \EndIf
    \State $q_{\mathrm{prev}}\gets q_e$.
\EndWhile
\State \Return
$\hbbeta=
(\hbbeta_0^{(e)\top},\ldots,\hbbeta_K^{(e)\top})^\top$.
\end{algorithmic}
}
\end{algorithm}

In this procedure, ALB updates one coefficient block at a time. We call one ordered pass in which every user updates its block once a \emph{sweep}. Let $\hbbeta_k^{(e)}$ denote the estimate held by $B_k$ after $e$ complete sweeps. During sweep $e+1$, users $B_1,\ldots,B_K$ update sequentially, followed by $B_0$. Thus, when $B_k$, $1\leq k\leq K$, is updated, blocks $B_1,\ldots,B_{k-1}$ already have their current-sweep values $\hbbeta_l^{(e+1)}$, whereas $B_{k},\ldots,B_K$ and $B_0$ retain their previous-sweep values $\hbbeta_l^{(e)}$ and $\hbbeta_0^{(e)}$, respectively. Accordingly, $B_k$ minimizes
    \begin{equation}
    \label{distributed Qn k not 0}
    \begin{aligned}
        &\cQ_{n,k}\!\left(\bbeta_k;
        \hbbeta_0^{(e)},\{\hbbeta_l^{(e+1)}\}_{1\leq l<k},
        \{\hbbeta_l^{(e)}\}_{k<l\leq K}\right)\\
        &\quad=\frac{a_1}{2}\bbeta_k^\top\tSigma_k\bbeta_k
        +\frac{(1-a_1)c_{tr}}{2p}\norm{\bbeta_k}_2^2
        +a_2\sum_{1\leq l<k}\frac{1}{n_{lk}}
          \left(\hbbeta_l^{(e+1)\top}
          \bX^{l,k\top}\bX^{k,l}\bbeta_k\right)\\
        &\qquad
        +a_2\sum_{\{l:k<l\leq K\}\cup\{0\}}
         \frac{1}{n_{lk}} \left(\hbbeta_l^{(e)\top}
          \bX^{l,k\top}\bX^{k,l}\bbeta_k\right)
        -\frac{1}{n_k^y}\by_k^\top\bX^k\bbeta_k
        +\lambda\norm{\bbeta_k}_1.
    \end{aligned}
    \end{equation}
Once users $B_1,\ldots,B_K$ have been updated, $B_0$ uses all $\{\hbbeta_k^{(e+1)}\}_{k=1}^K$ and minimizes
\begin{equation}
\label{distributed Qn k 0}
\begin{aligned}
    \cQ_{n,0}\!\left(\bbeta_0;\{\hbbeta_k^{(e+1)}\}_{k=1}^K\right)
    &=\frac{a_1}{2}\bbeta_0^\top\tSigma_0\bbeta_0
      +\frac{(1-a_1)c_{tr}}{2p}\norm{\bbeta_0}_2^2\\
    &\quad+ a_2\sum_{k=1}^K
      \frac{1}{n_{k0}}\left(\hbbeta_k^{(e+1)\top}
      \bX^{k,0\top}\bX^{0,k}\bbeta_0\right)
      -\frac{\by_0^\top\bX^0\bbeta_0}{n_0^y}
      +\lambda\norm{\bbeta_0}_1.
\end{aligned}
\end{equation}
Here, the label $B_0$ fixes a reference point in the cyclic update order rather than a distinct statistical role. User $B_0$ holds its own data block, has no access to the other users' raw data, and is not a separate coordinating server. Thus, any user can be designated as $B_0$ without changing the objective or the theoretical results.

Importantly, Equations \eqref{distributed Qn k not 0} and \eqref{distributed Qn k 0} show that the exchanged summaries evaluate the original cross terms exactly rather than define a surrogate objective. Because each cross term is evaluated through an $n_{kl}$-dimensional summary instead of a $p_k\times p_l$ matrix, one complete sweep communicates $O(\sum_{k<l}n_{kl})=O(K^2n)$ scalars, or $O(n)$ for fixed $K$, regardless of $p$.

In implementation, during each sweep, every user solves its block subproblem by coordinate descent as detailed in Section \ref{app: proximal gradient descent} in the Supplement. After each complete sweep, ALB evaluates the relative change in the penalized loss and stops when this change is below a prescribed tolerance. Otherwise, it proceeds to the next sweep. \textcolor{blue}{The global objective value used in the stopping rule is obtained by
summing user‑local scalar contributions and scalar pairwise
cross‑term contributions. No covariance block or raw feature matrix is centralized for this calculation.} The resulting distributed iterates converge geometrically to the centralized estimator $\cbbeta$, as established in Section \ref{subsec: algorithmic convergence}.

\subsection{Distributed Debiasing and Inference}
\label{subsec: debiased estimator}

We next construct coordinatewise inference for the sparse estimator. The $\ell_1$ penalty introduces coordinatewise shrinkage bias, so direct inference based on $\hbbeta$ is generally invalid. To remove this first-order bias, we combine an unpenalized available-case score with an inverse-covariance \textcolor{blue}{column} for the target coordinate. We begin with the latter. For $j\in\{1,\ldots,p\}$, let $\be_j\in\mathbb R^p$ be the $j$th canonical basis vector and define the population precision column $\bthe_j^*=\Sigma^{-1}\be_j$. Because $\Sigma$ is unknown, we estimate this column by
    \begin{equation}
    \label{global inverse column}
        \cbthe_j
        \coloneqq\argmin_{\bthe\in\mathbb R^p}
        \left\{\frac12\bthe^\top \hSigma^{(j)} \bthe
        -\be_j^\top\bthe+\lambda_j\norm{\bthe}_1\right\},
    \end{equation}
where $\hSigma^{(j)}=a_{1,j}\tSigma_I+a_{2,j}\tSigma_C+(1-a_{1,j})c_{tr}\mathrm I_p/p$ and $\lambda_j\geq0$ is the precision-column penalty. Unlike the point-estimation matrix $\hSigma$, $\hSigma^{(j)}$ may use weights $(a_{1,j},a_{2,j})$ selected specifically for precision-column estimation. To avoid coordinate-by-coordinate tuning, Section \ref{app: precision covariance selection} in the Supplement describes an efficient strategy that shares one weight pair across precision columns. Since this objective has the same quadratic block structure as \eqref{Qn}, Algorithm \ref{alg:optimization} applies after replacing $(a_1,a_2)$ with $(a_{1,j},a_{2,j})$, without changing the communication scheme. Let $\hbthe_j$ denote its distributed solution. Theorem \ref{the: algorithmic convergence} in the next section establishes that $\cbthe_j$ is unique and $\hbthe_j$ converges to $\cbthe_j$. Thus, the same decentralized machinery yields the required precision \textcolor{blue}{column}.

With the precision \textcolor{blue}{column} in hand, we next construct the score. The population normal equation $\Sigma\bbeta^*-\mbC=\mathbf0$ suggests the unpenalized available-case score $S_n(\bbeta)=\tSigma\bbeta-\tmbC$. Combining this score with $\hbthe_j$ gives the debiased estimator
    \begin{equation}
    \label{distributed individual debiased estimator}
        \hb_j \coloneqq \hbeta_j - \hbthe_j^\top S_n(\hbbeta).
    \end{equation}

To see how this correction removes the first-order bias, note that the score satisfies the affine identity
    \begin{equation}
    \label{expansion of score function}
    S_n(\hbbeta) = S_n(\bbeta^*) + \tSigma (\hbbeta-\bbeta^*).
    \end{equation}
Substituting \eqref{expansion of score function} into \eqref{distributed individual debiased estimator} gives
    \begin{equation}
    \begin{aligned}
    \label{debiased error decomposition}
        \hb_j -\beta^*_j &= -\bthe_j^{*\top} S_n(\bbeta^*) - \underbrace{(\hbthe_j-\bthe_j^*)^\top S_n(\bbeta^*)}_{\Delta_1} \\
        &- \underbrace{\hbthe_j^\top\left( \tSigma-\hSigma^{(j)} \right)(\hbbeta-\bbeta^*)}_{\Delta_2} - \underbrace{\left( \hSigma^{(j)} \hbthe_j - \be_j \right)^\top (\hbbeta-\bbeta^*)}_{\Delta_3}. 
    \end{aligned}
    \end{equation}
The first term is the leading stochastic term. The three remainders arise, respectively, from estimating $\bthe_j^*$, using different covariance matrices in the precision-column loss and the score, and approximating the inverse equation. Section \ref{subsec: oracle property} shows that $\sqrt n\Delta_r=o_p(1)$ for $r=1,2,3$, so the leading term determines the limiting distribution.

Based on the leading term, we derive its variance by expressing the score as a sum of sample-ID-level contributions. Recall the observation indicators $r_i^k$ and $r_i^y$. Write $R_i=(r_i^y,r_i^0,\ldots,r_i^K)$ and $R=(R_1,\ldots,R_n)$ for the missingness pattern, and adopt $n_{kk}=n_k$. Define the $p_k\times p_l$ matrix $A_{i,kl}$ and the $p_k$-vector $q_{i,k}$ by
    \begin{equation}
    \label{population influence contribution}
    \begin{aligned}
        A_{i,kl}\coloneqq\frac{n}{n_{kl}}r_i^kr_i^l(\bx_i^k\bx_i^{l\top}-\Sigma_{kl}),
        \qquad
        q_{i,k}\coloneqq\frac{n}{n_k^y}r_i^kr_i^y(y_i\bx_i^k-\mbC_k),
    \end{aligned}
    \end{equation}
Assemble $A_i=(A_{i,kl})_{k,l=0}^K$ and $q_i=(q_{i,0}^\top,\ldots,q_{i,K}^\top)^\top$, and set $\bphi_i=A_i\bbeta^*-q_i$. Then $S_n(\bbeta^*)=(\tSigma-\Sigma)\bbeta^*-(\tmbC-\mbC)=n^{-1}\sum_{i=1}^n\bphi_i$. Consequently, the conditional variance of the leading term is $v_{n,j}=\bthe_j^{*\top}\Omega_n\bthe_j^*$, where $\Omega_n=n\Var\{S_n(\bbeta^*)\mid R\}$.
    
Since this variance depends on unknown population quantities, we replace them in $\bphi_i$ with sample analogues. Specifically, define $\hA_{i,kl}=\frac{n}{n_{kl}}r_i^kr_i^l(\bx_i^k\bx_i^{l\top}-\tSigma_{kl})$ and $\hq_{i,k}=\frac{n}{n_k^y}r_i^kr_i^y(y_i\bx_i^k-\tmbC_k)$. Assemble $\hA_i=(\hA_{i,kl})_{k,l=0}^K$ and $\hq_i=(\hq_{i,0}^\top,\ldots,\hq_{i,K}^\top)^\top$, and set $\hbphi_i=\hA_i\hbbeta-\hq_i$. The variance estimator is
    \begin{equation}
    \label{estimator of variance}
    \hv_{n,j} \coloneqq \frac{1}{n-1}\sum_{i=1}^n\left(\hbthe_j^\top (\hbphi_i-\overline{\hbphi})\right)^2,
    \end{equation}
where $\overline{\hbphi}=\sum_{i=1}^n \hbphi_i/n$. Contributions associated with the same sample ID are aggregated before centering and squaring, so that the variance estimator accounts for covariance induced by overlapping samples.

Finally, combining the debiased estimator with this variance estimator yields the following asymptotic $(1-\alpha)$ confidence interval for $0<\alpha<1$:
    \begin{equation}
    \label{distributed CI}
        \left[\hb_j-\Phi^{-1}(1-\alpha/2)\sqrt{\hv_{n,j}/n},\;
        \hb_j+\Phi^{-1}(1-\alpha/2)\sqrt{\hv_{n,j}/n}\right],
    \end{equation}
    where $\Phi(\cdot)$ is the standard normal cumulative distribution function.
    
The preceding construction uses global notation only for exposition. Its block structure permits decentralized computation of both $\hb_j$ and $\hv_{n,j}$ without centralizing $\hbphi_i$. Section \ref{app: pseudocode} in the Supplement gives the full procedure. Again, for fixed $K$, each precision-column sweep and the final aggregation communicate $O(n)$ scalars independently of $p$.

\section{Theoretical Properties}
\label{sec: theory}
This section separates the performance of ALB into statistical, optimization, and inference components. Section \ref{subsec: estimation theory} first establishes concentration of the available-case moments and the resulting estimation rate for the centralized target $\cbbeta$. Section \ref{subsec: algorithmic convergence} then controls the distance between the distributed iterate $\hbbeta^{(e)}$ and $\cbbeta$. Finally, Section \ref{subsec: oracle property} combines these results with a conditional central limit theorem to justify coordinatewise inference.

\subsection{Moment Concentration and Estimation Error}
\label{subsec: estimation theory}
Recall that $\cbbeta$ is the centralized available-case estimator defined in Section \ref{subsec: problem set up}. The total estimation error admits the decomposition
   \[
       \hbbeta^{(e)}-\bbeta^*
       =\{\hbbeta^{(e)}-\cbbeta\}+\{\cbbeta-\bbeta^*\},
   \]
where the first term is optimization error after $e$ ALB sweeps and the second is statistical error. This subsection controls the second term by first establishing concentration of the available-case moments and then applying a sparse-regression argument. 

We begin with conditions. Condition \ref{con: ratios} ensures that every moment required by the procedure has an effective sample size proportional to $n$, while Condition \ref{con: missing pattern} makes the corresponding available-case moments unbiased. \textcolor{red}{Throughout, Condition~\ref{con: missing pattern} is assumed, and results are conditional on realizations of $R$ satisfying Condition~\ref{con: ratios}, unless stated otherwise.}
    
    \begin{assumption}
    \label{con: ratios}
    Let $n^y=\sum_{i=1}^nr_i^y$ and $n_{kl}^y=\sum_{i=1}^nr_i^yr_i^kr_i^l$. Assume that, as $n\to\infty$,
    \[
        n_k/n\to\rho_k,\quad
        n_{kl}/n\to\rho_{kl},\quad
        n_k^y/n\to\rho_k^y,\quad
        n^y/n\to\rho^y,\quad
        n_{kl}^y/n\to\rho_{kl}^y,
    \]
    where all limits belong to $(0,1]$ and statements involving $(k,l)$ apply to $k\neq l$.
    \end{assumption}

    \begin{assumption}
    \label{con: missing pattern}
        Assume that $\{(\bx_i,y_i,R_i)\}_{i=1}^n$ are independent and identically distributed and that $R_i$ is independent of $(\bx_i,y_i)$. Dependence among the components of $R_i$ is allowed.
    \end{assumption}
        
    \begin{remark}
    \label{remark: missing pattern}
        The above conditions constitute the missing completely at random (MCAR) mechanism. MCAR makes each available-case moment unbiased, and the positive limiting proportions give every required moment an effective sample size of order $n$. Complete observations are unnecessary. The full intersection may be empty even when all pairwise overlaps grow proportionally with $n$.
    \end{remark}

    The next three conditions concern covariance geometry, tail behavior, and sparse identifiability.

    \begin{assumption}
    \label{con: PD}
         Assume that the population covariance matrix $\Sigma$ is positive definite. Choose $(a_1,a_2)$ so that $\hSigma$ is positive definite and, for inference on coordinate $j$, choose $(a_{1,j},a_{2,j})$ so that $\hSigma^{(j)}$ is positive definite. Assume also that the covariates are standardized: $\sigma_{jj}=\E(X_j^2)=1$ and $\tilde\sigma_{jj}=1$ for $j=1,\ldots,p$.
    \end{assumption}
    \begin{remark}
    \label{remark: PD}
        These positive-definiteness requirements have distinct roles: $\Sigma$ identifies the population target, $\hSigma$ makes $\cQ_n$ strongly convex, and $\hSigma^{(j)}$ does the same for the precision-column objective in \eqref{global inverse column}. Standardization can be achieved by rescaling and gives $c_{tr}/p=1$. Sections \ref{app: covariance regularization}--\ref{app: shifted Lanczos iteration} in the Supplement describe practical procedures for selecting $(a_1,a_2)$ and $(a_{1,j},a_{2,j})$ under the requirements of Condition \ref{con: PD}, including distributed estimation of the smallest eigenvalues used to screen candidate pairs.
    \end{remark}

    \begin{assumption}
    \label{con: subgaussian}
        \textcolor{red}{Assume that the covariates and response are centered, so that $\E(X_j)=0$ and $\E(Y)=0$, and have uniformly sub-Gaussian tails: there exists a constant $L>0$, independent of $n$ and $p$, such that}
        \begin{equation*}
        \begin{aligned}
            &\E \exp\{tX_j\} \leq \exp\left(\frac{L^2 t^2}{2}\right),
            &&j=1,\ldots,p,\quad t\in\mathbb R,\\
            &\E \exp\{tY\}\leq \exp\left(\frac{L^2 t^2}{2}\right),
            &&t\in\mathbb R.
        \end{aligned}
        \end{equation*}
    \end{assumption}

    \begin{assumption}
    \label{con: RE}
        
        Let $S\coloneqq\operatorname{supp}(\bbeta^*)$. Assume that there exists a constant $m>0$, independent of $n$ and $p$, such that
        $$
            \bdelta^\top\Sigma\bdelta\geq m\norm{\bdelta}_2^2
        $$
        for every $\bdelta\in\mathbb R^p\setminus\{\mathbf0\}$ satisfying $\norm{\bdelta_{S^c}}_1\leq7\norm{\bdelta_S}_1$.
        
    \end{assumption}

    With these conditions in place, we first quantify the estimation error of the moment blocks. Adopt the convention $n_{kk}=n_k$ and define the smallest covariance, response, and overall effective sample sizes by $N_{\Sigma,n}\coloneqq\min_{0\leq k,l\leq K}n_{kl}$, $N_{y,n}\coloneqq\min_{0\leq k\leq K}n_k^y$, and $N_n\coloneqq N_{\Sigma,n}\wedge N_{y,n}$. For $j,t=1,\ldots,p$, define
    $\tilde\sigma_{jt}\coloneqq(\tSigma)_{jt},$
     $\sigma_{jt}\coloneqq(\Sigma)_{jt}.$
    The next two propositions extend Theorems 1 and 2 of \citet{yu2020optimal}, respectively, to allow missing responses.

    \begin{proposition}
    \label{the: sample cov}
         Suppose Conditions \ref{con: missing pattern}--\ref{con: subgaussian} hold. Let $1-a_1=\textcolor{blue}{O\{\sqrt{(\log p)/(\min_k n_k)}\}}$, $1-a_2=\textcolor{blue}{O\{\sqrt{(\log p)/(\min_{k\neq l}n_{kl})}\}}$. If $\min_{j,t}N_{jt}^{xx}\geq6\log p$, then, for $v_1=8\sqrt6(1+4L^2)$,
         $$
          \max_{j,t}\bbP \left( \left.\abs{\tilde{\sigma}_{jt} - \sigma_{jt}} \geq v_1 \sqrt{\frac{\log p}{N_{jt}^{xx}}}\,\right|R \right) \leq \frac{4}{p^3},
         $$
         $$
          \bbP \left( \left.\norm{\tSigma - \Sigma}_{\max} \geq v_1 \sqrt{\frac{\log p}{\min_{j,t} N_{jt}^{xx}}}\,\right|R \right) \leq \frac{4}{p}.
         $$
         Let $\tilde c_j$ and $c_j$ denote the $j$th entries of $\tmbC$ and $\mbC$. If $\min_jN_j^{xy}\geq4\log p$, then, for \textcolor{blue}{$v_2=16\left(1+(4L^2)/{\min\{\Var(Y),1\}}\right) \max\{\Var(Y),1\}$},
         $$
          \max_j\bbP \left( \left.\abs{\tilde c_j-c_j} \geq v_2 \sqrt{\frac{\log p}{N_j^{xy}}}\,\right|R \right) \leq \frac{4}{p^2},
         $$
         $$
          \bbP \left( \left.\norm{\tmbC - \mbC}_{\infty} \geq v_2 \sqrt{\frac{\log p}{\min_j N_j^{xy}}}\,\right|R \right) \leq \frac{4}{p}.
         $$
         Furthermore, there exists a constant $v_3>0$ such that
         $$
          \bbP\left(\left.\norm{\hSigma - \Sigma}_{\max} \geq v_3 \sqrt{\frac{\log p}{\min_{j,t}N_{jt}^{xx}}}\,\right|R\right) \leq \frac{4}{p}.
         $$
         If $1-a_{1,j}=O\{\sqrt{(\log p)/(\min_k n_k)}\}$ and $1-a_{2,j}=O\{\sqrt{(\log p)/(\min_{k\neq l}n_{kl})}\}$, the same bound holds with $\hSigma$ replaced by $\hSigma^{(j)}$.
    \end{proposition}

    Each bound uses the sample size available for the corresponding moment, so the smallest required overlap controls the uniform error. The stated orders of $1-a_1$ and $1-a_2$ also keep the positive-definiteness adjustment within the sampling error.
    
    \begin{proposition}
    \label{the: global consistency}
        Suppose Conditions \ref{con: missing pattern}--\ref{con: RE} hold. Let $a_1$ and $a_2$ satisfy the rates in Proposition \ref{the: sample cov}. If $s\sqrt{\frac{\log p}{N_{\Sigma,n}}}=o(1)$ and $\lambda=2\norm{\tmbC-\hSigma\bbeta^*}_{\infty}$, then
        \begin{align*}
            \norm{\cbbeta-\bbeta^*}_2
            =O_p\left(\norm{\bbeta^*}_1\sqrt{\frac{s\log p}{N_n}}\right),\quad
            \norm{\cbbeta-\bbeta^*}_1
            =O_p\left(s\norm{\bbeta^*}_1\sqrt{\frac{\log p}{N_n}}\right).
        \end{align*}
    \end{proposition}

    Proposition \ref{the: global consistency} translates the moment concentration in Proposition \ref{the: sample cov} into estimation rates. Since the objective involves both covariance and response-covariate moments, its effective sample size is the smaller of the two, $N_n=N_{\Sigma,n}\wedge N_{y,n}$.

\subsection{Algorithmic Convergence}
\label{subsec: algorithmic convergence}

The preceding result controls the statistical error $\cbbeta-\bbeta^*$. We next bound the optimization error $\hbbeta^{(e)}-\cbbeta$ generated by the distributed updates.


    \begin{theorem}
    \label{the: algorithmic convergence}
        Suppose Condition \ref{con: PD} holds and every block subproblem in Algorithm \ref{alg:optimization} is minimized exactly. Let $\hbbeta^{(e)}=(\hbbeta_0^{(e)\top},\ldots,\hbbeta_K^{(e)\top})^\top$, $\lambda_n^{(min)}=\lambda_{\min}(\hSigma)$, $\lambda_n^{(max)}=\lambda_{\max}(\hSigma)$. Define
        \begin{align*}
            q_n^{(1)}&=\left\{1-
            \frac{(\lambda_n^{(min)})^2}
            {(\lambda_n^{(min)})^2+(K+1)(\lambda_n^{(max)})^2}
            \right\}^{1/2}\in(0,1),\\
            c_n&=\cQ_n(\hbbeta^{(0)})-\cQ_n(\cbbeta),\qquad
            q_n^{(2)}=\sqrt{\frac{2c_n}{\lambda_n^{(min)}}}\in[0,\infty).
        \end{align*}
        Then the following statements hold.

        \begin{enumerate}
        \item[(i)] For every $e\geq0$,
        \begin{equation}
        \label{algorithmic l2 bound}
            \norm{\hbbeta^{(e)}-\cbbeta}_2 \leq q_n^{(2)} \cdot (q_n^{(1)})^e,
        \end{equation}

        \item[(ii)] If the conditions of Proposition \ref{the: global consistency} also hold, then
        $$
            \norm{\hbbeta^{(e)}-\bbeta^*}_2
            \leq q_n^{(2)}(q_n^{(1)})^e
            +O_p\left(\norm{\bbeta^*}_1\sqrt{\frac{s\log p}{N_n}}\right).
        $$
    
        \item[(iii)] If $q_n^{(2)}=0$, then $\hbbeta^{(0)}=\cbbeta$. Otherwise, define $\lceil x\rceil_+=\max\{0,\lceil x\rceil\}$, where $\lceil x\rceil$ is the smallest integer not less than $x$. If
        $$
            e\geq
            \left\lceil
            \frac{\log\left\{q_n^{(2)}\sqrt{np/(\log p)}\right\}}
            {\log\{1/q_n^{(1)}\}}
            \right\rceil_+,
        $$
        then
        $$
            \norm{\hbbeta^{(e)}-\cbbeta}_1
            \leq\sqrt{\frac{\log p}{n}}.
        $$
        \end{enumerate}
    \end{theorem}

    Part (i) establishes geometric convergence per sweep, while part (ii) shows that ALB attains the centralized statistical rate once the optimization term is negligible. Part (iii) gives an explicit sweep count that reduces the $\ell_1$ optimization error to the level required for debiasing in Section \ref{subsec: oracle property}. Because the argument depends only on the quadratic term $\hSigma$, it also applies to the precision-column problem \eqref{global inverse column}. Specifically, replacing $\hSigma$ by $\hSigma^{(j)}$ and redefining the eigenvalues, objective gap, and minimizer accordingly yields the same three conclusions. Thus, both distributed iterates required for inference can be controlled by the same convergence result.

\subsection{Asymptotic Inference}
\label{subsec: oracle property}
    We now combine the statistical and optimization bounds with the debiasing decomposition in \eqref{debiased error decomposition}. The argument has three parts: a conditional central limit theorem for the leading score term, control of the three debiasing remainders, and consistency of the plug-in variance estimator. We consider a fixed coordinate $j$ and a fixed number of users $K+1$ as $n,p\to\infty$. The following conditions are considered.

    \begin{assumption}
    \label{con: conditional variance of y}
        Assume that there exists a constant $c>0$, independent of $n$ and $p$, such that $\Var(Y\mid\bx)\geq c$ almost surely.
    \end{assumption}

    \begin{assumption}
    \label{con: overlap covariance}
        Let $\mathbf P_n\in\mathbb R^{p\times p}$ be the block matrix whose $(k,l)$ block is $(\mathbf P_n)_{kl}
            =(n n_{kl}^y)/(n_k^yn_l^y)
            \mathbf1_{p_k}\mathbf1_{p_l}^\top$, $k,l=0,\ldots,K$, with $n_{kk}^y=n_k^y$. Assume that there exists a constant $\kappa_P>0$, independent of $n$ and $p$, such that, for all sufficiently large $n$, $\lambda_{\min}(\mathbf P_n\circ\Sigma)\geq\kappa_P.$ Here $\circ$ denotes the Hadamard product.
    \end{assumption}

    \begin{remark}
    \label{remark: conditional variance and overlap covariance}
        Condition \ref{con: conditional variance of y} allows misspecification and heteroscedasticity while excluding a degenerate score variance. For Condition \ref{con: overlap covariance}, setting the $k$th block of $\bd_i$ to $nr_i^yr_i^k\mathbf1_{p_k}/n_k^y$ gives $\mathbf P_n=n^{-1}\sum_i\bd_i\bd_i^\top$. The Schur product theorem then makes $\mathbf P_n\circ\Sigma$ positive semidefinite, and the lower eigenvalue bound ensures that the overlap structure does not produce a degenerate score direction, paralleling the variance condition in \citet{xue2025statistical}.
    \end{remark}
    
    \begin{lemma}
    \label{le: normality}
    Suppose Conditions \ref{con: ratios}--\ref{con: subgaussian} and \ref{con: conditional variance of y}--\ref{con: overlap covariance} hold. {Let $B_{\beta,n}\coloneqq1+\norm{\bbeta^*}_1$ and assume}
    \begin{equation}
    \label{coordinate Lyapunov rate}
        {\norm{\bthe_j^*}_1^4B_{\beta,n}^4/
        \left(n\norm{\bthe_j^*}_2^4\right)\longrightarrow0.}
    \end{equation}
    Then, $v_{n,j}\geq c\kappa_P$ for all sufficiently large $n$ and
    $$
        \frac{\sqrt n\,\bthe_j^{*\top}S_n(\bbeta^*)}{\sqrt{v_{n,j}}}
        \xrightarrow{d}\mathcal N(0,1).
    $$
    \end{lemma}

 
    \begin{remark}
    \label{remark: normality}
        Positive response and pairwise-overlap proportions yield the $\sqrt n$ normalization. In contrast, Theorem 2 of \citet{zhang2026additive} is indexed by the number of complete observations $n_{\mathrm{comp}}$. Our framework retains the $\sqrt n$ rate even when $n_{\mathrm{comp}}=0$, provided all required overlaps remain proportional to $n$.
    \end{remark}
    
    Lemma \ref{le: normality} handles the leading score term. We next combine it with the estimation and optimization rates from Sections \ref{subsec: estimation theory} and \ref{subsec: algorithmic convergence} to control the remainders and the plug-in variance estimator. Recall $N_{\Sigma,n}$, $N_{y,n}$, and $N_n$ from Section \ref{subsec: estimation theory}, and let $s_j\coloneqq\norm{\bthe_j^*}_0$. Define
    \begin{align*}
    r_{\beta,n}&=s\norm{\bbeta^*}_1\sqrt{\frac{\log p}{N_n}}+\sqrt{\frac{\log p}{n}}, &r_{\theta,n}&=s_j\norm{\bthe_j^*}_1\sqrt{\frac{\log p}{N_{\Sigma,n}}}+\sqrt{\frac{\log p}{n}},\\
    a_n&=\norm{\bbeta^*}_1\sqrt{\frac{\log p}{N_{\Sigma,n}}}
          +\sqrt{\frac{\log p}{N_{y,n}}}, &\mu_{j,n}&=\norm{\bthe_j^*}_1\sqrt{\frac{\log p}{N_{\Sigma,n}}}+\sqrt{\frac{\log p}{n}},
    \qquad h_n=\sqrt{\frac{\log p}{N_{\Sigma,n}}}.
    \end{align*}
    Also set $T_{\theta,n}\coloneqq\norm{\bthe_j^*}_1+r_{\theta,n}$.
    The rates $r_{\beta,n}$ and $r_{\theta,n}$ combine the statistical $\ell_1$ errors and the optimization tolerance for the regression and precision-column estimators, respectively. The remaining quantities $a_n$, $\mu_{j,n}$, and $h_n$ control the score error, inverse-equation residual, and covariance adjustment, respectively. 

    To obtain the inference result, let $\hbbeta$ and $\hbthe_j$ denote distributed iterates satisfying $\norm{\hbbeta-\cbbeta}_1\leq\sqrt{(\log p)/n}$ and $\norm{\hbthe_j-\cbthe_j}_1\leq\sqrt{(\log p)/n}$, respectively. Part (iii) of Theorem \ref{the: algorithmic convergence} gives a sufficient stopping rule for both objectives.

    \begin{theorem}
    \label{the: oracle property}
        Suppose Conditions \ref{con: ratios}--\ref{con: overlap covariance} hold, together with  the conditions in Propositions \ref{the: sample cov} and \ref{the: global consistency}. Let $S_j\coloneqq\operatorname{supp}(\bthe_j^*)$. Assume that Condition \ref{con: RE} also holds with $S$ replaced by $S_j$ and that $s_j\sqrt{\frac{\log p}{N_{\Sigma,n}}}=o(1)$, so that Proposition \ref{the: global consistency} applies to the inverse-column problem. Assume $1-a_{1,j}=O\{\sqrt{(\log p)/(\min_k n_k)}\}$ and $1-a_{2,j}=O\{\sqrt{(\log p)/(\min_{k\neq l}n_{kl})}\}$. Choose its theoretical penalty as $\lambda_j=2\norm{({\hSigma^{(j)}}-\Sigma)\bthe_j^*}_{\infty}
            =O_p\left(\norm{\bthe_j^*}_1
            \sqrt{\frac{\log p}{N_{\Sigma,n}}}\right)$. Assume the coordinate-specific Lyapunov condition \eqref{coordinate Lyapunov rate} and
        \begin{align}
        \label{oracle rate conditions}
        \sqrt n\,r_{\theta,n}a_n&\to0,&
        \sqrt n\,r_{\beta,n}(\norm{\bthe_j^*}_1+r_{\theta,n})h_n&\to0,&
        \sqrt n\,r_{\beta,n}\mu_{j,n}&\to0,
        \end{align}
        together with
        \begin{equation}
        \label{variance plugin rate conditions}
        r_{\theta,n}B_{\beta,n}\log p/\norm{\bthe_j^*}_2\to0,
        \quad T_{\theta,n}r_{\beta,n}\log p/\norm{\bthe_j^*}_2\to0,
        \quad T_{\theta,n}a_n/\norm{\bthe_j^*}_2\to0.
        \end{equation}
        Then, for the fixed coordinate $j$,
        \begin{align}
        \label{oracle asymptotic linearity}
            \sqrt n(\hb_j-\beta_j^*)
            &=-\frac1{\sqrt n}\sum_{i=1}^n\bthe_j^{*\top}\bphi_i+o_p(1),\\
        \label{oracle variance consistency}
            \hv_{n,j}/v_{n,j}&\xrightarrow{p}1,
        \end{align}
        and hence
        $$
        \frac{\sqrt{n}(\hb_j-\beta_j^*)}{\sqrt{\hv_{n,j}}}
        \xrightarrow{d}\mathcal{N}(0,1),
        $$
        where $\hv_{n,j}$ is defined in \eqref{estimator of variance}. 
            \end{theorem}

The three conditions in \eqref{oracle rate conditions} respectively make $\sqrt n\Delta_1$, $\sqrt n\Delta_2$, and $\sqrt n\Delta_3$ asymptotically negligible, which yields the linear representation in \eqref{oracle asymptotic linearity}. The three conditions in \eqref{variance plugin rate conditions} respectively control the effects of estimating $\bthe_j^*$, $\bbeta^*$, and the available-case moments on $\hv_{n,j}$. Their normalization by $\norm{\bthe_j^*}_2$ reflects the natural standard-deviation scale because \eqref{score variance lower bound} in the Supplement gives $\sqrt{v_{n,j}}\geq\sqrt{c\kappa_P}\norm{\bthe_j^*}_2$. When $\norm{\bbeta^*}_1$ and $\norm{\bthe_j^*}_1$ are bounded, a simple sufficient condition for \eqref{oracle rate conditions} and \eqref{variance plugin rate conditions} is $(s+s_j+1)(\log p)^{3/2}/\sqrt n\to0$, which is comparable to the sparsity and dimensionality condition of \citet{van2014asymptotically}. More generally, the stated conditions allow these norms to grow sufficiently slowly relative to the effective sample sizes.

\section{ALB with One-Time Perturbed Releases}
\label{sec: perturbed estimator}

Although ALB does not transmit raw covariates, its sample-level messages may still disclose individual values. An observed response may be shared with every user whose covariates are available for that ID, and repeated linear summaries may reveal exact covariate information. For example, if $\bx_i^l=0$, the corresponding entry of every transmitted linear combination remains zero across sweeps. We therefore introduce one-time perturbations for the transmitted responses and cross-user prediction summaries.  The within-user covariance $\tSigma_k$ remains unperturbed because user $B_k$ computes and retains it locally.

\subsection{Perturbation Mechanism and Objective}
\label{subsec: perturbation mechanism}

We perturb responses and covariates separately. For the responses, let $\mathcal R_r\subseteq\mathcal I^y$ be the set of IDs whose observed responses are initially held by user $B_r$, for $r=0,\ldots,K$. We assume that every observed response has exactly one initial holder. Thus, the sets $\mathcal R_0,\ldots,\mathcal R_K$ are disjoint and $\bigcup_{r=0}^K\mathcal R_r=\mathcal I^y$. For $i\in\mathcal R_r$, user $B_r$ releases $y_i^{\mathrm{priv}}=y_i+\xi_i^r$, \textcolor{red}{where $\xi_i^r$ is independent of the data and all other perturbations.} This response is perturbed once and the same release is used by every user that subsequently needs it. Let $\by_k^{\mathrm{priv}}\in\mathbb R^{n_k}$ align these releases with the rows of $\bX^k$, using zero when $i\notin\mathcal I_k^y$.

For the covariates, each user $B_k$ independently draws a random matrix $\Xi^k\in\mathbb R^{n_k\times p_k}$ once and sets $\bX_{\mathrm{priv}}^k=\bX^k+\Xi^k$, with rows ordered by $\mathcal I_k$. For $l\neq k$, let $\bX_{\mathrm{priv}}^{k,l}\in\mathbb R^{n_{kl}\times p_k}$ denote the rows indexed by the shared IDs $\mathcal I_{kl}$. When $B_k$ updates $\bbeta_k$, user $B_l$ sends $\bu_{l\to k}^{\mathrm{priv}}(\bbeta_l)=\bX_{\mathrm{priv}}^{l,k}\bbeta_l\in\mathbb R^{n_{kl}}$. The perturbation matrix $\Xi^l$ is sampled only once and reused whenever $B_l$ forms a cross-user message.

Together, these one-time releases induce perturbed cross-user and response-covariate moments while leaving local moments unchanged. Specifically, the local covariance remains $\tSigma_k=n_k^{-1}\bX^{k\top}\bX^k$. For $k\neq l$, define $\tSigma_{kl}^{\mathrm{priv}}=n_{kl}^{-1}\bX_{\mathrm{priv}}^{k,l\top}\bX_{\mathrm{priv}}^{l,k}$ and $\tmbC_k^{\mathrm{priv}}=(n_k^y)^{-1}\bX^{k\top}\by_k^{\mathrm{priv}}$, and stack the latter as $\tmbC^{\mathrm{priv}}$. Let $\tSigma_C^{\mathrm{priv}}$ collect the perturbed cross-user blocks. Then $\tSigma^{\mathrm{priv}}=\tSigma_I+\tSigma_C^{\mathrm{priv}}$ and $\hSigma^{\mathrm{priv}}=a_1\tSigma_I+a_2\tSigma_C^{\mathrm{priv}}+(1-a_1)(c_{tr}/p)\mathrm I_p$, where $c_{tr}=\sum_{k=0}^K\operatorname{tr}(\tSigma_k)$. Accordingly, replacing the transmitted moments in $\cQ_n$ gives the perturbed estimator
\begin{equation}
\label{private objective}
    \cbbeta^{\mathrm{priv}}
    =\argmin_{\bbeta\in\mathbb R^p}\cQ_n^{\mathrm{priv}}(\bbeta),
    \qquad
    \cQ_n^{\mathrm{priv}}(\bbeta)
    =\frac12\bbeta^\top\hSigma^{\mathrm{priv}}\bbeta
     -\tmbC^{\mathrm{priv}\top}\bbeta
     +\lambda\norm{\bbeta}_1.
\end{equation}
Algorithm \ref{alg:optimization} applies with the same communication order and optimizes $\cQ_n^{\mathrm{priv}}$ without transmitting raw cross-user records. For inference on coordinate $j$, we use $\hSigma^{\mathrm{priv},(j)}=a_{1,j}^{\mathrm{priv}}\tSigma_I+a_{2,j}^{\mathrm{priv}}\tSigma_C^{\mathrm{priv}}+(1-a_{1,j}^{\mathrm{priv}})(c_{tr}/p)\mathrm I_p$ to estimate the precision column and debias with $S_n^{\mathrm{priv}}(\bbeta)=\tSigma^{\mathrm{priv}}\bbeta-\tmbC^{\mathrm{priv}}$.

\subsection{Statistical Properties}
\label{subsec: privacy utility}

We now quantify the statistical cost of perturbation. The analysis allows general sub-Gaussian perturbations, with the Gaussian perturbations used in the numerical studies as a special case.

\begin{assumption}
\label{con: noises}
For $k,r=0,\ldots,K$, the entries of $\Xi^k$ and $\{\xi_i^r:i\in\mathcal R_r\}$ are mutually independent, centered, and sub-Gaussian, with variances $\tau_{X,k}^2$ and $\tau_{Y,r}^2$, respectively. Their sub-Gaussian parameters are bounded uniformly in $n$ and $p$. The perturbations are independent of the complete data and the missingness pattern and are sampled only once.
\end{assumption}

Under this condition, $\E(\tSigma_k)=\Sigma_{kk}$, $\E(\tSigma_{kl}^{\mathrm{priv}})=\Sigma_{kl}$, and $\E(\tmbC_k^{\mathrm{priv}})=\mbC_k$. Thus, perturbation preserves the population targets and changes only the variability of the transmitted moments.

\begin{proposition}[Concentration of Perturbed Moments]
\label{the: perturbed moment concentration}
Suppose Conditions \ref{con: missing pattern}, \ref{con: PD}, \ref{con: subgaussian}, and \ref{con: noises} hold, with the positive-definiteness requirements in Condition \ref{con: PD} imposed on $\hSigma^{\mathrm{priv}}$ and $\hSigma^{\mathrm{priv},(j)}$. Let $1-a_1=\textcolor{blue}{O\{\sqrt{(\log p)/(\min_k n_k)}\}}$, $1-a_2=\textcolor{blue}{O\{\sqrt{(\log p)/(\min_{k\neq l}n_{kl})}\}}$. If $\min_{j,t}N_{jt}^{xx}\geq c_0\log p$ and $\min_jN_j^{xy}\geq c_0\log p$ for a sufficiently large constant $c_0$, then there exist constants $C_\Sigma,C_C,C_H>0$, independent of $n$ and $p$, such that
\begin{align*}
    \bbP\left(\left.\norm{\tSigma^{\mathrm{priv}}-\Sigma}_{\max}
    >C_\Sigma\sqrt{\frac{\log p}{\min_{j,t}N_{jt}^{xx}}}\,\right|R\right)
    &\leq \frac4p,\\
    \bbP\left(\left.\norm{\tmbC^{\mathrm{priv}}-\mbC}_{\infty}
    >C_C\sqrt{\frac{\log p}{\min_jN_j^{xy}}}\,\right|R\right)
    &\leq \frac4p,\\
    \bbP\left(\left.\norm{\hSigma^{\mathrm{priv}}-\Sigma}_{\max}
    >C_H\sqrt{\frac{\log p}{\min_{j,t}N_{jt}^{xx}}}\,\right|R\right)
    &\leq \frac4p.
\end{align*}
If $1-a_{1,j}^{\mathrm{priv}}=O\{\sqrt{(\log p)/(\min_k n_k)}\}$ and $1-a_{2,j}^{\mathrm{priv}}=O\{\sqrt{(\log p)/(\min_{k\neq l}n_{kl})}\}$, the final bound also holds with $\hSigma^{\mathrm{priv}}$ replaced by $\hSigma^{\mathrm{priv},(j)}$.
\end{proposition}

Consequently, when the perturbation scales remain bounded, these bounds have the same orders as their unperturbed counterparts: the perturbations change the constants but not the effective-sample-size rates.

\begin{theorem}[Estimation Rate for Perturbed ALB]
\label{cor: private global consistency}
Suppose the conditions of Proposition \ref{the: global consistency} and Proposition \ref{the: perturbed moment concentration} hold, with $\hSigma$ and $\tmbC$ replaced by $\hSigma^{\mathrm{priv}}$ and $\tmbC^{\mathrm{priv}}$. If $\lambda=2\norm{\tmbC^{\mathrm{priv}}-
    \hSigma^{\mathrm{priv}}\bbeta^*}_{\infty}$, then
\begin{align*}
    \norm{\cbbeta^{\mathrm{priv}}-\bbeta^*}_2
    =O_p\left(\norm{\bbeta^*}_1
      \sqrt{\frac{s\log p}{N_n}}\right), \quad
    \norm{\cbbeta^{\mathrm{priv}}-\bbeta^*}_1 =O_p\left(s\norm{\bbeta^*}_1
      \sqrt{\frac{\log p}{N_n}}\right).
\end{align*}
Moreover, if $\hSigma^{\mathrm{priv}}$ is positive definite and each block subproblem is solved exactly, conclusions (i)--(iii) of Theorem \ref{the: algorithmic convergence} hold with $\cQ_n$, $\cbbeta$, and $\hSigma$ replaced by $\cQ_n^{\mathrm{priv}}$, $\cbbeta^{\mathrm{priv}}$, and $\hSigma^{\mathrm{priv}}$, respectively.
\end{theorem}

For coordinatewise inference, we replace the moments and estimators in Section \ref{subsec: debiased estimator} with their perturbed counterparts. Under the corresponding sparsity and rate conditions of Theorem \ref{the: oracle property}, together with a positive lower bound on the perturbed influence variance, the same remainder, Lyapunov, and variance-consistency arguments yield a studentized Gaussian limit.

\textcolor{blue}{To complement the utility analysis with a formal privacy guarantee, Section \ref{app: release level ldp} in the Supplement gives an optional $(\epsilon,\delta)$-local differential privacy calibration for clipped one-time identity-level releases.}

\section{Simulations}
\label{sec: simulation}

\subsection{\textcolor{blue}{Simulation Setup}}
\label{subsec: MSE simulation}

Following the simulation designs of \citet{yu2020optimal} and
\citet{van2014asymptotically}, we evaluate the finite-sample
estimation and inference performance of ALB and the competing
methods described below.

For each observation, we generate
$y_i=\bx_i^\top\bbeta^*+\epsilon_i$, where
$\bx_i\sim\mathcal N(\mathbf0,\Sigma)$ and
$\epsilon_i\sim\mathcal N(0,1)$. The dimension is fixed at $p=900$, and the
covariates are partitioned equally among three users, each holding
300 coordinates. We consider two covariance structures. The first is autoregressive,
$ \Sigma=\left(0.6^{|j-t|}\right)_{j,t=1}^{900},$ whereas the second is block diagonal with 180 independent
$5\times5$ blocks having unit diagonal entries and off-diagonal entries equal to 0.15.

We combine these covariance structures with two signal patterns. Under the first
pattern,
$\bbeta^*
=0.5(\mathbf1_5^\top,\mathbf0_{295}^\top,
\mathbf1_5^\top,\mathbf0_{295}^\top,
\mathbf1_5^\top,\mathbf0_{295}^\top)^\top$,
so that the active coefficients are clustered within each user block. Under the second pattern, five coordinates are selected uniformly without replacement from each 300-dimensional user block and assigned a coefficient of $0.5$, while all remaining coefficients are set to zero. The selected locations are sampled once and held fixed across replications. \textcolor{red}{These combinations define Settings 1--4, all using the missingness pattern in Figure~\ref{fig: simulation settings}.}

\begin{figure}[!t]
    \centering
    \includegraphics[width=1\linewidth]{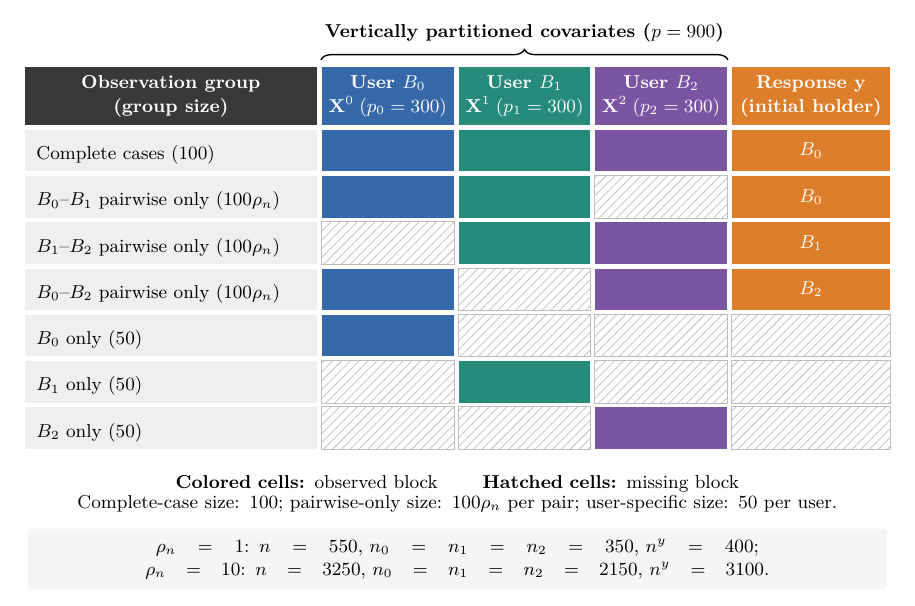}
    \caption{\textcolor{red}{Block-missing pattern used in the simulations.}}
    \label{fig: simulation settings}
\end{figure}

Specifically, each training data set contains 100 complete cases, three disjoint
pairwise-only groups of size $100\rho_n$, and three user-specific groups of
size 50, where $\rho_n$ is the ratio of the pairwise-only group size to the
complete-case sample size. Each pairwise-only observation contains the response
and the two corresponding covariate blocks, whereas each user-specific
observation contains only one covariate block. The three pairwise response
blocks are held by $B_0$, $B_1$, and $B_2$, respectively, and $B_0$ also holds
the complete-case responses. We set $\rho_n\in\{1,10\}$, yielding
$n=250+300\rho_n\in\{550,3250\}$ and pairwise overlap sizes
$100(1+\rho_n)$, while the number of complete cases remains fixed. For every setting and every value of $\rho_n$, the experiment is replicated $M=200$ times.

\subsection{Prediction and Estimation Performance}

We first evaluate the performance of ALB in prediction and
coefficient estimation. We compare it with four alternatives. (1) ALB-P is the
one-time perturbed version of ALB; the standard deviation of the Gaussian
perturbation for each user is independently drawn from
$\operatorname{Unif}[0.3,0.5]$. (2) DISCOM \citep{yu2020optimal} centrally
assembles the available-case moments and minimizes the same regularized
objective as ALB, thereby providing a direct benchmark for the optimization
error. (3) CC-Lasso fits the Lasso using only
the 100 complete observations. (4) Oracle Lasso reveals \textcolor{blue}{all masked covariate blocks and responses for all records} and is therefore infeasible. ALB, ALB-P, and DISCOM use the tuning procedure described in Section \ref{app: covariance regularization} in the Supplement, whereas CC-Lasso and Oracle Lasso select $\lambda$ by cross-validation. ALB and ALB-P are run for 10 sweeps. An independent test sample of size $n_{\mathrm{test}}=500$ is generated in each replication. For the $m$th replication, we evaluate the prediction mean squared error and the coefficient estimation error by \textcolor{blue}{$\operatorname{MSE}^{[m]}
    =n_{\mathrm{test}}^{-1}\sum_{i=1}^{n_{\mathrm{test}}}
      \left(y_{i}^{\mathrm{test},[m]}-\hat y_{i}^{\mathrm{test},[m]}\right)^2,$ and $
    \operatorname{Err}^{[m]}
    =\norm{\hbbeta^{[m]}-\bbeta^*}_2$.} Table \ref{tab: MSE_12} reports the results for Settings 1 and 2, while Table \ref{tab: MSE_34} in the Supplement reports the corresponding results for Settings 3 and 4.

\begin{table}[!t]
\centering
\caption{Prediction and coefficient-estimation performance in Settings 1 and 2. Values are means (SDs) over 200 replications.}
\label{tab: MSE_12}

\setlength{\tabcolsep}{2.8pt}
\renewcommand{\arraystretch}{0.75}
\begin{tabular}{@{}clcccc@{}}
\toprule
$\rho_n$ $(n)$
& Method
& \multicolumn{2}{c}{Setting 1}
& \multicolumn{2}{c}{Setting 2} \\
\cmidrule(lr){3-4}\cmidrule(lr){5-6}
& & MSE & Err & MSE & Err \\
\midrule

\multirow{5}{*}{$1\ (550)$}
& ALB
& 1.374 (0.174) & 0.637 (0.102)
& 1.497 (0.198) & 0.644 (0.096) \\
& ALB-P
& 1.438 (0.209) & 0.690 (0.114)
& 1.771 (0.314) & 0.777 (0.119) \\
& DISCOM
& 1.374 (0.174) & 0.637 (0.102)
& 1.497 (0.198) & 0.644 (0.096) \\
& CC-Lasso
& 1.984 (0.402) & 0.975 (0.166)
& 2.766 (0.775) & 1.151 (0.210) \\
& Oracle Lasso
& 1.106 (0.073) & 0.317 (0.045)
& 1.138 (0.077) & 0.323 (0.038) \\

\addlinespace[0.15em]
\midrule
\multirow{5}{*}{$10\ (3250)$}
& ALB
& 1.075 (0.074) & 0.339 (0.055)
& 1.034 (0.065) & 0.169 (0.031) \\
& ALB-P
& 1.109 (0.077) & 0.412 (0.061)
& 1.039 (0.066) & 0.185 (0.035) \\
& DISCOM
& 1.075 (0.074) & 0.339 (0.055)
& 1.034 (0.065) & 0.169 (0.031) \\
& CC-Lasso
& 1.984 (0.374) & 0.978 (0.182)
& 2.667 (0.619) & 1.136 (0.188) \\
& Oracle Lasso
& 1.019 (0.066) & 0.122 (0.017)
& 1.023 (0.064) & 0.124 (0.013) \\
\bottomrule
\end{tabular}
\end{table}

Tables \ref{tab: MSE_12} and \ref{tab: MSE_34} reveal three main patterns. First, ALB and DISCOM are identical up to the reported digits in all four settings. Thus, after 10 sweeps, the optimization error is numerically negligible, consistent with the geometric convergence in Theorem \ref{the: algorithmic convergence}. Second, increasing $\rho_n$ from 1 to 10 markedly reduces both the MSE and the coefficient estimation error of ALB, whereas CC-Lasso changes little because it always uses the same 100 complete observations. In Settings 1 and 2, the MSE reductions of ALB relative to CC-Lasso are 30.8\% and 45.9\% when $\rho_n=1$, and 45.8\% and 61.2\% when $\rho_n=10$, respectively. The gap between ALB and Oracle Lasso also narrows as the pairwise overlap increases. This improvement is consistent with Proposition \ref{the: global consistency}, whose rates depend on the available overlaps rather than only on complete cases. Third, ALB-P has moderately larger errors than ALB because of the added perturbation, but it remains more accurate than CC-Lasso for both metrics in every setting and continues to improve as $\rho_n$ increases.

\subsection{Confidence-Interval Performance}
\label{subsec: CI simulation}

We next evaluate the finite-sample performance of the 95\%
coordinatewise confidence intervals in \eqref{distributed CI}. In addition to
ALB and its perturbed version ALB-P, we consider four benchmarks:
(1) ALB-CP computes the regularized precision columns centrally while retaining
the ALB point estimator and available-case score, thereby isolating the numerical
effect of distributed precision-column optimization; (2) CC-DL
\citep{van2014asymptotically} applies debiased Lasso using only the 100 complete
cases; (3) \textcolor{blue}{Oracle-DL reveals the masked covariate blocks and \textcolor{red}{responses} for all records before applying debiased Lasso and is therefore infeasible}; and (4) ALB-OP replaces the estimated precision \textcolor{blue}{column} for coordinate $j$ with the population \textcolor{blue}{column} $\Sigma^{-1}\be_j$, thereby isolating the effect of precision-\textcolor{blue}{column} estimation. The precision-column tuning parameters for ALB, ALB-P, and ALB-CP are selected as described in Section \ref{app: precision covariance selection} in the Supplement, whereas CC-DL and Oracle-DL use cross-validation and nodewise regression.

Let $\cS=\operatorname{supp}(\bbeta^*)$, and let $\operatorname{CI}_{m,j}$
denote the confidence interval for $\beta_j^*$ in replication $m$. For
$\mathcal A\in\{\cS,\cS^c_{100}\}$, define
$
\operatorname{AvgCov}(\mathcal A)
=
(M|\mathcal A|)^{-1}
\sum_{j\in\mathcal A}\sum_{m=1}^{M}
\mathbf 1\{\beta_j^*\in \operatorname{CI}_{m,j}\},
\operatorname{AvgLen}(\mathcal A)
=
(M|\mathcal A|)^{-1}
\sum_{j\in\mathcal A}\sum_{m=1}^{M}
\operatorname{length}(\operatorname{CI}_{m,j}),$ and $
\operatorname{TPR}=
(M|\cS|)^{-1}
\sum_{j\in\cS}\sum_{m=1}^{M}
\mathbf 1\{0\notin \operatorname{CI}_{m,j}\},
$ where $\cS^c_{100}$ consists of 100 inactive coordinates
sampled without replacement from $\cS^c$. Table \ref{tab: CI_12} reports the results for Settings 1 and 2, while the
corresponding results for Settings 3 and 4 are reported in Table
\ref{tab: CI_34_small} in the Supplement.

\begin{table}[!t]
\centering
\caption{Performance of the 95\% coordinatewise confidence intervals in Settings 1 and 2. Results are based on 200 replications.}
\label{tab: CI_12}

{\small
\setlength{\tabcolsep}{3pt}
\renewcommand{\arraystretch}{0.8}
\resizebox{\textwidth}{!}{
\begin{tabular}{lcccccccccc}
\toprule
& \multicolumn{2}{c}{AvgCov$(\cS)$}
& \multicolumn{2}{c}{AvgLen$(\cS)$}
& \multicolumn{2}{c}{AvgCov$(\cS^c_{100})$}
& \multicolumn{2}{c}{AvgLen$(\cS^c_{100})$}
& \multicolumn{2}{c}{TPR} \\
\cmidrule(lr){2-3}
\cmidrule(lr){4-5}
\cmidrule(lr){6-7}
\cmidrule(lr){8-9}
\cmidrule(lr){10-11}
Method
& $\rho_n=1$ & $\rho_n=10$
& $\rho_n=1$ & $\rho_n=10$
& $\rho_n=1$ & $\rho_n=10$
& $\rho_n=1$ & $\rho_n=10$
& $\rho_n=1$ & $\rho_n=10$ \\
\midrule

\multicolumn{11}{c}{\textit{Setting 1}} \\
\addlinespace[0.15em]
ALB
& 0.909 & 0.925
& 0.514 & 0.293
& 0.948 & 0.951
& 0.518 & 0.300
& 0.930 & 1.000 \\
ALB-P
& 0.919 & 0.927
& 0.749 & 0.411
& 0.951 & 0.950
& 0.767 & 0.428
& 0.680 & 0.995 \\
ALB-CP
& 0.909 & 0.925
& 0.514 & 0.293
& 0.948 & 0.951
& 0.518 & 0.300
& 0.930 & 1.000 \\
CC-DL
& 0.563 & 0.553
& 0.318 & 0.308
& 0.942 & 0.936
& 0.320 & 0.310
& 0.948 & 0.949 \\
Oracle-DL
& 0.885 & 0.938
& 0.208 & 0.094
& 0.952 & 0.951
& 0.209 & 0.094
& 1.000 & 1.000 \\
ALB-OP
& 0.944 & 0.951
& 0.734 & 0.337
& 0.948 & 0.952
& 0.717 & 0.328
& 0.731 & 0.999 \\

\addlinespace[0.60em]
\multicolumn{11}{c}{\textit{Setting 2}} \\
\addlinespace[0.15em]
ALB
& 0.900 & 0.939
& 0.354 & 0.179
& 0.947 & 0.953
& 0.350 & 0.179
& 0.998 & 1.000 \\
ALB-P
& 0.899 & 0.947
& 0.454 & 0.221
& 0.950 & 0.950
& 0.449 & 0.221
& 0.965 & 1.000 \\
ALB-CP
& 0.900 & 0.939
& 0.354 & 0.179
& 0.947 & 0.953
& 0.350 & 0.179
& 0.998 & 1.000 \\
CC-DL
& 0.533 & 0.512
& 0.281 & 0.260
& 0.942 & 0.942
& 0.280 & 0.259
& 0.959 & 0.960 \\
Oracle-DL
& 0.891 & 0.938
& 0.162 & 0.070
& 0.949 & 0.951
& 0.162 & 0.070
& 1.000 & 1.000 \\
ALB-OP
& 0.909 & 0.946
& 0.387 & 0.191
& 0.948 & 0.953
& 0.386 & 0.191
& 0.992 & 1.000 \\
\bottomrule
\end{tabular}
}}
\end{table}

We obtain three main findings. First, ALB and ALB-CP are identical or nearly identical throughout the four settings. Hence, distributing the precision-column calculation introduces negligible numerical error after the prescribed sweeps, as predicted by Theorem \ref{the: algorithmic convergence}. Second, as $\rho_n$ increases, the proposed procedures approach the two oracle benchmarks, Oracle-DL and ALB-OP. When $\rho_n=10$, the active-coordinate coverage of ALB and ALB-P ranges from 0.925 to 0.949 across the four settings, compared with 0.938 to 0.951 for the two oracle benchmarks. Their inactive-coordinate coverage ranges from 0.947 to 0.953, and their TPRs are essentially one. The ALB intervals also become substantially shorter as $\rho_n$ increases. These findings are in line with the studentized asymptotic normality in Theorem \ref{the: oracle property}. Third, ALB and ALB-P achieve substantially higher active-coordinate coverage than CC-DL while maintaining inactive-coordinate coverage close to the nominal level. Although ALB-P produces wider intervals because of the additional perturbation, its coverage remains comparable to that of ALB, consistent with the perturbed studentized result in Section \ref{sec: perturbed estimator}.

\section{ADNI Data Analysis}
\label{sec: ADNI}

We analyze multimodal data from the Alzheimer's Disease Neuroimaging Initiative (ADNI) \citep{mueller2005alzheimer}. CSF biomarkers, structural MRI, PET, and cognitive outcomes are processed by different pipelines and are not jointly observed at every visit, yielding a block-missing structure. Moreover, the sensitivity of these patient-level data motivates an analysis that avoids exchanging raw records. Section \ref{app: ADNI modalities} in the Supplement describes the modalities, user allocation, sample construction, and preprocessing.

\begin{table}[!t]
\centering
\caption{ADNI prediction and feature-selection performance over 200 train-test splits. Values are means (SDs); paired labels list the prediction method first.}
\label{tab: ADNI}

{\footnotesize
\setlength{\tabcolsep}{7pt}
\renewcommand{\arraystretch}{0.82}
\begin{tabular}{@{}lcc@{}}
\toprule
Method & Test MSE & Number selected \\
\midrule
ALB
    & 0.561 (0.139)
    & 18.540 (2.729) \\
ALB-P
    & 0.567 (0.149)
    & 16.185 (3.839) \\
DISCOM / ALB-CP
    & 0.561 (0.139)
    & 18.555 (2.717) \\
CC-Lasso / CC-DL
    & 0.861 (0.282)
    & 21.325 (10.197) \\
\bottomrule
\end{tabular}
}
\end{table}

\begin{figure}[!t]
    \centering
    \includegraphics[width=1\linewidth]{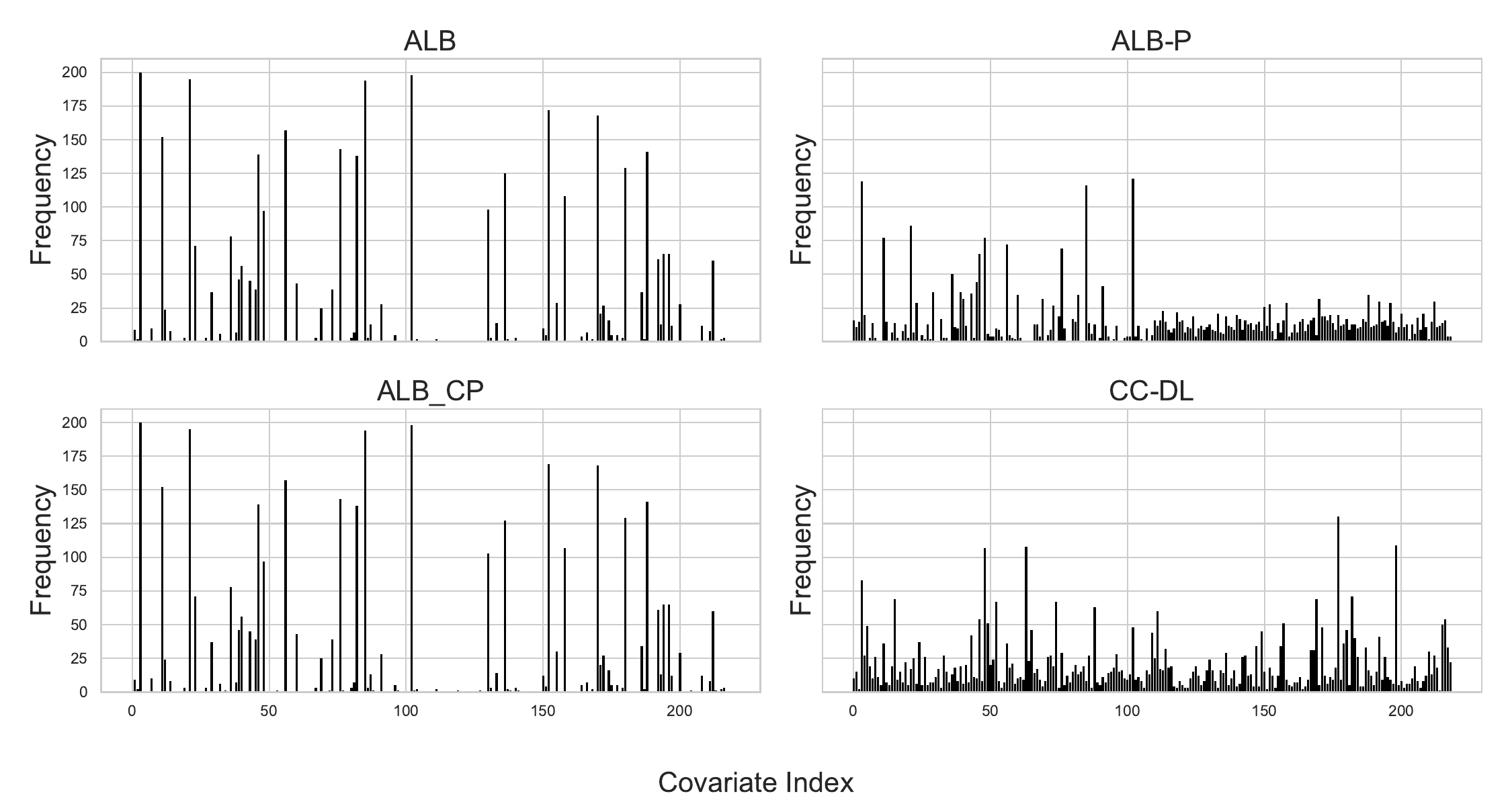}
    \caption{Selection frequencies of the 219 covariates across 200 train-test
    splits for predicting MMSE.}
    \label{fig: ADNI}
\end{figure}

Table \ref{tab: ADNI} summarizes the prediction and feature-selection results over 200 train-test splits. A feature is selected when its confidence interval excludes zero. For prediction, ALB and DISCOM yield the same mean test MSE of 0.561, representing a 34.9\% reduction relative to the CC-Lasso MSE of 0.861. ALB-P retains comparable prediction accuracy, with an MSE of 0.567, and selects 16.185 features on average, compared with 18.540 for ALB. Turning to feature selection, Figure \ref{fig: ADNI} shows that ALB and ALB-CP produce nearly identical selection frequencies. CC-DL has a more diffuse selection profile, whereas ALB-P selects fewer features and exhibits somewhat greater variation across splits. ALB selects ST104TA, ST96SV, ST121TA, and ST74TA in more than 190 of the 200 splits. These variables measure right pars-opercularis
cortical thickness, right lateral-ventricle volume, right transverse-temporal cortical thickness, and right caudal-middle-frontal cortical thickness, respectively. The three most frequently selected variables under ALB-P are ST96SV, ST104TA, and ST74TA. These recurrent selections are consistent with previous findings linking cortical thinning and ventricular enlargement to Alzheimer's disease and cognitive impairment
\citep{yang2022associations,nestor2008ventricular,li2025voxel,
rizvi2021association}.

\section{Discussion}
\label{sec: conclusion}
We introduced ALB for high-dimensional estimation and inference with vertically distributed, block-missing data. For a fixed number of users, ALB solves the centralized DISCOM objective without pooling raw records, communicates $O(n)$ scalars per sweep, and supports distributed debiasing. Theory establishes optimization, estimation, and inference guarantees. Simulations and the ADNI analysis show performance close to centralized benchmarks, and one-time perturbation reduces direct disclosure while retaining the established convergence orders under the stated noise scaling.

Future work may extend ALB to generalized linear, survival, and longitudinal models and accommodate data-dependent or nonignorable missingness. A second direction is to develop end-to-end privacy guarantees for the full iterative protocol, together with adaptive perturbation schemes that improve the privacy-utility tradeoff. Finally, communication-efficient variants for a growing number of users, asynchronous updates, and partial user participation would broaden the method's use in large and heterogeneous federated systems.


\textbf{Disclosure Statement.} The authors report that they have no competing interests to declare.

\textbf{Data Availability Statement.} The ADNI data used in this article are available from the ADNI website: \href{https://adni.loni.usc.edu/data-samples/adni-data/}{https://adni.loni.usc.edu/data-samples/adni-data/}. 




\begingroup
\small
\setlength{\bibsep}{0pt}
\renewcommand{\baselinestretch}{1.75}\selectfont
\putbib[citation]
\endgroup
\end{bibunit}

\clearpage
\begin{bibunit}[abbrvnat]
\spacingset{1.8}
\makeatletter
\def\hyper@natlinkstart#1{%
  \Hy@backout{supp.#1}%
  \hyper@linkstart{cite}{cite.supp.#1}%
  \def\hyper@nat@current{supp.#1}%
}
\def\hyper@natlinkbreak#1#2{%
  \hyper@linkend#1\hyper@linkstart{cite}{cite.supp.#2}%
}
\def\hyper@natanchorstart#1{%
  \Hy@raisedlink{\hyper@anchorstart{cite.supp.#1}}%
}
\makeatother
\setcounter{section}{0}
\renewcommand{\thesection}{S.\arabic{section}}
\renewcommand{\theHsection}{supp.\arabic{section}}
\renewcommand{\theHsubsection}{\theHsection.\arabic{subsection}}
\setcounter{figure}{0}
\renewcommand{\thefigure}{S.\arabic{figure}}
\renewcommand{\theHfigure}{supp.figure.\arabic{figure}}
\setcounter{table}{0}
\renewcommand{\thetable}{S.\arabic{table}}
\renewcommand{\theHtable}{supp.table.\arabic{table}}
\begin{center}
{\Large\bfseries Supplementary Material to\\[0.5em]
``High-Dimensional Assisted Learning for Vertically Distributed Data with Blockwise Missingness''}
\end{center}
\vspace{2\baselineskip}

\noindent\textbf{Abstract.} This supplementary material contains proofs, an optional release-level local differential privacy calibration, additional simulation results, and ADNI data details. It also provides implementation details for parameter tuning, positive-definiteness screening, distributed eigenvalue estimation, coordinate-descent optimization, and coordinatewise inference.

\section{Proofs}
\label{app: proofs}

\subsection{Proofs of Propositions}

\subsubsection{Proof of Proposition \ref{the: sample cov}}
\begin{proof}
    The elementwise covariance bound is Theorem 1 of \citet{yu2020optimal}; applying a union bound over the $p^2$ entries gives the stated max-norm bound. The same concentration argument applied to the available response-covariate pairs gives the two bounds for $\tmbC$, with $N_j^{xy}$ replacing the covariate-pair sample size. It remains to verify the bound for the positive-definite modification $\hSigma$, which is proved in Appendix of \citet{yu2020optimal}. The argument for $\hSigma^{(j)}$ follows the same routine and is therefore omitted.
\end{proof}


    

\subsubsection{Proof of Proposition \ref{the: global consistency}}
\begin{proof}
    The proof follows Theorem 2 of \citet{yu2020optimal}, with the response-covariate concentration rate replaced by the bound in Proposition \ref{the: sample cov}. Specifically, their Lasso basic-inequality argument and the restricted eigenvalue condition yield
    $$
        \norm{\cbbeta-\bbeta^*}_2
        =O_p\left(\norm{\bbeta^*}_1
        \sqrt{\frac{s\log p}{N_n}}\right).
    $$
    This is the first assertion. For completeness, let $\bdelta=\cbbeta-\bbeta^*$ and $S=\operatorname{supp}(\bbeta^*)$. The same basic inequality gives the cone bound used in \citet{yu2020optimal}, and in particular
    $$
        \norm{\bdelta}_1
        \leq8\norm{\bdelta_S}_1
        \leq8\sqrt{s}\norm{\bdelta}_2.
    $$
    Combining the two displays gives
    $$
        \norm{\cbbeta-\bbeta^*}_1
        =O_p\left(s\norm{\bbeta^*}_1
        \sqrt{\frac{\log p}{N_n}}\right),
    $$
    as required.
\end{proof}
        
\subsection{Proofs of Main Theoretical Results}

\subsubsection{Proof of Lemma \ref{le: normality}}
\begin{proof}
    For a scalar random variable $Z$, write
$$
    \norm{Z}_{\psi_1}\coloneqq
    \inf\{t>0:\E\exp(\abs{Z}/t)\leq2\}.
$$
The variable $Z$ is sub-exponential when this Orlicz norm is finite. We first derive the variance of the score.
    \begin{lemma}
    \label{le: score variance}
        Conditional on $R$, the influence vectors $\bphi_1,\ldots,\bphi_n$ in \eqref{population influence contribution} are independent, satisfy $\E(\bphi_i\mid R)=0$, and
        \begin{equation}
        \label{exact score variance}
            \Omega_n=n\Var\{S_n(\bbeta^*)\mid R\}
            =\frac1n\sum_{i=1}^n\Var(\bphi_i\mid R).
        \end{equation}
        Under Conditions \ref{con: conditional variance of y} and \ref{con: overlap covariance},
        \begin{equation}
        \label{score variance lower bound}
            \Omega_n\succeq c(\mathbf P_n\circ\Sigma)
            \succeq c\kappa_P\mathrm I_p.
        \end{equation}
    \end{lemma}
    \begin{proof}
        Conditional on an MCAR pattern $R$, the complete-data observations remain independent and identically distributed. Since $\E(\bx_i^k\bx_i^{l\top})=\Sigma_{kl}$ and $\E(y_i\bx_i^k)=\mbC_k$, we have $\E(A_i\mid R)=0$, $\E(q_i\mid R)=0$, and hence $\E(\bphi_i\mid R)=0$. The exact representation $S_n(\bbeta^*)=n^{-1}\sum_i\bphi_i$ and conditional independence prove \eqref{exact score variance}.

        It remains to establish the lower bound. Define
        $$
            D_i=\operatorname{diag}\left(
            r_i^0\frac n{n_0^y}\mathrm I_{p_0},\ldots,
            r_i^K\frac n{n_K^y}\mathrm I_{p_K}\right).
        $$
        Conditional on $(\bx_i,R)$, the only randomness in $\bphi_i=A_i\bbeta^*-q_i$ that is needed for a lower bound comes from $y_i$ in $q_i$. By the law of total variance,
        \begin{align*}
            \Var(\bphi_i\mid R)
            &\succeq\E\{\Var(\bphi_i\mid\bx_i,R)\mid R\}\\
            &=r_i^yD_i\E\{\Var(y_i\mid\bx_i)\bx_i\bx_i^\top\}D_i\\
            &\succeq c r_i^yD_i\Sigma D_i.
        \end{align*}
        The $(k,l)$ block of $n^{-1}\sum_i c r_i^yD_i\Sigma D_i$ equals
        $$
            c\frac{n n_{kl}^y}{n_k^yn_l^y}\Sigma_{kl}.
        $$
        Therefore this matrix is exactly $c(\mathbf P_n\circ\Sigma)$, and Condition \ref{con: overlap covariance} completes the proof.
    \end{proof}

        Set $Z_{i,j}=\bthe_j^{*\top}\bphi_i$. By Lemma \ref{le: score variance}, conditional on $R$ these variables are independent, centered, and
        $$
            \frac1n\sum_{i=1}^n\Var(Z_{i,j}\mid R)=v_{n,j}.
        $$
        Moreover, \eqref{score variance lower bound} and $\Sigma_{jj}=1$ give
        $$
            v_{n,j}\geq c\kappa_P\norm{\bthe_j^*}_2^2
            \geq c\kappa_P.
        $$
        For the last inequality, the Schur-complement bound gives $\theta_{jj}^*=(\Sigma^{-1})_{jj}\geq1/\Sigma_{jj}=1$, \textcolor{blue}{where $\theta_{jj}^*$ is the $j$th entry of $\bthe_j^*$}, and hence $\norm{\bthe_j^*}_2\geq1$.

        Under Condition \ref{con: ratios}, all weights $n/n_{kl}$ and $n/n_k^y$ in \eqref{population influence contribution} are uniformly bounded. Products of sub-Gaussian variables are sub-exponential. Hence, uniformly in $i$,
        $$
            \norm{Z_{i,j}}_{\psi_1}
            \leq C\norm{\bthe_j^*}_1(1+\norm{\bbeta^*}_1).
        $$
        The standard moment bound for a sub-exponential variable therefore yields
        $$
            \max_{1\leq i\leq n}\E(|Z_{i,j}|^4\mid R)
            \leq C\norm{\bthe_j^*}_1^4(1+\norm{\bbeta^*}_1)^4.
        $$
        Consequently the Lyapunov ratio with exponent four satisfies
        $$
            \frac{\sum_{i=1}^n\E(|Z_{i,j}|^4\mid R)}
            {\left\{\sum_{i=1}^n\Var(Z_{i,j}\mid R)\right\}^2}
            \leq C\frac{\norm{\bthe_j^*}_1^4(1+\norm{\bbeta^*}_1)^4}
            {n v_{n,j}^2}\longrightarrow0.
        $$
        Lyapunov's central limit theorem now gives
        $$
            \frac{n^{-1/2}\sum_{i=1}^nZ_{i,j}}{\sqrt{v_{n,j}}}
            =\frac{\sqrt n\,\bthe_j^{*\top}S_n(\bbeta^*)}{\sqrt{v_{n,j}}}
            \xrightarrow{d}\mathcal N(0,1).
        $$
    \end{proof}

\subsubsection{Proof of Theorem \ref{the: algorithmic convergence}}
\begin{proof}
    In addition to the sweep index $e$ used in the theorem, let $e'$ count individual block updates. One complete sweep contains $K+1$ block updates, so $e=\lfloor e'/(K+1)\rfloor$. Throughout the proof, $\lambda_n^{(min)}=\lambda_{\min}(\hSigma)$ and $\lambda_n^{(max)}=\lambda_{\max}(\hSigma)$, consistent with the theorem statement.
    
    \begin{lemma}
    \label{le: strongly convex}
        The objective $\cQ_n$ is $\lambda_n^{(min)}$-strongly convex. In particular,
        \begin{equation}
        \label{strongly convex}
            \cQ_n(\bbeta^{(1)})\geq \cQ_n(\bbeta^{(2)}) + g_2^\top (\bbeta^{(1)}-\bbeta^{(2)}) + \frac{\lambda_n^{(min)}}{2}\norm{\bbeta^{(1)}-\bbeta^{(2)}}_2^2,
        \end{equation}
        for every $g_2\in\partial\cQ_n(\bbeta^{(2)})$.
    \end{lemma}
    \begin{proof}
        The quadratic loss $\cL_n$ has Hessian $\hSigma$ and is therefore $\lambda_n^{(min)}$-strongly convex. Adding the convex function $\lambda\norm{\cdot}_1$ preserves the same strong convexity. Inequality \eqref{strongly convex} is the corresponding subgradient inequality; see Lemma 2 of \citet{zhou2018fencheldualitystrongconvexity}.
    \end{proof}

    \begin{lemma}
    \label{le: diff of params bounded by diff of loss}
        $\cQ_n(\hbbeta^{(e')})-\cQ_n(\cbbeta)\geq \frac{\lambda_n^{(min)}}{2}\norm{\hbbeta^{(e')}-\cbbeta}_2^2$.
    \end{lemma}
    \begin{proof}
        Since $\cbbeta$ minimizes $\cQ_n$, the Fermat condition gives $\mathbf0\in\partial\cQ_n(\cbbeta)$. Applying \eqref{strongly convex} with $\bbeta^{(1)}=\hbbeta^{(e')}$ and $\bbeta^{(2)}=\cbbeta$ proves the claim.
    \end{proof}

    \begin{lemma}
    \label{le: diff of loss bounded by sub-gradient}
        For any $\bbeta\in\mathbb R^p$ and $g\in\partial\cQ_n(\bbeta)$,
        \begin{equation}
        \label{diff of loss bounded by sub-gradient}
            \cQ_n(\bbeta)-\cQ_n(\cbbeta)\leq \frac{1}{2\lambda_n^{(min)}}\norm{g}_2^2.
        \end{equation}
    \end{lemma}
    \begin{proof}
        By \eqref{strongly convex}, for any $\bbeta,z\in\mathbb R^p$ and $g\in\partial\cQ_n(\bbeta)$,
        \begin{equation*}
            \cQ_n(z)\geq \cQ_n(\bbeta) + g^\top (z-\bbeta)+\frac{\lambda_n^{(min)}}{2}\norm{z-\bbeta}_2^2.
        \end{equation*}
        Denote the right-hand side by $H(z)$ and minimize both sides over $z$:
        \begin{equation*}
            \min_z\cQ_n(z)\geq \min_z H(z).
        \end{equation*}
        The first-order condition for the quadratic function $H$ is
        \begin{equation*}
            \nabla H(z) = g+\lambda_n^{(min)}(z-\bbeta)=0.
        \end{equation*}
        Hence, $z^*=\bbeta-g/\lambda_n^{(min)}$. Substitution gives
        \begin{equation*}
        \begin{aligned}
            \min_z\cQ_n(z) &\geq \cQ_n(\bbeta)-\frac{1}{\lambda_n^{(min)}}\norm{g}_2^2+\frac{1}{2\lambda_n^{(min)}}\norm{g}_2^2\\
            &= \cQ_n(\bbeta)-\frac{1}{2\lambda_n^{(min)}}\norm{g}_2^2.
        \end{aligned}
        \end{equation*}
        Since $\cbbeta$ minimizes $\cQ_n$,
        \begin{equation*}
            \cQ_n(\cbbeta)\geq \cQ_n(\bbeta)-\frac{1}{2\lambda_n^{(min)}}\norm{g}_2^2.
        \end{equation*}
        Rearranging proves the result.
    \end{proof}

    \begin{lemma}
    \label{le: linear convergence}
        \begin{equation*}
            \cQ_n(\hbbeta^{(e'+1)})-\cQ_n(\cbbeta)\leq \eta_n[\cQ_n(\hbbeta^{(e'-K)})-\cQ_n(\cbbeta)],
        \end{equation*}
        where $0< \eta_n\coloneqq1-\frac{(\lambda_n^{(min)})^2}{(\lambda_n^{(min)})^2+(K+1)(\lambda_n^{(max)})^2}<1$. 
    \end{lemma}
    \begin{proof}
        Without loss of generality, suppose $B_K$ performs block update $e'+1$. During the most recent complete sweep, $B_k$ was updated from iterate $e'+k-K$ to iterate $e'+1+k-K$, for $k=0,\ldots,K$. Exact minimization of that block subproblem and the Fermat condition give $\mathbf0\in\partial_{\bbeta_k}\cQ_n(\hbbeta^{(e'+1+k-K)})$. The two consecutive iterates differ only in block $k$, so \eqref{strongly convex} yields
        \begin{equation}
        \label{next step diff of params bounded by next step diff of loss}
            \cQ_n(\hbbeta^{(e'+k-K)})-\cQ_n(\hbbeta^{(e'+1+k-K)})\geq \frac{\lambda_n^{(min)}}{2}\norm{\hbbeta^{(e'+k-K)}-\hbbeta^{(e'+1+k-K)}}_2^2, \quad k=0,\cdots,K.
        \end{equation}
        Summing over the $K+1$ block updates gives
        \begin{equation}
        \label{telescope of next step diff}
            \cQ_n(\hbbeta^{(e'-K)})-\cQ_n(\hbbeta^{(e'+1)})\geq \frac{\lambda_n^{(min)}}{2}\sum_{k=0}^K \norm{\hbbeta^{(e'+k-K)}-\hbbeta^{(e'+1+k-K)}}_2^2.
        \end{equation}

    For each $k$, choose $\bv_k\in\lambda\partial_{\bbeta_k}\norm{\hbbeta^{(e'+1+k-K)}}_1$ such that $\bv_k=-\nabla_{\bbeta_k}\cL_n(\hbbeta^{(e'+1+k-K)})$. Because the $\ell_1$ penalty is block-separable and block $k$ remains unchanged after its update, $\bv_k$ is also a valid block subgradient at the final iterate $\hbbeta^{(e'+1)}$. Therefore,
    \begin{equation}
    \label{invariant sub-gradient}
        \xi_k^{(e'+1)}\coloneqq \nabla_{\bbeta_k}\cL_n(\hbbeta^{(e'+1)})+\bv_k \in \partial_{\bbeta_k} \cQ_n(\hbbeta^{(e'+1)}).
    \end{equation}
    By the definition of $\bv_k$,
    \begin{equation}
    \label{invariant sub-gradient expansion}
         \xi_k^{(e'+1)}= \nabla_{\bbeta_k}\cL_n(\hbbeta^{(e'+1)})-\nabla_{\bbeta_k}\cL_n(\hbbeta^{(e'+1+k-K)}).
    \end{equation}
    Since $\nabla\cL_n(\bbeta)=\hSigma\bbeta-\tmbC$ and $\norm{\hSigma}_2=\lambda_n^{(max)}$,
    \begin{equation}
    \label{bound of each sub-gradient}
        \norm{\xi_k^{(e'+1)}}_2^2\leq \norm{\nabla\cL_n(\hbbeta^{(e'+1)})-\nabla\cL_n(\hbbeta^{(e'+1+k-K)})}_2^2\leq (\lambda_n^{(max)})^2\norm{\hbbeta^{(e'+1)}-\hbbeta^{(e'+1+k-K)}}_2^2,
    \end{equation}
    Let $\xi^{(e'+1)}\coloneqq(\xi_0^{(e'+1)\top},\ldots,\xi_K^{(e'+1)\top})^\top$. Summing the preceding bounds and noting that any block update appears in at most $K+1$ of the resulting differences gives
    \begin{equation}
    \begin{aligned}
    \label{bound of all sub-gradient}
        \norm{\xi^{(e'+1)}}_2^2&\leq (\lambda_n^{(max)})^2\sum_{k=0}^K \norm{\hbbeta^{(e'+1)}-\hbbeta^{(e'+1+k-K)}}_2^2\\
        &\leq (K+1)(\lambda_n^{(max)})^2\sum_{k=0}^K \norm{\hbbeta^{(e'+1+k-K)}_k - \hbbeta^{(e'+k-K)}_k}_2^2,
    \end{aligned}
    \end{equation}
    Moreover, $\xi^{(e'+1)}\in\partial\cQ_n(\hbbeta^{(e'+1)})$. Lemma \ref{le: diff of loss bounded by sub-gradient} therefore gives
    \begin{equation}
    \label{diff of loss between global minimizer bounded by diff of loss between next epoch}
    \begin{aligned}
        \cQ_n(\hbbeta^{(e'+1)})-\cQ_n(\cbbeta)&\leq \frac{1}{2\lambda_n^{(min)}}\norm{\xi^{(e'+1)}}_2^2\\
        &\leq \frac{(K+1)(\lambda_n^{(max)})^2}{2\lambda_n^{(min)}}\sum_{k=0}^K \norm{\hbbeta^{(e'+1+k-K)}_k - \hbbeta^{(e'+k-K)}_k}_2^2 \\
        &\leq \frac{(K+1)(\lambda_n^{(max)})^2}{2\lambda_n^{(min)}}\sum_{k=0}^K \norm{\hbbeta^{(e'+1+k-K)} - \hbbeta^{(e'+k-K)}}_2^2 \\
        &\leq \frac{(K+1)(\lambda_n^{(max)})^2}{(\lambda_n^{(min)})^2}(\cQ_n(\hbbeta^{(e'-K)})-\cQ_n(\hbbeta^{(e'+1)}),
    \end{aligned}
    \end{equation}
    where the last inequality uses \eqref{telescope of next step diff}. Let $D_n\coloneqq (K+1)(\lambda_n^{(max)})^2/(\lambda_n^{(min)})^2$. Rearranging \eqref{diff of loss between global minimizer bounded by diff of loss between next epoch} gives
    \begin{equation*}
    \begin{aligned}
        (1+D_n)(\cQ_n(\hbbeta^{(e'+1)})-\cQ_n(\cbbeta))&\leq D_n(\cQ_n(\hbbeta^{(e'-K)})-\cQ_n(\cbbeta))\\
        \cQ_n(\hbbeta^{(e'+1)})-\cQ_n(\cbbeta)&\leq \frac{D_n}{1+D_n}(\cQ_n(\hbbeta^{(e'-K)})-\cQ_n(\cbbeta)).
    \end{aligned}
    \end{equation*}
    Since
    $$
        \frac{D_n}{1+D_n}
        =1-\frac{(\lambda_n^{(min)})^2}
        {(\lambda_n^{(min)})^2+(K+1)(\lambda_n^{(max)})^2}
        =\eta_n,
    $$
    and $\lambda_n^{(min)}>0$ under Condition \ref{con: PD}, the proof is complete.
    \end{proof}
    
    We now prove the three assertions of Theorem \ref{the: algorithmic convergence}.
        \emph{Part (i).} Combining Lemmas \ref{le: diff of params bounded by diff of loss} and \ref{le: linear convergence} gives
        \begin{equation*}
            \norm{\hbbeta^{(e'+1)}-\cbbeta}_2^2\leq \frac{2}{\lambda_n^{(min)}}\eta_n (\cQ_n(\hbbeta^{(e'-K)})-\cQ_n(\cbbeta)).
        \end{equation*}
        Iterating this inequality over $e$ complete sweeps and using $c_n=\cQ_n(\hbbeta^{(0)})-\cQ_n(\cbbeta)$ yields
        \begin{equation*}
            \norm{\hbbeta^{(e)}-\cbbeta}_2^2 \leq \frac{2}{\lambda_n^{(min)}}\eta_n^e c_n.
        \end{equation*}
        Taking square roots and using $q_n^{(1)}=\sqrt{\eta_n}$ and $q_n^{(2)}=\sqrt{2c_n/\lambda_n^{(min)}}$ proves \eqref{algorithmic l2 bound}. Condition \ref{con: PD} ensures $q_n^{(1)}\in(0,1)$, while $q_n^{(2)}\geq0$; the latter equals zero exactly when the initial iterate is already $\cbbeta$.

        \emph{Part (ii).} The triangle inequality, Proposition \ref{the: global consistency}, and part (i) imply
        \begin{equation*}
        \begin{aligned}
            \norm{\hbbeta^{(e)}-\bbeta^*}_2
            &\leq \norm{\hbbeta^{(e)}-\cbbeta}_2+\norm{\cbbeta-\bbeta^*}_2\\
            &\leq q_n^{(2)}(q_n^{(1)})^e
            +O_p\left(\norm{\bbeta^*}_1\sqrt{\frac{s\log p}{N_n}}\right).
        \end{aligned}
        \end{equation*}

        \emph{Part (iii).} If $q_n^{(2)}=0$, the conclusion is immediate. Otherwise, part (i) and $\norm{\bv}_1\leq\sqrt p\norm{\bv}_2$ give
        $$
            \norm{\hbbeta^{(e)}-\cbbeta}_1
            \leq\sqrt p\,q_n^{(2)}(q_n^{(1)})^e.
        $$
        The right-hand side is at most $\sqrt{(\log p)/n}$ whenever
        $$
            e\log\{1/q_n^{(1)}\}
            \geq\log\left\{q_n^{(2)}\sqrt{\frac{np}{\log p}}\right\}.
        $$
        Taking the nonnegative integer ceiling gives the stopping rule in the theorem.
        
    \end{proof}

\subsubsection{Proof of Theorem \ref{the: oracle property}}
\begin{proof}
Under Condition \ref{con: subgaussian}, the products appearing in the influence vectors are sub-exponential, which provides the moment and maximal inequalities used below.

    We first record the rates supplied by Section \ref{subsec: algorithmic convergence}. Proposition \ref{the: global consistency} and part (iii) of Theorem \ref{the: algorithmic convergence} give
    \begin{equation}
    \label{beta rate used in debiasing}
        \norm{\hbbeta-\bbeta^*}_1=O_p(r_{\beta,n}).
    \end{equation}
    Apply the same two theorems to the inverse-column loss in \eqref{global inverse column}. Its population target is $\bthe_j^*$ because $\Sigma\bthe_j^*=\be_j$, and its first-order stochastic error is
    $$
        \norm{({\hSigma^{(j)}}-\Sigma)\bthe_j^*}_\infty
        \leq\norm{{\hSigma^{(j)}}-\Sigma}_{\max}\norm{\bthe_j^*}_1.
    $$
    The stated choice of $\lambda_j$ therefore yields
    \begin{equation}
    \label{theta rate used in debiasing}
        \norm{\hbthe_j-\bthe_j^*}_1=O_p(r_{\theta,n}),
        \qquad
        \norm{\hbthe_j}_1=O_p(\norm{\bthe_j^*}_1+r_{\theta,n}).
    \end{equation}
    If $\cbthe_j$ denotes the exact minimizer in \eqref{global inverse column}, its KKT condition is
    $$
        {\hSigma^{(j)}}\cbthe_j-\be_j+\lambda_j z_j=0,
        \qquad z_j\in\partial\norm{\cbthe_j}_1,\quad\norm{z_j}_\infty\leq1.
    $$
    Hence, using the distributed optimization rate and $\norm{{\hSigma^{(j)}}}_{\max}=O_p(1)$,
    \begin{equation}
    \label{distributed inverse KKT rate}
    \begin{aligned}
        \norm{\be_j-{\hSigma^{(j)}}\hbthe_j}_\infty
        &\leq\lambda_j+\norm{{\hSigma^{(j)}}}_{\max}
        \norm{\hbthe_j-\cbthe_j}_1\\
        &=O_p(\mu_{j,n}).
    \end{aligned}
    \end{equation}
    Proposition \ref{the: sample cov} also implies
    \begin{equation}
    \label{score and surrogate rates}
    \begin{aligned}
        \norm{S_n(\bbeta^*)}_\infty
        &\leq\norm{\tSigma-\Sigma}_{\max}\norm{\bbeta^*}_1
          +\norm{\tmbC-\mbC}_\infty=O_p(a_n),\\
        \norm{{\hSigma^{(j)}}-\tSigma}_{\max}&=O_p(h_n).
    \end{aligned}
    \end{equation}
    The second line follows directly from the definition of {$\hSigma^{(j)}$}, the rates of {$1-a_{1,j}$ and $1-a_{2,j}$}, and $c_{tr}/p=1$ under the standardization in Condition \ref{con: PD}.

    We next bound the three exact remainders in \eqref{debiased error decomposition}. By H\"older's inequality and \eqref{beta rate used in debiasing}--\eqref{score and surrogate rates},
    \begin{align*}
        \sqrt n|\Delta_1|
        &\leq\sqrt n\norm{\hbthe_j-\bthe_j^*}_1
                 \norm{S_n(\bbeta^*)}_\infty
          =O_p(\sqrt n\,r_{\theta,n}a_n)=o_p(1),\\
        \sqrt n|\Delta_2|
        &\leq\sqrt n\norm{\hbthe_j}_1
                 \norm{\tSigma-{\hSigma^{(j)}}}_{\max}
                 \norm{\hbbeta-\bbeta^*}_1\\
        &=O_p\{\sqrt n(\norm{\bthe_j^*}_1+r_{\theta,n})h_nr_{\beta,n}\}=o_p(1),\\
        \sqrt n|\Delta_3|
        &\leq\sqrt n\norm{{\hSigma^{(j)}}\hbthe_j-\be_j}_\infty
                 \norm{\hbbeta-\bbeta^*}_1
          =O_p(\sqrt n\,\mu_{j,n}r_{\beta,n})=o_p(1).
    \end{align*}
    Thus the exact decomposition in Section \ref{subsec: debiased estimator} gives
    \begin{equation}
    \label{asymptotic linear representation}
        \sqrt n(\hb_j-\beta_j^*)
        =-\frac1{\sqrt n}\sum_{i=1}^n\bthe_j^{*\top}\bphi_i+o_p(1).
    \end{equation}

    It remains to prove consistency of the feasible variance estimator. Write
    $$
        Z_{i,j}=\bthe_j^{*\top}\bphi_i,\quad
        \widehat Z_{i,j}=\hbthe_j^\top\hbphi_i.
    $$
    Standard maximal moment inequalities for sub-exponential variables, Condition \ref{con: ratios}, and the bounded weights in \eqref{population influence contribution} imply
    \begin{equation}
    \label{influence envelope rates}
        \frac1n\sum_{i=1}^n\norm{\bphi_i}_\infty^2
        =O_p(B_{\beta,n}^2\log^2p),
        \qquad
        \frac1n\sum_{i=1}^n\norm{\hA_i}_{\max}^2
        =O_p(\log^2p).
    \end{equation}
    Indeed, because $K$ is fixed, each coordinate of $\bphi_i$ is a finite sum of centered products of sub-Gaussian variables with sub-exponential norm bounded by $CB_{\beta,n}$. The expected squared maximum over $p$ coordinates is therefore at most $CB_{\beta,n}^2\log^2p$, and Markov's inequality gives the first rate. Applying the same argument to the $p^2$ entries of $\hA_i$ gives the second rate.

    Let $\bdelta_\theta=\hbthe_j-\bthe_j^*$ and $\bdelta_\beta=\hbbeta-\bbeta^*$. Direct expansion gives
    \begin{align*}
        \widehat Z_{i,j}-Z_{i,j}
        ={}&\bdelta_\theta^\top\bphi_i
        +\hbthe_j^\top\hA_i\bdelta_\beta\\
        &+\hbthe_j^\top\{(\hA_i-A_i)\bbeta^*-(\hq_i-q_i)\}.
    \end{align*}
    Uniformly in $i$, the regular observation proportions give
    $$
        \norm{\hA_i-A_i}_{\max}\leq C\norm{\tSigma-\Sigma}_{\max},
        \qquad
        \norm{\hq_i-q_i}_\infty\leq C\norm{\tmbC-\mbC}_\infty.
    $$
    Combining these inequalities with \eqref{beta rate used in debiasing}, \eqref{theta rate used in debiasing}, and \eqref{influence envelope rates}, we obtain
    \begin{align}
    \label{mean square influence stability}
        \left\{\frac1n\sum_{i=1}^n(\widehat Z_{i,j}-Z_{i,j})^2\right\}^{1/2}
        =O_p\big[&r_{\theta,n}B_{\beta,n}\log p
        +(\norm{\bthe_j^*}_1+r_{\theta,n})r_{\beta,n}\log p\nonumber\\
        &+(\norm{\bthe_j^*}_1+r_{\theta,n})a_n\big].
    \end{align}
    By \eqref{variance plugin rate conditions}, the right-hand side of \eqref{mean square influence stability} is $o_p(\norm{\bthe_j^*}_2)$. The lower bound \eqref{score variance lower bound} gives $\sqrt{v_{n,j}}\geq\sqrt{c\kappa_P}\norm{\bthe_j^*}_2$, and hence
    \begin{equation}
    \label{relative mean square influence stability}
        \left\{\frac1n\sum_{i=1}^n(\widehat Z_{i,j}-Z_{i,j})^2\right\}^{1/2}
        =o_p(\sqrt{v_{n,j}}).
    \end{equation}

    The fourth-moment calculation in the proof of Lemma \ref{le: normality} and conditional independence yield
    \begin{align*}
        &\Var\left\{\frac1{n v_{n,j}}\sum_{i=1}^n
        [Z_{i,j}^2-\E(Z_{i,j}^2\mid R)]\,\middle|\,R\right\}\\
        &\qquad\leq\frac1{n^2v_{n,j}^2}
        \sum_{i=1}^n\E(Z_{i,j}^4\mid R)
        \leq C\frac{\norm{\bthe_j^*}_1^4B_{\beta,n}^4}
        {n\norm{\bthe_j^*}_2^4}=o(1).
    \end{align*}
    Because $n^{-1}\sum_i\E(Z_{i,j}^2\mid R)=v_{n,j}$, it follows that
    \begin{equation}
    \label{relative oracle second moment}
        \frac{n^{-1}\sum_{i=1}^nZ_{i,j}^2}{v_{n,j}}
        \xrightarrow{p}1.
    \end{equation}

    To account for centering and the $n-1$ divisor in \eqref{estimator of variance}, define $\bar Z_j=n^{-1}\sum_iZ_{i,j}$ and $\overline{\widehat Z}_j=n^{-1}\sum_i\widehat Z_{i,j}$. Conditional on $R$, $\E(\bar Z_j\mid R)=0$ and $\Var(\bar Z_j\mid R)=v_{n,j}/n$, so $\bar Z_j^2/v_{n,j}=O_p(n^{-1})=o_p(1)$. Therefore \eqref{relative oracle second moment} implies
    \begin{equation}
    \label{relative centered oracle variance}
        \frac{1}{(n-1)v_{n,j}}\sum_{i=1}^n(Z_{i,j}-\bar Z_j)^2
        \xrightarrow{p}1.
    \end{equation}

    Centering is an orthogonal projection in $\mathbb R^n$, so \eqref{relative mean square influence stability} also gives
    \begin{equation*}
        \left\{\frac1n\sum_{i=1}^n
        \big[(\widehat Z_{i,j}-\overline{\widehat Z}_j)
             -(Z_{i,j}-\bar Z_j)\big]^2\right\}^{1/2}
        =o_p(\sqrt{v_{n,j}}).
    \end{equation*}
    By \eqref{relative centered oracle variance}, the oracle centered empirical norm is $O_p(\sqrt{v_{n,j}})$, and the triangle inequality gives the same order for its feasible counterpart. Applying Cauchy--Schwarz to the difference of their squared centered norms, dividing by $v_{n,j}$, and using $n/(n-1)\to1$ gives
    \begin{equation*}
        \frac{\hv_{n,j}}{v_{n,j}}-
        \frac{1}{(n-1)v_{n,j}}\sum_{i=1}^n(Z_{i,j}-\bar Z_j)^2
        =o_p(1).
    \end{equation*}
    Together with \eqref{relative centered oracle variance}, this proves $\hv_{n,j}/v_{n,j}\to1$ in probability. Combining this result, Lemma \ref{le: normality}, and \eqref{asymptotic linear representation} with Slutsky's theorem gives the desired result.
\end{proof}

\subsection{Proofs of the Perturbed-Moment and Estimation Results}

We first establish the response-covariate concentration bound using the argument in the original proof. As in that proof, the response is centered before constructing the empirical moments, so $\E Y=0$. Write $\sigma_y^2=\Var(Y)>0$ and, for a response initially held by $B_r$, define its standardized noise variance by
$$
    \eta_r^2\coloneqq\frac{\tau_{Y,r}^2}{\sigma_y^2}.
$$
Let $L'>0$ be a common sub-Gaussian parameter for $X_j/\sqrt{\sigma_{jj}}$, $(Y-\E Y)/\sigma_y$, and $\xi_i^r/\sigma_y$, uniformly over $j$ and $r$. Such a constant exists by Conditions \ref{con: subgaussian} and \ref{con: noises}. Define
$$
    C_r\coloneqq2\{72L'^2+8+2\eta_r^2\},
    \qquad C_{\max}\coloneqq\max_{0\leq r\leq K}C_r.
$$

\begin{lemma}[Extension of Lemma 1 in \citet{ravikumar2011high}]
\label{le: convergence of perturbed individual cross-covariance}
Let $c_j=\E(YX_j)$, and let $\widetilde c_j^{\mathrm{priv}}$ be the $j$th entry of $\tmbC^{\mathrm{priv}}$. If coordinate $j$ is held by $B_k$, then, conditional on a regular missingness pattern $R$,
$$
    \bbP\left(
    \left.\abs{\widetilde c_j^{\mathrm{priv}}-c_j}>\delta\,\right|R
    \right)
    \leq4\exp\left\{-\frac{n_k^y\delta^2}
    {\sigma_y^2\sigma_{jj}C_{\max}^2}\right\}
$$
for $0\leq\delta\leq\sigma_y\sqrt{\sigma_{jj}}C_{\max}/2$.
\end{lemma}

\subsubsection{Proof of Lemma \ref{le: convergence of perturbed individual cross-covariance}}

\begin{proof}
To prove the lemma, define
$$
    \widetilde{\mathcal A}_j(\delta)
    \coloneqq
    \left\{
    \left|\frac1{n_k^y}\sum_{i\in\mathcal I_k^y}
    y_i^{\mathrm{priv}}X_{ij}-c_j\right|>\delta
    \right\},
    \qquad
    \rho_j\coloneqq\frac{c_j}{\sigma_y\sqrt{\sigma_{jj}}}.
$$
Cauchy--Schwarz gives $|\rho_j|\leq1$. Set
$$
    \overline y_i^{\mathrm{priv}}=\frac{y_i^{\mathrm{priv}}-\E Y}{\sigma_y},
    \qquad
    \overline X_{ij}=\frac{X_{ij}}{\sqrt{\sigma_{jj}}},
    \qquad
    U_{ij}=\overline y_i^{\mathrm{priv}}+\overline X_{ij},
    \qquad
    V_{ij}=\overline y_i^{\mathrm{priv}}-\overline X_{ij}.
$$

\begin{lemma}
\label{le: perturbed decoupling}
For $i\in\mathcal I_k^y\cap\mathcal R_r$, the variables $U_{ij}$ and $V_{ij}$ are sub-Gaussian with parameter $3L'$. Moreover, for every $\delta>0$,
\begin{align*}
\bbP\{\widetilde{\mathcal A}_j(\delta)\mid R\}
\leq{}&
\bbP\left\{
\left.\left|\sum_{r=0}^K\sum_{i\in\mathcal I_k^y\cap\mathcal R_r}
(U_{ij}^2-\mu_{U,r})\right|
>\frac{2n_k^y\delta}{\sigma_y\sqrt{\sigma_{jj}}}
\,\right|R\right\}\\
&+\bbP\left\{
\left.\left|\sum_{r=0}^K\sum_{i\in\mathcal I_k^y\cap\mathcal R_r}
(V_{ij}^2-\mu_{V,r})\right|
>\frac{2n_k^y\delta}{\sigma_y\sqrt{\sigma_{jj}}}
\,\right|R\right\},
\end{align*}
where
$$
    \mu_{U,r}=2(1+\rho_j)+\eta_r^2,
    \qquad
    \mu_{V,r}=2(1-\rho_j)+\eta_r^2.
$$
\end{lemma}

\begin{proof}
For $i\in\mathcal R_r$, write
$\overline y_i^{\mathrm{priv}}=(Y_i-\E Y)/\sigma_y+\xi_i^r/\sigma_y$.
Independence of the response and perturbation implies that
$\overline y_i^{\mathrm{priv}}$ is sub-Gaussian with parameter
$\sqrt2L'$. By Cauchy--Schwarz,
\begin{align*}
    \E\{\exp(tU_{ij})\}
    &\leq
    \E\{\exp(2t\overline y_i^{\mathrm{priv}})\}^{1/2}
    \E\{\exp(2t\overline X_{ij})\}^{1/2}
    \leq\exp\left(\frac{9L'^2t^2}{2}\right).
\end{align*}
Thus, $U_{ij}$ is sub-Gaussian with parameter $3L'$; the same argument applies to $V_{ij}$.

For $i\in\mathcal I_k^y\cap\mathcal R_r$, direct calculation gives
$$
    \E(U_{ij}^2)=\mu_{U,r},
    \qquad
    \E(V_{ij}^2)=\mu_{V,r}.
$$
The standardized noise variance cancels from their difference, and hence
$$
    \overline y_i^{\mathrm{priv}}\overline X_{ij}-\rho_j
    =\frac14(U_{ij}^2-\mu_{U,r})
     -\frac14(V_{ij}^2-\mu_{V,r}).
$$
Summing this identity over $\mathcal I_k^y$, partitioning the indices by
$\mathcal R_0,\ldots,\mathcal R_K$, and applying the union bound prove the result.
\end{proof}

For $i\in\mathcal I_k^y\cap\mathcal R_r$, define
$Z_{ij}^U=U_{ij}^2-\mu_{U,r}$. We verify the Bernstein moment condition
\begin{equation}
\label{perturbed moment condition}
    \sup_{m\geq2}
    \left\{\frac{\E|Z_{ij}^U|^m}{m!}\right\}^{1/m}
    \leq\frac{C_{\max}}2.
\end{equation}
Using $(a+b)^m\leq2^m(a^m+b^m)$ gives
\begin{equation}
\label{perturbed initial moment bound}
    \left\{\frac{\E|Z_{ij}^U|^m}{m!}\right\}^{1/m}
    \leq2\left[
    \left\{\frac{\E|U_{ij}|^{2m}}{m!}\right\}^{1/m}
    +\frac{\mu_{U,r}}{(m!)^{1/m}}
    \right].
\end{equation}
Because $U_{ij}$ is sub-Gaussian with parameter $3L'$, Lemma 1.4 of
\citet{bickel2008regularized} and $m!\geq(m/e)^m$ give
$$
    \frac{\E|U_{ij}|^{2m}}{m!}
    \leq2^{m+1}(3L')^{2m}.
$$
Substitution into \eqref{perturbed initial moment bound} yields
\begin{align*}
    \sup_{m\geq2}
    \left\{\frac{\E|Z_{ij}^U|^m}{m!}\right\}^{1/m}
    &\leq36\sqrt2L'^2+\sqrt2\mu_{U,r}\\
    &\leq72L'^2+8+2\eta_r^2
    \leq\frac{C_{\max}}2,
\end{align*}
where $|\rho_j|\leq1$ is used in the second inequality. This verifies
\eqref{perturbed moment condition}. Bernstein's inequality therefore gives,
for $0<t\leq C_{\max}$,
$$
    \bbP\left\{
    \left.\left|\sum_{r=0}^K\sum_{i\in\mathcal I_k^y\cap\mathcal R_r}
    Z_{ij}^U\right|>n_k^yt\,\right|R\right\}
    \leq2\exp\left\{-\frac{n_k^yt^2}{4C_{\max}^2}\right\}.
$$
Taking $t=2\delta/(\sigma_y\sqrt{\sigma_{jj}})$ gives
$$
    \bbP\left\{
    \left.\left|\sum_{r=0}^K\sum_{i\in\mathcal I_k^y\cap\mathcal R_r}
    (U_{ij}^2-\mu_{U,r})\right|
    >\frac{2n_k^y\delta}{\sigma_y\sqrt{\sigma_{jj}}}
    \,\right|R\right\}
    \leq2\exp\left\{-\frac{n_k^y\delta^2}
    {\sigma_y^2\sigma_{jj}C_{\max}^2}\right\}.
$$
The same bound holds with $V_{ij}$ and $\mu_{V,r}$ in place of
$U_{ij}$ and $\mu_{U,r}$. Lemma \ref{le: perturbed decoupling} then proves the stated inequality.
\end{proof}

\subsubsection{Proof of Proposition \ref{the: perturbed moment concentration}}

\begin{proof}
For the covariance bound, within-user entries are unchanged, and the proof is the same as the corresponding part of Proposition \ref{the: sample cov}. For a cross-user entry, the two perturbed covariates remain sub-Gaussian under Condition \ref{con: noises} and all available samples have identical distribution. Therefore the results directly follow from Lemma 1 in \citet{ravikumar2011high} and Theorem 1 of \citet{yu2020optimal}.

It remains to prove the response-covariate bound. By Condition \ref{con: PD}, $\sigma_{jj}=1$. Set
$$
    C_C=\sqrt2\,\sigma_y C_{\max}.
$$
If $N_j^{xy}=n_k^y\geq8\log p$, then
$C_C\sqrt{(\log p)/N_j^{xy}}\leq\sigma_yC_{\max}/2$.
Lemma \ref{le: convergence of perturbed individual cross-covariance} consequently gives
$$
    \max_j\bbP\left(
    \left.|\widetilde c_j^{\mathrm{priv}}-c_j|
    \geq C_C\sqrt{\frac{\log p}{N_j^{xy}}}\,\right|R
    \right)
    \leq\frac4{p^2}.
$$
A union bound over $j=1,\ldots,p$ yields
$$
    \bbP\left(
    \left.\norm{\tmbC^{\mathrm{priv}}-\mbC}_\infty
    \geq C_C\sqrt{\frac{\log p}{\min_jN_j^{xy}}}\,\right|R
    \right)
    \leq\frac4p.
$$
Finally, the concentration bound for $\hSigma^{\mathrm{priv}}$ follows from the covariance bound and the same positive-definiteness adjustment used in the proof of Proposition \ref{the: sample cov}, after enlarging the constant to $C_H$.
\end{proof}

\subsubsection{Proof of Theorem \ref{cor: private global consistency}}

\begin{proof}
Theorem \ref{the: perturbed moment concentration} gives the same max-norm orders for the assembled covariance and response-covariate moments as those used in the proof of Proposition \ref{the: global consistency}. Their population targets remain $\Sigma$ and $\mbC$, respectively, so the Lasso basic inequality and restricted-eigenvalue argument apply with $(\hSigma,\tmbC,\cbbeta)$ replaced by $(\hSigma^{\mathrm{priv}},\tmbC^{\mathrm{priv}},\cbbeta^{\mathrm{priv}})$. This proves the two estimation rates. The final statement follows because Theorem \ref{the: algorithmic convergence} is deterministic conditional on a positive-definite quadratic matrix and exact block minimization.
\end{proof}

\section{Optional Release-Level Local Differential Privacy Calibration}
\label{app: release level ldp}

\textcolor{blue}{In this supplementary section, we complement the one-time perturbation in Section \ref{subsec: perturbation mechanism}, which reduces direct disclosure, with a formal privacy budget. \textcolor{red}{To obtain this formal privacy guarantee, we introduce an optional calibration for clipped one-time identity releases.}} 

Local differential privacy requires each owner to randomize a record before transmission and therefore does not rely on a trusted aggregator.

\begin{definition}[Local Differential Privacy \citep{kasiviswanathan2011can, duchi2013local}]
    A randomized mechanism $\mathcal{M}$ satisfies $(\epsilon,\delta)$-local differential privacy if, for any two possible individual records $v,v'$ and every measurable output set $S$,
    $$
        \mathbb{P}\{\mathcal{M}(v)\in S\}
        \leq e^\epsilon\mathbb{P}\{\mathcal{M}(v')\in S\}+\delta.
    $$
\end{definition}

We calibrate Gaussian perturbations to the $\ell_2$-sensitivity of each party-local identity release.

\begin{assumption}
\label{con: sensitivity}
    Before perturbation, each user clips its local records so that $\norm{\bx_i^k}_2\leq\Delta_{X,k}$ for every $i\in\mathcal I_k$, and each initially held response satisfies $|y_i|\leq\Delta_Y$.
\end{assumption}
Under Condition \ref{con: sensitivity}, the identity releases of $\bx_i^k$ and $y_i$ have sensitivities at most $2\Delta_{X,k}$ and $2\Delta_Y$, respectively.

\begin{theorem}[Release-Level Privacy of One-Time Perturbed Identity Releases]
\label{the: ldp}
    Suppose Condition \ref{con: sensitivity} holds. For user $B_k$, choose budgets $0<\epsilon_{X,k},\epsilon_{Y,k}<1$ and $0<\delta_{X,k},\delta_{Y,k}<1$ satisfying
    $$
        \epsilon_{X,k}+\epsilon_{Y,k}\leq\epsilon,
        \qquad \delta_{X,k}+\delta_{Y,k}\leq\delta,
    $$
    where the response budget is needed only for a local record whose response is initially held by $B_k$. Let the entries of $\Xi^k$ and the corresponding response perturbations be independent Gaussian variables with standard deviations $\tau_{X,k}$ and $\tau_{Y,k}$ satisfying
    $$
        \tau_{X,k}\geq
        \frac{2\Delta_{X,k}\sqrt{2\log(1.25/\delta_{X,k})}}
             {\epsilon_{X,k}},
        \qquad
        \tau_{Y,k}\geq
        \frac{2\Delta_Y\sqrt{2\log(1.25/\delta_{Y,k})}}
             {\epsilon_{Y,k}}.
    $$
    Then the joint one-time release of the perturbed covariate block and, when applicable, the perturbed response is $(\epsilon,\delta)$-locally differentially private with respect to each party-local record. Any subsequent computation based only on these releases, public information, and independent randomness inherits the same guarantee by post-processing. An equal allocation for a local record containing both a covariate block and a response uses $(\epsilon_{X,k},\delta_{X,k})=(\epsilon_{Y,k},\delta_{Y,k})=(\epsilon/2,\delta/2)$.
\end{theorem}

\textcolor{blue}{Theorem \ref{the: ldp} establishes an $(\epsilon,\delta)$-local differential privacy guarantee for each clipped party-local record under the stated Gaussian-noise calibration. The guarantee applies jointly to the released covariate block and response whenever both are present, and it extends by post-processing to every subsequent output computed from these releases, public information, and independent randomness. If an individual contributes multiple released blocks across parties, the corresponding person-level guarantee follows by composing the privacy budgets of those blocks.}

\begin{proof}
For a covariate block held by $B_k$, consider the one-time identity mechanism
$$
    \mathcal M_{X,k}(\bx_i^k)=\bx_i^k+\Xi_i^k.
$$
Its $\ell_2$-sensitivity is at most $2\Delta_{X,k}$. The Gaussian calibration in the theorem therefore makes $\mathcal M_{X,k}$ an $(\epsilon_{X,k},\delta_{X,k})$-locally differentially private mechanism \citep{dwork2014algorithmic}. Similarly,
$$
    \mathcal M_{Y,k}(y_i)=y_i+\xi_i^k,
    \qquad i\in\mathcal R_k,
$$
has sensitivity at most $2\Delta_Y$ and is $(\epsilon_{Y,k},\delta_{Y,k})$-locally differentially private under the stated calibration.

For a party-local record containing both components, sequential composition gives
$$
    (\epsilon_{X,k}+\epsilon_{Y,k},
     \delta_{X,k}+\delta_{Y,k})\text{-LDP}.
$$
If the party holds no response for that record, only the covariate budget is incurred. The budget inequalities in the theorem and monotonicity of differential privacy yield the claimed $(\epsilon,\delta)$ guarantee for the joint one-time release. The post-processing statement follows directly for any subsequent output that depends only on these perturbed releases, public information, and independent randomness.
\end{proof}

\section{Additional Simulation Results}
\label{app: additional simulations}

Figure \ref{fig:loss} shows that the ALB and ALB-P objectives stabilize within 10 sweeps in all four settings, \textcolor{red}{consistent with the geometric convergence result in Theorem} \ref{the: algorithmic convergence}.

Tables \ref{tab: MSE_34} and \ref{tab: CI_34_small} provide the corresponding results for Settings 3 and 4. At $\rho_n=1$, active-coordinate coverage is below the nominal level for all feasible methods under random signal locations. One explanation is that the small-sample regularization parameters $a_{1,j}$ and $a_{2,j}$ may remain far from one and induce bias. Even so, relative to CC-DL, ALB raises active-coordinate coverage from 0.459 to 0.753 in Setting 3 and from 0.597 to 0.814 in Setting 4, while also improving TPR. At $\rho_n=10$, ALB attains active-coordinate coverage of 0.942 and 0.932, inactive-coordinate coverage of 0.950 and 0.948, and TPR equal to one. These results support the asymptotic theory while illustrating the difficulty of the smaller-sample random-support setting.

\begin{figure}[!t]
    \centering
    \includegraphics[width=1\linewidth]{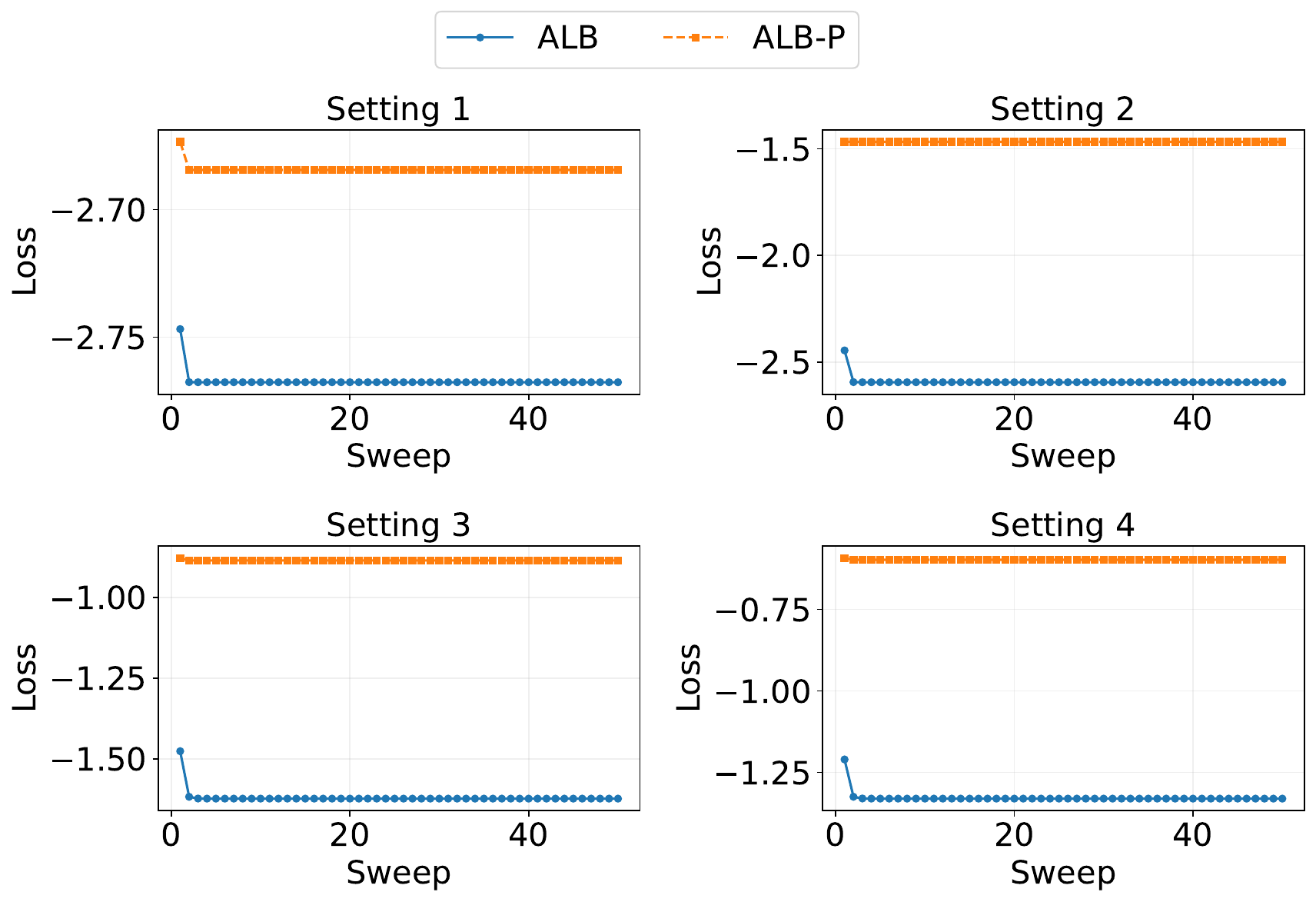}
    \caption{Training loss across ALB sweeps in Settings 1--4.}
    \label{fig:loss}
\end{figure}

\begin{table}[!t]
\centering
\caption{Prediction and coefficient-estimation performance in Settings 3 and 4. Values are means (SDs) over 200 replications.}
\label{tab: MSE_34}
\setlength{\tabcolsep}{2.8pt}
\renewcommand{\arraystretch}{0.75}
\begin{tabular}{@{}clcccc@{}}
\toprule
$\rho_n$ $(n)$ & Method
& \multicolumn{2}{c}{Setting 3}
& \multicolumn{2}{c}{Setting 4} \\
\cmidrule(lr){3-4}\cmidrule(lr){5-6}
& & MSE & Err & MSE & Err \\
\midrule
\multirow{5}{*}{$1\ (550)$}
& ALB & 2.007 (0.234) & 1.066 (0.111) & 2.024 (0.242) & 1.003 (0.107) \\
& ALB-P & 2.131 (0.247) & 1.153 (0.110) & 2.143 (0.267) & 1.062 (0.111) \\
& DISCOM & 2.007 (0.234) & 1.066 (0.111) & 2.024 (0.242) & 1.003 (0.107) \\
& CC-Lasso & 3.531 (0.643) & 1.651 (0.174) & 3.719 (0.643) & 1.643 (0.187) \\
& Oracle Lasso & 1.182 (0.089) & 0.449 (0.049) & 1.186 (0.082) & 0.424 (0.044) \\
\addlinespace[0.15em]
\midrule
\multirow{5}{*}{$10\ (3250)$}
& ALB & 1.125 (0.083) & 0.362 (0.071) & 1.134 (0.080) & 0.359 (0.051) \\
& ALB-P & 1.179 (0.084) & 0.452 (0.076) & 1.169 (0.081) & 0.406 (0.041) \\
& DISCOM & 1.125 (0.083) & 0.362 (0.071) & 1.134 (0.080) & 0.359 (0.051) \\
& CC-Lasso & 3.577 (0.624) & 1.652 (0.167) & 3.664 (0.653) & 1.626 (0.194) \\
& Oracle Lasso & 1.031 (0.067) & 0.171 (0.017) & 1.028 (0.065) & 0.162 (0.015) \\
\bottomrule
\end{tabular}
\end{table}

\begin{table}[!t]
\centering
\caption{Performance of the 95\% coordinatewise confidence intervals in Settings 3 and 4. \textcolor{blue}{Results are based on 200 replications.}}
\label{tab: CI_34_small}

{\small
\setlength{\tabcolsep}{3pt}
\renewcommand{\arraystretch}{0.8}
\resizebox{\textwidth}{!}{
\begin{tabular}{lcccccccccc}
\toprule
& \multicolumn{2}{c}{AvgCov$(\cS)$}
& \multicolumn{2}{c}{AvgLen$(\cS)$}
& \multicolumn{2}{c}{AvgCov$(\cS^c_{100})$}
& \multicolumn{2}{c}{AvgLen$(\cS^c_{100})$}
& \multicolumn{2}{c}{TPR} \\
\cmidrule(lr){2-3}
\cmidrule(lr){4-5}
\cmidrule(lr){6-7}
\cmidrule(lr){8-9}
\cmidrule(lr){10-11}
Method
& $\rho_n=1$ & $\rho_n=10$
& $\rho_n=1$ & $\rho_n=10$
& $\rho_n=1$ & $\rho_n=10$
& $\rho_n=1$ & $\rho_n=10$
& $\rho_n=1$ & $\rho_n=10$ \\
\midrule

\multicolumn{11}{c}{\textit{Setting 3}} \\
\addlinespace[0.15em]
ALB
& 0.753 & 0.942
& 0.331 & 0.190
& 0.942 & 0.950
& 0.327 & 0.189
& 0.988 & 1.000 \\
ALB-P
& 0.811 & 0.949
& 0.431 & 0.258
& 0.945 & 0.950
& 0.428 & 0.258
& 0.915 & 1.000 \\
ALB-CP
& 0.753 & 0.942
& 0.331 & 0.190
& 0.942 & 0.950
& 0.327 & 0.189
& 0.988 & 1.000 \\
CC-DL
& 0.459 & 0.529
& 0.410 & 0.444
& 0.935 & 0.937
& 0.409 & 0.443
& 0.692 & 0.670 \\
Oracle-DL
& 0.906 & 0.942
& 0.203 & 0.094
& 0.947 & 0.950
& 0.203 & 0.094
& 1.000 & 1.000 \\
ALB-OP
& 0.885 & 0.950
& 0.458 & 0.207
& 0.948 & 0.950
& 0.455 & 0.208
& 0.957 & 1.000 \\

\addlinespace[0.60em]
\multicolumn{11}{c}{\textit{Setting 4}} \\
\addlinespace[0.15em]
ALB
& 0.814 & 0.932
& 0.298 & 0.136
& 0.948 & 0.948
& 0.295 & 0.136
& 1.000 & 1.000 \\
ALB-P
& 0.825 & 0.942
& 0.344 & 0.165
& 0.949 & 0.953
& 0.342 & 0.164
& 0.995 & 1.000 \\
ALB-CP
& 0.814 & 0.932
& 0.298 & 0.136
& 0.948 & 0.948
& 0.295 & 0.136
& 1.000 & 1.000 \\
CC-DL
& 0.597 & 0.577
& 0.455 & 0.434
& 0.943 & 0.944
& 0.455 & 0.433
& 0.736 & 0.739 \\
Oracle-DL
& 0.902 & 0.950
& 0.160 & 0.069
& 0.948 & 0.948
& 0.160 & 0.069
& 1.000 & 1.000 \\
ALB-OP
& 0.861 & 0.943
& 0.329 & 0.145
& 0.949 & 0.948
& 0.325 & 0.145
& 1.000 & 1.000 \\
\bottomrule
\end{tabular}
}}
\end{table}

\clearpage
We further examine a design in which no record contains all
three covariate blocks. Except for the missingness pattern, the covariance structures, coefficient vectors,
dimensions, tuning procedures, and test-sample construction remain as in Settings 1--4 in Section \ref{subsec: MSE simulation}. The 100 complete records in the original design are removed. For each pair of users, there are $n_{\mathrm{pair}}\in\{100,1000\}$ records for which the response and exactly those two users' covariate blocks are observed. In addition, each user holds 50 single-user records without a response. Hence, the total sample sizes are $n=3n_{\mathrm{pair}}+150\in\{450,3150\}$. In particular, every cross-block covariance can be estimated from its pairwise overlap, and every response--covariate moment remains estimable despite the absence of complete cases. Consequently, CC-Lasso and CC-DL are omitted. All results below are averaged over $M=200$ replications.

\begin{table}[H]
\centering
\caption{Prediction and coefficient-estimation performance without complete cases in Settings 1 and 2. Values are means (SDs) over 200 replications.}
\label{tab: MSE_nocomp_12}
{
\setlength{\tabcolsep}{2.8pt}
\renewcommand{\arraystretch}{0.75}
\begin{tabular}{@{}clcccc@{}}
\toprule
$n_{\mathrm{pair}}$ $(n)$ & Method
& \multicolumn{2}{c}{Setting 1}
& \multicolumn{2}{c}{Setting 2} \\
\cmidrule(lr){3-4}\cmidrule(lr){5-6}
& & MSE & Err & MSE & Err \\
\midrule
\multirow{4}{*}{$100\ (450)$}
& ALB & 1.706 (0.324) & 0.842 (0.130) & 2.321 (0.532) & 0.988 (0.168) \\
& ALB-P & 1.888 (0.402) & 0.932 (0.130) & 2.463 (0.635) & 1.023 (0.175) \\
& DISCOM & 1.706 (0.324) & 0.842 (0.130) & 2.321 (0.532) & 0.988 (0.168) \\
& Oracle Lasso & 1.131 (0.081) & 0.357 (0.058) & 1.172 (0.090) & 0.358 (0.043) \\
\addlinespace[0.15em]
\midrule
\multirow{4}{*}{$1000\ (3150)$}
& ALB & 1.084 (0.072) & 0.345 (0.064) & 1.041 (0.064) & 0.178 (0.031) \\
& ALB-P & 1.120 (0.076) & 0.421 (0.074) & 1.048 (0.064) & 0.196 (0.034) \\
& DISCOM & 1.084 (0.072) & 0.345 (0.064) & 1.041 (0.064) & 0.178 (0.031) \\
& Oracle Lasso & 1.024 (0.062) & 0.128 (0.019) & 1.030 (0.064) & 0.128 (0.014) \\
\bottomrule
\end{tabular}
}
\end{table}

\begin{table}[H]
\centering
\caption{Prediction and coefficient-estimation performance without complete cases in Settings 3 and 4. Values are means (SDs) over 200 replications.}
\label{tab: MSE_nocomp_34}
{
\setlength{\tabcolsep}{2.8pt}
\renewcommand{\arraystretch}{0.75}
\begin{tabular}{@{}clcccc@{}}
\toprule
$n_{\mathrm{pair}}$ $(n)$ & Method
& \multicolumn{2}{c}{Setting 3}
& \multicolumn{2}{c}{Setting 4} \\
\cmidrule(lr){3-4}\cmidrule(lr){5-6}
& & MSE & Err & MSE & Err \\
\midrule
\multirow{4}{*}{$100\ (450)$}
& ALB & 2.499 (0.388) & 1.303 (0.139) & 2.599 (0.551) & 1.252 (0.185) \\
& ALB-P & 2.927 (0.570) & 1.455 (0.171) & 3.089 (0.685) & 1.429 (0.213) \\
& DISCOM & 2.499 (0.388) & 1.303 (0.139) & 2.599 (0.551) & 1.252 (0.185) \\
& Oracle Lasso & 1.227 (0.094) & 0.501 (0.053) & 1.228 (0.085) & 0.478 (0.050) \\
\addlinespace[0.15em]
\midrule
\multirow{4}{*}{$1000\ (3150)$}
& ALB & 1.147 (0.081) & 0.391 (0.061) & 1.164 (0.081) & 0.392 (0.045) \\
& ALB-P & 1.188 (0.084) & 0.459 (0.073) & 1.185 (0.080) & 0.417 (0.043) \\
& DISCOM & 1.147 (0.081) & 0.391 (0.061) & 1.164 (0.081) & 0.392 (0.045) \\
& Oracle Lasso & 1.036 (0.064) & 0.175 (0.020) & 1.037 (0.064) & 0.167 (0.017) \\
\bottomrule
\end{tabular}
}
\end{table}

Tables \ref{tab: MSE_nocomp_12} and
\ref{tab: MSE_nocomp_34} show that ALB remains well defined in the absence of complete records and matches DISCOM to the reported precision. Increasing each pairwise overlap from 100 to 1000 substantially reduces both errors across all four settings.
At the larger overlap, the MSEs of ALB range from 1.041 to 1.164, compared with 1.024 to 1.037 for Oracle Lasso. Although ALB-P incurs a moderate perturbation cost, its performance also improves substantially as the overlap increases. Together, these results show that the information needed for estimation can be assembled from pairwise overlaps rather than complete observations.

Tables \ref{tab: CI_nocomp_12} and \ref{tab: CI_nocomp_34} show that inference also remains feasible without complete cases. ALB and ALB-CP agree to the reported digits, confirming that the distributed precision-column computation closely reproduces its centralized counterpart. When $n_{\mathrm{pair}}=1000$, the active-set coverage of ALB is 0.917, 0.946, 0.940, and 0.932 in Settings 1--4, respectively; inactive-set coverage ranges from 0.948 to 0.950, and TPR equals one throughout. The active coverage is lower and the intervals are wider when $n_{\mathrm{pair}}=100$, particularly for random signal locations, because small overlaps make cross-block and precision-column quantities more difficult to estimate. ALB-P generally produces wider intervals and somewhat higher active coverage, while its TPR can be lower at the smaller overlap. Overall, the larger-overlap results support the pairwise-overlap asymptotic theory and show that complete records are not required for either point estimation or coordinatewise inference.

\begin{table}[H]
\centering
\caption{Performance of the 95\% coordinatewise confidence intervals without complete cases in Settings 1 and 2. Values are based on 200 replications.}
\label{tab: CI_nocomp_12}
{\small
\setlength{\tabcolsep}{3pt}
\renewcommand{\arraystretch}{0.8}
\resizebox{\textwidth}{!}{
\begin{tabular}{lcccccccccc}
\toprule
& \multicolumn{2}{c}{AvgCov$(\cS)$}
& \multicolumn{2}{c}{AvgLen$(\cS)$}
& \multicolumn{2}{c}{AvgCov$(\cS^c_{100})$}
& \multicolumn{2}{c}{AvgLen$(\cS^c_{100})$}
& \multicolumn{2}{c}{TPR} \\
\cmidrule(lr){2-3}\cmidrule(lr){4-5}\cmidrule(lr){6-7}
\cmidrule(lr){8-9}\cmidrule(lr){10-11}
Method
& $100$ & $1000$ & $100$ & $1000$ & $100$ & $1000$
& $100$ & $1000$ & $100$ & $1000$ \\
\midrule
\multicolumn{11}{c}{\textit{Setting 1}} \\
\addlinespace[0.15em]
ALB & 0.906 & 0.917 & 0.764 & 0.311 & 0.948 & 0.950 & 0.770 & 0.320 & 0.640 & 1.000 \\
ALB-P & 0.917 & 0.924 & 1.061 & 0.437 & 0.946 & 0.945 & 1.094 & 0.454 & 0.386 & 0.984 \\
ALB-CP & 0.906 & 0.917 & 0.764 & 0.311 & 0.948 & 0.950 & 0.770 & 0.320 & 0.640 & 1.000 \\
Oracle-DL & 0.872 & 0.928 & 0.226 & 0.095 & 0.951 & 0.952 & 0.227 & 0.096 & 1.000 & 1.000 \\
ALB-OP & 0.943 & 0.944 & 1.151 & 0.359 & 0.946 & 0.950 & 1.115 & 0.349 & 0.391 & 0.998 \\
\addlinespace[0.60em]
\multicolumn{11}{c}{\textit{Setting 2}} \\
\addlinespace[0.15em]
ALB & 0.894 & 0.946 & 0.520 & 0.191 & 0.946 & 0.950 & 0.510 & 0.190 & 0.942 & 1.000 \\
ALB-P & 0.918 & 0.950 & 0.613 & 0.235 & 0.947 & 0.949 & 0.603 & 0.235 & 0.892 & 1.000 \\
ALB-CP & 0.894 & 0.946 & 0.520 & 0.191 & 0.946 & 0.950 & 0.510 & 0.190 & 0.942 & 1.000 \\
Oracle-DL & 0.891 & 0.937 & 0.177 & 0.071 & 0.949 & 0.952 & 0.177 & 0.071 & 1.000 & 1.000 \\
ALB-OP & 0.910 & 0.954 & 0.571 & 0.203 & 0.949 & 0.949 & 0.567 & 0.203 & 0.849 & 1.000 \\
\bottomrule
\end{tabular}
}}
\end{table}

\begin{table}[H]
\centering
\caption{Performance of the 95\% coordinatewise confidence intervals without complete cases in Settings 3 and 4. Values are based on 200 replications.}
\label{tab: CI_nocomp_34}
{\small
\setlength{\tabcolsep}{3pt}
\renewcommand{\arraystretch}{0.8}
\resizebox{\textwidth}{!}{
\begin{tabular}{lcccccccccc}
\toprule
& \multicolumn{2}{c}{AvgCov$(\cS)$}
& \multicolumn{2}{c}{AvgLen$(\cS)$}
& \multicolumn{2}{c}{AvgCov$(\cS^c_{100})$}
& \multicolumn{2}{c}{AvgLen$(\cS^c_{100})$}
& \multicolumn{2}{c}{TPR} \\
\cmidrule(lr){2-3}\cmidrule(lr){4-5}\cmidrule(lr){6-7}
\cmidrule(lr){8-9}\cmidrule(lr){10-11}
Method
& $100$ & $1000$ & $100$ & $1000$ & $100$ & $1000$
& $100$ & $1000$ & $100$ & $1000$ \\
\midrule
\multicolumn{11}{c}{\textit{Setting 3}} \\
\addlinespace[0.15em]
ALB & 0.646 & 0.940 & 0.407 & 0.199 & 0.939 & 0.949 & 0.405 & 0.198 & 0.862 & 1.000 \\
ALB-P & 0.756 & 0.953 & 0.571 & 0.271 & 0.945 & 0.950 & 0.568 & 0.271 & 0.630 & 1.000 \\
ALB-CP & 0.646 & 0.940 & 0.407 & 0.199 & 0.939 & 0.949 & 0.405 & 0.198 & 0.862 & 1.000 \\
Oracle-DL & 0.894 & 0.951 & 0.220 & 0.095 & 0.946 & 0.946 & 0.220 & 0.095 & 1.000 & 1.000 \\
ALB-OP & 0.848 & 0.949 & 0.589 & 0.217 & 0.946 & 0.950 & 0.590 & 0.218 & 0.721 & 1.000 \\
\addlinespace[0.60em]
\multicolumn{11}{c}{\textit{Setting 4}} \\
\addlinespace[0.15em]
ALB & 0.699 & 0.932 & 0.371 & 0.143 & 0.943 & 0.948 & 0.367 & 0.142 & 0.947 & 1.000 \\
ALB-P & 0.733 & 0.936 & 0.464 & 0.173 & 0.943 & 0.951 & 0.461 & 0.173 & 0.841 & 1.000 \\
ALB-CP & 0.699 & 0.932 & 0.371 & 0.143 & 0.943 & 0.948 & 0.367 & 0.142 & 0.947 & 1.000 \\
Oracle-DL & 0.883 & 0.945 & 0.174 & 0.071 & 0.948 & 0.949 & 0.174 & 0.071 & 1.000 & 1.000 \\
ALB-OP & 0.775 & 0.944 & 0.414 & 0.152 & 0.947 & 0.949 & 0.411 & 0.152 & 0.935 & 1.000 \\
\bottomrule
\end{tabular}
}}
\end{table}

\par\vspace*{0.5\baselineskip}\kern0pt
\section{ADNI Data and Preprocessing}
\label{app: ADNI modalities}

For the analysis, we assign the modalities to three users. User $B_0$ holds five UPENN CSF biomarkers: amyloid beta ($\text{A}\beta_{42}$), total tau, phosphorylated tau, and two associated ratios. User $B_1$ holds 116 cortical-thickness and subcortical-volume measures from the UCSF cross-sectional FreeSurfer pipeline, whereas $B_2$ holds 98 standardized uptake value ratios from the UC Berkeley PET pipeline. Matching participants and visits yields 219 covariates. The response, held by ADNI clinical sites, is the 30-point Mini-Mental State Examination (MMSE) score \citep{mf1975mini,tombaugh1992mini}.

We use month-48 observations from the second phase of ADNI (ADNI-2). After preprocessing, 92 records are complete. The incomplete sample contains 15 CSF--MRI, 65 CSF--PET, and 99 MRI--PET records with MMSE, together with one CSF-only, two MRI-only, and three PET-only records without MMSE. In each of 200 train-test splits, half of the complete records form the test set; the remaining complete and all incomplete records form the training set. Tuning and perturbation follow Section \ref{subsec: MSE simulation}, and normalization uses training-set statistics to prevent information leakage.

\section{Implementation Details}
\label{app: implementation details}

\subsection[Point-Estimation Tuning]{Point-Estimation Tuning}
\label{app: covariance regularization}

For point estimation, the two covariance-regularization parameters are selected jointly over a finite grid $\mathcal A\subset[0,1]$. Define
\begin{equation*}
    \hSigma(a_1,a_2)
    =a_1\tSigma_I+a_2\tSigma_C
      +(1-a_1)\frac{c_{tr}}{p}\mathrm I_p,
    \qquad (a_1,a_2)\in\mathcal A^2.
\end{equation*}
The numerical studies use 11 equally spaced values in $[0,1]$ for each parameter. Because the grid contains $(0,0)$, it includes the positive-definite diagonal-ridge endpoint $\hSigma(0,0)=(c_{tr}/p)\mathrm I_p$ whenever $c_{tr}>0$.

We form pattern-stratified folds so that each training and validation split contains observations from every realized block-missing pattern. Before cross-validation, each pair $(a_1,a_2)\in\mathcal A^2$ is screened for positive definiteness on the full sample and every training fold. Since no user assembles the covariance matrix, the screening uses the distributed shifted Lanczos procedure in Section \ref{app: shifted Lanczos iteration}. A pair is retained only if every estimated smallest eigenvalue exceeds $10^{-8}$. ALB and ALB-P are screened separately because their cross-user moments differ.

For each retained pair and coefficient penalty $\lambda\in\Lambda$, ALB is fitted without validation fold $v$, yielding $\hbbeta_{-v}(a_1,a_2,\lambda)$. Its held-out quadratic loss is
\begin{equation*}
    L_v(a_1,a_2,\lambda)
    =\frac12\hbbeta_{-v}(a_1,a_2,\lambda)^{\top}
       \tSigma_v\hbbeta_{-v}(a_1,a_2,\lambda)
      -\tmbC_v^{\top}\hbbeta_{-v}(a_1,a_2,\lambda),
\end{equation*}
where $\tSigma_v$ and $\tmbC_v$ are the unregularized available-case moments from fold $v$. We select
\begin{equation*}
    (\widehat a_1,\widehat a_2,\widehat\lambda)
    =\argmin_{(a_1,a_2)\in\mathcal F,\,\lambda\in\Lambda}
      \frac{1}{V}\sum_{v=1}^{V}L_v(a_1,a_2,\lambda),
\end{equation*}
where $\mathcal F$ is the Lanczos-screened set and $V$ is the number of folds. ALB is then refitted on all available observations using the selected triple. ALB-P follows the same procedure with perturbed cross-user and response moments and unperturbed within-user covariance blocks.

Table \ref{tab: eigenvalue} reports the smallest eigenvalues of the covariance matrices selected for precision-column estimation. The resulting matrices are positive definite in all simulation settings.

\begin{table}[!t]
    \centering
    \caption{Smallest eigenvalues of the covariance matrices selected for precision-column estimation.}
    \label{tab: eigenvalue}
    {\small
    \renewcommand{\arraystretch}{0.9}
    \begin{tabular}{ccrrrrrr}
    \toprule
    & & \multicolumn{3}{c}{ALB} & \multicolumn{3}{c}{ALB-P} \\
    \cmidrule(lr){3-5}\cmidrule(lr){6-8}
    Setting & $\rho_n$ & Minimum & Mean & SD & Minimum & Mean & SD \\
    \midrule
    \multirow{2}{*}{1} & 1  & 0.01 & 0.101 & 0.001 & 0.061 & 0.078 & 0.006 \\
                       & 10 & 0.087 & 0.120 & 0.012 & 0.063 & 0.111 & 0.019 \\
    \midrule
    \multirow{2}{*}{2} & 1  & 0.36 & 0.713 & 0.091 & 0.284 & 0.634 & 0.111 \\
                       & 10 & 0.134 & 0.285 & 0.043 & 0.145 & 0.272 & 0.053 \\
    \midrule
    \multirow{2}{*}{3} & 1  & 0.099 & 0.101 & 0.001 & 0.061 & 0.078 & 0.006 \\
                       & 10 & 0.086 & 0.120 & 0.012 & 0.065 & 0.107 & 0.022 \\
    \midrule
    \multirow{2}{*}{4} & 1  & 0.501 & 0.716 & 0.088 & 0.299 & 0.637 & 0.122 \\
                       & 10 & 0.137 & 0.277 & 0.044 & 0.165 & 0.265 & 0.056 \\
    \bottomrule
    \end{tabular}
    }
\end{table}

\subsection{Precision-Column Tuning}
\label{app: precision covariance selection}

Precision-column estimation uses inference-specific covariance regularization, which may differ from that used for point estimation. We consider a computationally efficient strategy that shares one parameter pair across coordinates. Let $\mathcal F_{\theta}$ contain the pairs retained by the full-sample and training-fold Lanczos checks in Section \ref{app: covariance regularization}. For shared tuning, $\mathcal J_{\mathrm{tun}}$ is formed by sampling the same prescribed number of coordinates without replacement from each covariate block; the simulations use five per block.

For $j\in\mathcal J_{\mathrm{tun}}$ and validation fold $v$, define the training-fold estimate
\begin{equation*}
{
    \hbthe_{j,-v}(a_1,a_2,\lambda_j)
    =\argmin_{\bthe\in\mathbb R^p}
    \left\{
      \frac12\bthe^\top\hSigma_{-v}(a_1,a_2)\bthe
      -\be_j^\top\bthe+\lambda_j\norm{\bthe}_1
    \right\},
}
\end{equation*}
where $\hSigma_{-v}(a_1,a_2)$ is constructed from the training-fold available-case moments. The corresponding held-out loss is
\begin{equation*}
{
    L_{j,v}^{\mathrm{prec}}(a_1,a_2,\lambda_j)
    =\frac12\hbthe_{j,-v}(a_1,a_2,\lambda_j)^\top
      \tSigma_v\hbthe_{j,-v}(a_1,a_2,\lambda_j)
      -\be_j^\top\hbthe_{j,-v}(a_1,a_2,\lambda_j),
}
\end{equation*}
where $\tSigma_v$ is the unregularized available-case covariance from validation fold $v$. For each covariance pair, the precision penalty is profiled over $\Lambda_{\theta}$ as
\begin{equation*}
{
    \widehat\lambda_j(a_1,a_2)
    =\argmin_{\lambda_j\in\Lambda_{\theta}}
      \frac1V\sum_{v=1}^V
      L_{j,v}^{\mathrm{prec}}(a_1,a_2,\lambda_j).
}
\end{equation*}
We then average the optimized losses over the tuning coordinates and select
\begin{equation*}
{
    (\widehat a_{1,\theta},\widehat a_{2,\theta})
    =\argmin_{(a_1,a_2)\in\mathcal F_{\theta}}
      \frac{1}{|\mathcal J_{\mathrm{tun}}|}
      \sum_{j\in\mathcal J_{\mathrm{tun}}}
      \left\{
      \frac1V\sum_{v=1}^V
      L_{j,v}^{\mathrm{prec}}
      \bigl(a_1,a_2,\widehat\lambda_j(a_1,a_2)\bigr)
      \right\}.
}
\end{equation*}
The selected pair $(\widehat a_{1,\theta},\widehat a_{2,\theta})$ is shared across precision columns. For each reported coordinate $j$, we hold this pair fixed, select $\widehat\lambda_j$ by the same foldwise criterion, and refit \eqref{global inverse column} on all available observations. Ties between covariance pairs are resolved in favor of the larger $a_1$ and then the larger $a_2$. ALB-P performs this selection separately using perturbed cross-user moments, while ALB-CP uses the corresponding centralized foldwise loss. Algorithm \ref{alg: shared precision selection} summarizes the procedure.

\begin{algorithm}[!t]
\caption{Shared Tuning for Precision Estimation}
\label{alg: shared precision selection}{\renewcommand{\baselinestretch}{0.93}\selectfont
\begin{algorithmic}[1]
\Statex \textbf{Input:} Pattern-stratified folds $\{\mathcal I_v\}_{v=1}^V$, feasible set $\mathcal F_{\theta}$, penalty grid $\Lambda_{\theta}$, tuning coordinates $\mathcal J_{\mathrm{tun}}$, and reported coordinates $\mathcal J_{\mathrm{re}}$.
\Statex \textbf{Output:} Shared pair $(\widehat a_{1,\theta},\widehat a_{2,\theta})$ and estimates $\{\hbthe_j:j\in\mathcal J_{\mathrm{re}}\}$.
\For{each $(a_1,a_2)\in\mathcal F_{\theta}$}
    \For{each $j\in\mathcal J_{\mathrm{tun}}$}
        \State Evaluate $L_{j,v}^{\mathrm{prec}}(a_1,a_2,\lambda)$ over all folds and $\lambda\in\Lambda_{\theta}$.
        \State Choose $\widehat\lambda_j(a_1,a_2)$ by the smallest mean held-out loss.
    \EndFor
    \State Average the optimized losses over $\mathcal J_{\mathrm{tun}}$.
\EndFor
\State Select the pair with the smallest coordinate-averaged loss.
\For{each $j\in\mathcal J_{\mathrm{re}}$}
    \State Fix the selected pair and choose $\widehat\lambda_j$ by mean held-out loss.
    \State Refit \eqref{global inverse column} on all available observations to obtain $\hbthe_j$.
\EndFor
\State \Return $(\widehat a_{1,\theta},\widehat a_{2,\theta})$ and $\{\hbthe_j:j\in\mathcal J_{\mathrm{re}}\}$.
\end{algorithmic}}
\end{algorithm}

\subsection{Distributed Positive-Definiteness Screening}
\label{app: shifted Lanczos iteration}
{\color{blue}
The screening in Section \ref{app: covariance regularization} requires the smallest eigenvalue of each candidate covariance matrix on the full sample and on every training fold, although no user assembles any such matrix. Fix one candidate $(a_1,a_2)$ and one of these data splits, and let $\tSigma_{C,\times}$ denote the cross-user matrix whose $(k,l)$ block, for $k\ne l$, is
$$
    (\tSigma_{C,\times})_{kl}
    =\frac{1}{n_{kl}}
      \bX_{\times}^{k,l\top}\bX_{\times}^{l,k},
$$
with zero diagonal blocks. Here $\bX_{\times}^k=\bX^k$ for ALB and $\bX_{\times}^k=\bX_{\mathrm{priv}}^k$ for ALB-P, with all matrices and denominators restricted to the fixed split. Thus, define the generic matrix screened by Algorithm \ref{alg:shifted_lanczos} as
$$
    \hSigma_{\times}(a_1,a_2)
    =a_1\tSigma_I+a_2\tSigma_{C,\times}+d_0\mathrm I_p,
    \qquad d_0=\frac{(1-a_1)c_{tr}}{p}.
$$
It equals $\hSigma(a_1,a_2)$ for ALB and $\hSigma^{\mathrm{priv}}(a_1,a_2)$ for ALB-P. The same construction is repeated for every candidate and split.

For a block vector $\bv=(\bv_0^\top,\ldots,\bv_K^\top)^\top$, user $B_k$ computes the $k$th block of $\hSigma_{\times}\bv$ as
\begin{equation}
\label{distributed hSigma matvec}
    \mathcal H_{\times,k}(\bv)
    =\frac{a_1}{n_k}\bX^{k\top}\bX^k\bv_k
    +a_2\sum_{l\ne k}\frac{1}{n_{kl}}
      \bX_{\times}^{k,l\top}\bX_{\times}^{l,k}\bv_l
    +d_0\bv_k\textcolor{red}{.}
\end{equation}
The within-user product in \eqref{distributed hSigma matvec} always uses the unperturbed local matrix $\bX^k$; only the cross-user products use $\bX_{\mathrm{priv}}^k$ in ALB-P. To evaluate the latter products, $B_l$ sends $\boldsymbol\zeta_{l\to k}=\bX_{\times}^{l,k}\bv_l\in\mathbb R^{n_{kl}}$ to $B_k$, which then uses $\bX_{\times}^{k,l\top}\boldsymbol\zeta_{l\to k}$. Block inner products and norms are computed as $\langle\bx,\by\rangle_{\mathrm{blk}}=\sum_{k=0}^K\bx_k^\top\by_k$ and $\norm{\bx}_{\mathrm{blk}}=\langle\bx,\bx\rangle_{\mathrm{blk}}^{1/2}$ by aggregating user-specific scalars.

The algorithm applies Lanczos iteration to the shifted matrix
$$
    \mathbf B=\gamma\mathrm I_p-\hSigma_{\times}.
$$
Its largest eigenvalue is $\lambda_{\max}(\mathbf B)=\gamma-\lambda_{\min}(\hSigma_{\times})$, so a largest Ritz value $\vartheta$ of $\mathbf B$ gives the estimate $\gamma-\vartheta$ of $\lambda_{\min}(\hSigma_{\times})$. The shift is computed without first estimating $\lambda_{\max}(\hSigma_{\times})$. Specifically, define the nonnegative $(K+1)\times(K+1)$ matrix $\mathbf G$ by
\begin{align*}
    G_{kk}&=\frac{|a_1|}{n_k}\norm{\bX^k}_{\mathrm F}^2+|d_0|,\\
    G_{kl}&=\frac{|a_2|}{n_{kl}}
       \norm{\bX_{\times}^{k,l}}_{\mathrm F}
       \norm{\bX_{\times}^{l,k}}_{\mathrm F},\qquad k\ne l.
\end{align*}
The block-norm inequality gives $\norm{\hSigma_{\times}}_2\le\norm{\mathbf G}_2$. We therefore set
$$
    \gamma=(1+\delta)\norm{\mathbf G}_2+\epsilon_\gamma,
$$
where the implementation uses $\delta=10^{-3}$ and $\epsilon_\gamma=10^{-6}$. Only scalar Frobenius-norm summaries are needed to form $\mathbf G$.

For each random start, Algorithm \ref{alg:shifted_lanczos} draws and jointly normalizes a distributed vector $\mathbf q$. At step $t$, it evaluates $\bu_k=\mathcal H_{\times,k}(\mathbf q)$ and forms the $k$th block of the Lanczos residual as
$$
    \bw_k=\gamma\mathbf q_k-\bu_k
           -\omega_{\mathrm{prev}}\mathbf q_{\mathrm{prev},k}.
$$
It then sets $\alpha_t=\langle\mathbf q,\bw\rangle_{\mathrm{blk}}$, removes $\alpha_t\mathbf q$, performs two passes of full reorthogonalization against the distributed basis $\mathcal V$, and sets $\omega_{t+1}=\norm{\bw}_{\mathrm{blk}}$. The resulting projection is
$$
    \mathbf T_t
    =\operatorname{tridiag}(\omega_2,\ldots,\omega_t;
                             \alpha_1,\ldots,\alpha_t).
$$
If $(\vartheta_t,\bz_t)$ is the largest eigenpair of $\mathbf T_t$, then $r_t=\omega_{t+1}|(\bz_t)_t|$ is its Ritz residual norm and $\hat\lambda_{\min}^{(t)}=\gamma-\vartheta_t$ is the current smallest-eigenvalue estimate. A start terminates when $r_t\leq\tau\max\{1,|\hat\lambda_{\min}^{(t)}|\}$; otherwise, $\mathbf q_{\mathrm{prev}}$, $\mathbf q$, and $\omega_{\mathrm{prev}}$ are updated exactly as in Algorithm \ref{alg:shifted_lanczos}. Across all $R$ starts and at most $m=\min\{M,p\}$ steps per start, the algorithm retains $\vartheta_{\mathrm{best}}$, the largest computed Ritz value, and returns $\hat\lambda_{\min}=\gamma-\vartheta_{\mathrm{best}}$. It uses $O(RM)$ distributed matrix-vector products; for fixed $K$, each such product communicates $O(\sum_{k<l}n_{kl})=O(n)$ numbers. The implementation uses $M=50$, $\tau=10^{-6}$, $R=2$, and $\delta=10^{-3}$.
}

\begin{table}[!t]
    \centering
    \caption{Absolute error of distributed shifted Lanczos estimation over 200 replications.}
    \label{tab:lanczos_error}
    \vspace{0.1cm}
    \small
    \renewcommand{\arraystretch}{0.9}
    \begin{tabular}{clcc}
        \toprule
        Setting & Method & Mean absolute error & Standard deviation\\
        \midrule
        \multirow{2}{*}{1} & ALB        & 0.000096 & 0.000484  \\
        & ALB-P & 0.00006 & 0.000235 \\
        \multirow{2}{*}{2} & ALB        & 0.000426 & 0.001468 \\
         & ALB-P & 0.000134 & 0.000593 \\
         \multirow{2}{*}{3}& ALB        & 0.000096 & 0.000484  \\
         & ALB-P & 0.00006 & 0.000235  \\
          \multirow{2}{*}{4}& ALB        & 0.000426 & 0.001468 \\
         & ALB-P & 0.000134 & 0.000593 \\
        \bottomrule
    \end{tabular}
\end{table}

\textcolor{blue}{We also assess the numerical accuracy of Algorithm \ref{alg:shifted_lanczos} with the above implementation configuration under the four simulation settings in Section \ref{subsec: MSE simulation}. For each replication, we set $a_1=a_2=1$ and report
$\abs{\hat\lambda_{\min}-\lambda_{\min}\{\hSigma_{\times}(1,1)\}}$, where $\hat\lambda_{\min}$ is returned by Algorithm \ref{alg:shifted_lanczos}, whereas $\lambda_{\min}\{\hSigma_{\times}(1,1)\}$ is obtained by direct eigendecomposition of the assembled matrix solely for evaluation. Thus, $\hSigma_{\times}(1,1)=\tSigma$ for ALB and $\hSigma_{\times}(1,1)=\tSigma^{\mathrm{priv}}$ for ALB-P.} Table \ref{tab:lanczos_error} shows that the mean absolute error is at most $4.26\times10^{-4}$ for ALB and $1.34\times10^{-4}$ for ALB-P. The distributed estimates therefore closely match the exact eigenvalues across all settings, confirming the numerical accuracy of the shifted Lanczos step used to select the covariance-regularization parameters. \textcolor{blue}{Settings 1 and 3 share the same covariance structure, as do Settings 2 and 4; consequently, their corresponding Lanczos error summaries coincide.}

\begin{breakablealgorithm}{Distributed Shifted Lanczos Estimation}{alg:shifted_lanczos}
{\renewcommand{\baselinestretch}{0.93}\selectfont
\begin{algorithmic}[1]
\Statex \textbf{Input:} $(a_1,a_2,c_{tr})$, local designs $\{\bX^k\}_{k=0}^K$, cross-user designs $\{\bX_{\times}^{k,l}\}_{k\ne l}$, denominators $\{n_k,n_{kl}\}$, maximum Lanczos steps $M$, tolerance $\tau$, shift margin $\delta$, number of random starts $R$.
\Statex \textbf{Output:} \textcolor{blue}{Estimate $\hat\lambda_{\min}$ of $\lambda_{\min}(\hSigma_{\times})$.}
\State Define the block inner product $\langle\bx,\by\rangle_{\mathrm{blk}}\gets\sum_{k=0}^K\bx_k^\top\by_k$ and norm $\norm{\bx}_{\mathrm{blk}}\gets\langle\bx,\bx\rangle_{\mathrm{blk}}^{1/2}$.
\State Set $d_0\gets(1-a_1)c_{tr}/p$ and form $\mathbf G$ from the scalar block norms above.
\State $\gamma\gets(1+\delta)\norm{\mathbf G}_2+10^{-6}$, $m\gets\min\{M,p\}$, and $\vartheta_{\mathrm{best}}\gets-\infty$.
\For{$s=1,\ldots,R$}
    \State Each $B_k$ draws $\mathbf q_k\sim\mathcal N(\mathbf0,\mathrm I_{p_k})$; jointly normalize $\mathbf q\gets\mathbf q/\norm{\mathbf q}_{\mathrm{blk}}$.
    \State Set $\mathbf q_{\mathrm{prev}}\gets\mathbf0$, $\omega_{\mathrm{prev}}\gets0$, $\mathcal V\gets\varnothing$, and initialize empty sequences $\{\alpha_t\}$ and $\{\omega_{t+1}\}$.
    \For{$t=1,\ldots,m$}
        \State Append $\mathbf q$ to the distributed basis $\mathcal V$.
        \For{$l=0,\ldots,K$}
            \State $B_l$ sends $\boldsymbol\zeta_{l\to k}=\bX_{\times}^{l,k}\mathbf q_l$ to every $B_k$, $k\ne l$.
        \EndFor
        \For{$k=0,\ldots,K$}
            \State $B_k$ computes
            $\displaystyle \bu_k\gets\frac{a_1}{n_k}\bX^{k\top}\bX^k\mathbf q_k
            +a_2\sum_{l\ne k}\frac{\bX_{\times}^{k,l\top}\boldsymbol\zeta_{l\to k}}{n_{kl}}+d_0\mathbf q_k$.
            \State $\bw_k\gets\gamma\mathbf q_k-\bu_k-\omega_{\mathrm{prev}}\mathbf q_{\mathrm{prev},k}$.
        \EndFor
        \State $\alpha_t\gets\langle\mathbf q,\bw\rangle_{\mathrm{blk}}$ and $\bw\gets\bw-\alpha_t\mathbf q$.
        \For{$r=1,2$} \Comment{two-pass full reorthogonalization}
            \For{each $\bv\in\mathcal V$}
                \State $c\gets\langle\bv,\bw\rangle_{\mathrm{blk}}$ and $\bw\gets\bw-c\bv$.
            \EndFor
        \EndFor
        \State $\omega_{t+1}\gets\norm{\bw}_{\mathrm{blk}}$.
        \State Form $\mathbf T_t=\operatorname{tridiag}(\omega_2,\ldots,\omega_t;\alpha_1,\ldots,\alpha_t)$.
        \State Compute the largest eigenpair $(\vartheta_t,\bz_t)$ of $\mathbf T_t$ and set $\vartheta_{\mathrm{best}}\gets\max\{\vartheta_{\mathrm{best}},\vartheta_t\}$.
        \State $r_t\gets\omega_{t+1}|(\bz_t)_t|$ and $\hat\lambda_{\min}^{(t)}\gets\gamma-\vartheta_t$.
        \If{$r_t\le\tau\max\{1,|\hat\lambda_{\min}^{(t)}|\}$}
            \State \textbf{break}
        \EndIf
        \State $\mathbf q_{\mathrm{prev}}\gets\mathbf q$, $\mathbf q\gets\bw/\omega_{t+1}$, and $\omega_{\mathrm{prev}}\gets\omega_{t+1}$.
    \EndFor
\EndFor
\State \Return $\hat\lambda_{\min}\gets\gamma-\vartheta_{\mathrm{best}}$.
\end{algorithmic}}
\end{breakablealgorithm}

\subsection{Coordinate Descent and the ALB Stopping Rule}
\label{app: proximal gradient descent}
The Lasso subproblems arising in the implementation can be written in the generic quadratic form
\begin{equation*}
    Q(\bbeta)=\frac{1}{2}\bbeta^{\top}\mathbf{A}\bbeta-
    \mathbf{b}^{\top}\bbeta+\lambda\lVert\bbeta\rVert_1,
    \qquad \bbeta\in\mathbb{R}^{d},
\end{equation*}
where $\mathbf A$ is symmetric positive definite and $\mathbf b\in\mathbb R^d$. With the Cholesky factorization $\mathbf A=\mathbf L\mathbf L^\top$, define
\begin{equation*}
    \widetilde{\mathbf X}=\sqrt d\,\mathbf L^{\top},
    \qquad
    \widetilde{\by}=\sqrt d\,\mathbf L^{-1}\mathbf b.
\end{equation*}
These quantities satisfy $\widetilde{\mathbf X}^{\top}\widetilde{\mathbf X}/d=\mathbf A$ and $\widetilde{\mathbf X}^{\top}\widetilde{\by}/d=\mathbf b$. Hence, the corresponding least-squares Lasso objective equals $Q(\bbeta)$ up to an additive constant. We solve it with \texttt{scikit-learn}'s coordinate-descent routine \texttt{lasso\_path}, using the current block estimate as a warm start.

\begin{algorithm}[!t]
\caption{Coordinate-Descent Solver and ALB Stopping Rule}
\label{alg:lasso_path}
\begin{algorithmic}[1]
\Statex \textbf{Input:} $\mathbf A$, $\mathbf b$, $\lambda\geq0$, optional warm start $\bbeta^{(0)}$, coordinate-descent tolerance $\tau_{\mathrm{CD}}$, maximum coordinate-descent iterations $T_{\max}$, previous ALB loss $q_{\mathrm{prev}}$, and outer tolerance $\tau_{\mathrm{ALB}}$.
\Statex \textbf{Output:} Updated block $\widehat\bbeta$ and the ALB stopping decision.
\State If $\lambda=0$, set $\widehat\bbeta\gets\mathbf A^{-1}\mathbf b$.
\State If $\lambda>0$, compute the Cholesky factorization $\mathbf A=\mathbf L\mathbf L^{\top}$.
\State Set $\widetilde{\mathbf X}\gets\sqrt d\,\mathbf L^{\top}$ and $\widetilde{\by}\gets\sqrt d\,\mathbf L^{-1}\mathbf b$.
\State Call \texttt{sklearn.linear\_model.lasso\_path} with $\texttt{alphas}=\{\lambda\}$, $\texttt{precompute}=d\mathbf A$, $\texttt{Xy}=d\mathbf b$, optional initialization $\bbeta^{(0)}$, tolerance $\tau_{\mathrm{CD}}$, and maximum iterations $T_{\max}$; set $\widehat\bbeta$ to the returned coefficient vector.
\State After all user blocks have been updated, compute $q_e\gets\cQ_n(\hbbeta^{(e)})$ and $\delta_e\gets|q_e-q_{\mathrm{prev}}|/\max\{1,|q_{\mathrm{prev}}|\}$.
\State Stop ALB if $\delta_e\leq\tau_{\mathrm{ALB}}$; otherwise set $q_{\mathrm{prev}}\gets q_e$ and begin the next sweep, subject to $e_{\max}$.
\State \Return $\widehat\bbeta$ and the stopping decision.
\end{algorithmic}
\end{algorithm}

The inner coordinate-descent iterations use the tolerance and iteration limit supplied to \texttt{lasso\_path}. After sweep $e$, ALB computes the relative penalized-loss change
\begin{equation*}
    \delta_e=
    \frac{\left|\cQ_n(\hbbeta^{(e)})-\cQ_n(\hbbeta^{(e-1)})\right|}
    {\max\left\{1,\left|\cQ_n(\hbbeta^{(e-1)})\right|\right\}}.
\end{equation*}
The cyclic updates stop when $\delta_e\leq\tau_{\mathrm{ALB}}$ or the maximum number of sweeps is reached. \textcolor{blue}{The global objective value used in the stopping rule is obtained by
summing user‑local scalar contributions and scalar pairwise
cross‑term contributions. No covariance block or raw feature matrix is centralized for this calculation.} Algorithm \ref{alg:lasso_path} summarizes the transformation and stopping rule.

\subsection{Distributed Confidence-Interval Construction}
\label{app: pseudocode}
This subsection gives an end-to-end implementation of the confidence interval in \eqref{distributed CI}. Fix a target coordinate $j$, partition its estimated precision column as $\hbthe_j=(\hbthe_j^{0\top},\ldots,\hbthe_j^{K\top})^\top$, and let $B_0$ aggregate the final scalar messages. The precision column is obtained by applying the cyclic updates in Algorithm \ref{alg:optimization} to \eqref{global inverse column}: replace the coefficient block $\bbeta_k$ by $\bthe_k$, replace the response-covariate linear term by the corresponding block $\be_{j,k}$ of $\be_j$, replace $\lambda$ by $\lambda_j$, and use the precision-estimation parameters $(a_{1,j},a_{2,j})$ in place of the point-estimation parameters $(a_1,a_2)$. Thus, estimating $\hbthe_j$ has the same $O(n)$ communication order per sweep as estimating $\hbbeta$.

For the distributed variance calculation, partition $\hbphi_i=(\hbphi_i^{0\top},\ldots,\hbphi_i^{K\top})^\top$. Each user forms its block of $S_n(\hbbeta)$ from local moments and the summaries $\bu_{l\to k}(\hbbeta_l)$, then sends its scalar contribution to $\hbthe_j^\top S_n(\hbbeta)$ to $B_0$. User $B_k$ also sends $z_{j,i}^k=\hbthe_j^{k\top}\hbphi_i^k$ for $i=1,\ldots,n$. User $B_0$ aggregates these contributions by sample ID before centering and squaring.

Algorithm \ref{alg:distributed_ci} summarizes the procedure. It computes $\hb_j$ by decomposing $\hbthe_j^\top S_n(\hbbeta)$ across users and then evaluates \eqref{estimator of variance} after aggregating contributions with the same sample ID.

\begin{algorithm}[!t]
\caption{Distributed Confidence Interval for Coordinate $j$}
\label{alg:distributed_ci}
{\renewcommand{\baselinestretch}{0.93}\selectfont
\begin{algorithmic}[1]
    \Statex \textbf{Input:} Target $j$, level $1-\alpha$, $\hbbeta_0,\ldots,\hbbeta_K$, $(a_{1,j},a_{2,j},\lambda_j)$, maximum sweeps $e_{\max}$, initial precision-column blocks, and local data and missingness information.
    \Statex \textbf{Output:} Interval $CI_j$ in \eqref{distributed CI}.
    \State Let $k_j$ hold coordinate $j$, and let $\be_{j,k}$ be the block of $\be_j$ held by $B_k$.
    \State Apply Algorithm \ref{alg:optimization} to \eqref{global inverse column}, replacing $\bbeta_k$ by $\bthe_k$, the linear term by $\be_{j,k}$, $\lambda$ by $\lambda_j$, and $(a_1,a_2)$ by $(a_{1,j},a_{2,j})$, to obtain $\hbthe_j^0,\ldots,\hbthe_j^K$.
    \State For every $l\neq k$, $B_l$ sends $\bu_{l\to k}(\hbbeta_l)=\bX^{l,k}\hbbeta_l$ to $B_k$.
    \Statex \textbf{Debiasing and local variance contributions:}
    \For{$k=0,\ldots,K$}
        \State $B_k$ computes the $k$th block $\hs_k$ of $S_n(\hbbeta)=\tSigma\hbbeta-\tmbC$:
        \Statex \qquad $\hs_k\gets\tSigma_k\hbbeta_k+\displaystyle\sum_{l\neq k}\frac{1}{n_{kl}}\bX^{k,l\top}\bu_{l\to k}(\hbbeta_l)-\tmbC_k$.
        \State Using these summaries, $B_k$ computes $\hbphi_i^k\gets\sum_{l=0}^K\hA_{i,kl}\hbbeta_l-\hq_{i,k}$ for all $i$.
        \State $d_{j,k}\gets\hbthe_j^{k\top}\hs_k$ and $z_{j,i}^k\gets\hbthe_j^{k\top}\hbphi_i^k$ for $i=1,\ldots,n$.
        \State $B_k$ sends $d_{j,k}$ and $(z_{j,1}^k,\ldots,z_{j,n}^k)^\top$ to $B_0$.
    \EndFor
    \State If $k_j\neq0$, $B_{k_j}$ sends the scalar $\hbeta_j$ to $B_0$.
    \State $B_0$ computes $\hb_j\gets\hbeta_j-\hbthe_j^\top S_n(\hbbeta)=\hbeta_j-\displaystyle\sum_{k=0}^K d_{j,k}$.
    \State $B_0$ sets $z_{j,i}\gets\sum_{k=0}^Kz_{j,i}^k$ for every $i$, aggregating by sample ID before centering and squaring.
    \State {$\bar z_j\gets n^{-1}\sum_{i=1}^nz_{j,i}$ and $\hv_{n,j}\gets(n-1)^{-1}\sum_{i=1}^n(z_{j,i}-\bar z_j)^2$}; $r_{n,j}\gets\Phi^{-1}(1-\alpha/2)\sqrt{\hv_{n,j}/n}$.
    \State \Return $CI_j\gets[\hb_j-r_{n,j},\,\hb_j+r_{n,j}]$.
\end{algorithmic}}
\end{algorithm}

\clearpage
\begingroup
\small
\setlength{\bibsep}{0pt}
\renewcommand{\baselinestretch}{1.70}\selectfont
\putbib[citation]
\endgroup
\end{bibunit}

\end{document}

%% file: command.tex
\def\i[#1]{\textit{#1}}
\def\b[#1]{\textbf{#1}}

\def\hA{\widehat{A}}

\def\hb{\hat{b}}

\def\mbC{{\mathbf C}}
\def\tmbC{{\widetilde{\mathbf C}}}

\def\bd{{\boldsymbol d}}

\def\be{{\boldsymbol e}}

\def\cL{{\mathcal L}}

\def\bbP{{\mathbb P}}

\def\bbP{{\mathbb P}}

\def\cQ{\mathcal{Q}}
\def\hq{{\hat q}}

\def\hs{\widehat{ s}}
\def\cS{{\mathcal S}}

\def\hs{\hat{s}}

\def\bu{{\boldsymbol u}}

\def\bv{{\boldsymbol v}}
\def\hv{{\hat{v}}}

\def\bw{{\boldsymbol w}}

\def\bx{{\boldsymbol{x}}}

\def\bX{{\boldsymbol X}}

\def\by{{\boldsymbol y}}

\def\bz{{\boldsymbol z}}

\def\bz{\boldsymbol{z}}

\def\hbeta{\hat{\beta}}

\def\bbeta{\boldsymbol \beta}
\def\hbbeta{\hat{\boldsymbol \beta}}
\def\cbbeta{\check{\boldsymbol \beta}}

\def\bdelta{\boldsymbol \delta}

\def\bphi{\boldsymbol{\phi}}

\def\hbphi{\widehat{\bphi}}

\def\tSigma{\widetilde{\Sigma}}
\def\hSigma{\widehat{\Sigma}}

\def\hbthe{\hat{\boldsymbol \theta}}
\def\bthe{\boldsymbol \theta}

\def\cbthe{\check{\boldsymbol \theta}}

\def\type1{\mbox{Type-\uppercase\expandafter{\romannumeral1}}}

\DeclareMathOperator{\E}{\mathbb{E}}
\DeclareMathOperator{\Var}{\mbox{Var}}

\providecommand{\abs}[1]{\left\lvert#1\right\rvert}
\providecommand{\norm}[1]{\left\lVert#1\right\rVert}

\def\argmin{\mathop{\rm argmin}}

%% file: main.bbl
\begin{thebibliography}{46}
\providecommand{\natexlab}[1]{#1}
\providecommand{\url}[1]{\texttt{#1}}
\expandafter\ifx\csname urlstyle\endcsname\relax
  \providecommand{\doi}[1]{doi: #1}\else
  \providecommand{\doi}{doi: \begingroup \urlstyle{rm}\Url}\fi

\bibitem[Battey et~al.(2018)Battey, Fan, Liu, Lu, and
  Zhu]{battey2018distributed}
H.~Battey, J.~Fan, H.~Liu, J.~Lu, and Z.~Zhu.
\newblock Distributed testing and estimation under sparse high dimensional
  models.
\newblock \emph{Ann. Stat.}, 46\penalty0 (3):\penalty0 1352, 2018.

\bibitem[Cai et~al.(2016)Cai, Cai, and Zhang]{cai2016structured}
T.~Cai, T.~T. Cai, and A.~Zhang.
\newblock Structured matrix completion with applications to genomic data
  integration.
\newblock \emph{J. Am. Stat. Assoc.}, 111\penalty0 (514):\penalty0 621--633,
  2016.

\bibitem[Cai et~al.(2022)Cai, Liu, and Xia]{cai2022individual}
T.~Cai, M.~Liu, and Y.~Xia.
\newblock Individual data protected integrative regression analysis of
  high-dimensional heterogeneous data.
\newblock \emph{J. Am. Stat. Assoc.}, 117\penalty0 (540):\penalty0 2105--2119,
  2022.

\bibitem[Castiglia et~al.(2023)Castiglia, Zhou, Wang, Kadhe, Baracaldo, and
  Patterson]{castiglia2023less}
T.~Castiglia, Y.~Zhou, S.~Wang, S.~Kadhe, N.~Baracaldo, and S.~Patterson.
\newblock Less-vfl: Communication-efficient feature selection for vertical
  federated learning.
\newblock In \emph{ICML}, 2023.

\bibitem[Castiglia et~al.(2022)Castiglia, Das, Wang, and
  Patterson]{castiglia2022compressed}
T.~J. Castiglia, A.~Das, S.~Wang, and S.~Patterson.
\newblock Compressed-vfl: Communication-efficient learning with vertically
  partitioned data.
\newblock In \emph{ICML}, 2022.

\bibitem[Chen et~al.(2022)Chen, Liu, and Zhang]{chen2022first}
X.~Chen, W.~Liu, and Y.~Zhang.
\newblock First-order newton-type estimator for distributed estimation and
  inference.
\newblock \emph{J. Am. Stat. Assoc.}, 117\penalty0 (540):\penalty0 1858--1874,
  2022.

\bibitem[Dai et~al.(2020)Dai, Jiang, Bonomi, Li, Xiong, and
  Ohno-Machado]{dai2020verticox}
W.~Dai, X.~Jiang, L.~Bonomi, Y.~Li, H.~Xiong, and L.~Ohno-Machado.
\newblock Verticox: Vertically distributed cox proportional hazards model using
  the alternating direction method of multipliers.
\newblock \emph{IEEE Trans. Knowl. Data Eng.}, 34\penalty0 (2):\penalty0
  996--1010, 2020.

\bibitem[Diao et~al.(2022)Diao, Ding, and Tarokh]{diao2022gal}
E.~Diao, J.~Ding, and V.~Tarokh.
\newblock Gal: Gradient assisted learning for decentralized multi-organization
  collaborations.
\newblock In \emph{NeurIPS}, 2022.

\bibitem[Fan et~al.(2023{\natexlab{a}})Fan, Guo, and
  Wang]{fan2023communication}
J.~Fan, Y.~Guo, and K.~Wang.
\newblock Communication-efficient accurate statistical estimation.
\newblock \emph{J. Am. Stat. Assoc.}, 118\penalty0 (542):\penalty0 1000--1010,
  2023{\natexlab{a}}.

\bibitem[Fan et~al.(2023{\natexlab{b}})Fan, Li, and Lin]{fan2023residual}
Y.~Fan, J.-S. Li, and N.~Lin.
\newblock Residual projection for quantile regression in vertically partitioned
  big data.
\newblock \emph{Data Min. Knowl. Discov.}, 37\penalty0 (2):\penalty0 710--735,
  2023{\natexlab{b}}.

\bibitem[Fang et~al.(2021)Fang, Zhao, Tan, Chen, Yu, Wang, Wang, Zhou, and
  Zhang]{fang2021large}
W.~Fang, D.~Zhao, J.~Tan, C.~Chen, C.~Yu, L.~Wang, L.~Wang, J.~Zhou, and
  B.~Zhang.
\newblock Large-scale secure xgb for vertical federated learning.
\newblock In \emph{CIKM}, 2021.

\bibitem[Hardy et~al.(2017)Hardy, Henecka, Ivey-Law, Nock, Patrini, Smith, and
  Thorne]{hardy2017private}
S.~Hardy, W.~Henecka, H.~Ivey-Law, R.~Nock, G.~Patrini, G.~Smith, and
  B.~Thorne.
\newblock Private federated learning on vertically partitioned data via entity
  resolution and additively homomorphic encryption.
\newblock \emph{arXiv preprint arXiv:1711.10677}, 2017.

\bibitem[Huang et~al.(2022)Huang, Li, Sun, and Zhao]{huang2022coresets}
L.~Huang, Z.~Li, J.~Sun, and H.~Zhao.
\newblock Coresets for vertical federated learning: Regularized linear
  regression and $ k $-means clustering.
\newblock In \emph{NeurIPS}, 2022.

\bibitem[Javanmard and Montanari(2014)]{javanmard2014confidence}
A.~Javanmard and A.~Montanari.
\newblock Confidence intervals and hypothesis testing for high-dimensional
  regression.
\newblock \emph{J. Mach. Learn. Res.}, 15\penalty0 (1):\penalty0 2869--2909,
  2014.

\bibitem[Jordan et~al.(2019)Jordan, Lee, and Yang]{jordan2019communication}
M.~I. Jordan, J.~D. Lee, and Y.~Yang.
\newblock Communication-efficient distributed statistical inference.
\newblock \emph{J. Am. Stat. Assoc.}, 114\penalty0 (526), 2019.

\bibitem[Kang et~al.(2022)Kang, Liu, and Liang]{kang2022fedcvt}
Y.~Kang, Y.~Liu, and X.~Liang.
\newblock Fedcvt: Semi-supervised vertical federated learning with cross-view
  training.
\newblock \emph{ACM Trans. Intell. Syst. Technol.}, 13\penalty0 (4):\penalty0
  1--16, 2022.

\bibitem[Lee et~al.(2017)Lee, Liu, Sun, and Taylor]{lee2017communication}
J.~D. Lee, Q.~Liu, Y.~Sun, and J.~E. Taylor.
\newblock Communication-efficient sparse regression.
\newblock \emph{J. Mach. Learn. Res.}, 18\penalty0 (5):\penalty0 1--30, 2017.

\bibitem[Li et~al.(2025)Li, Xie, Zhang, Fu, and Li]{li2025voxel}
K.~Li, D.~Xie, Z.~Zhang, C.~Fu, and C.~Li.
\newblock Voxel- and surface-based morphometry in the cortical thickness and
  cortical and subcortical gray matter volume in patients with mild-to-moderate
  alzheimer’s disease.
\newblock \emph{Front. Aging Neurosci.}, 17:\penalty0 1546977, 2025.

\bibitem[Li et~al.(2022)Li, Liang, Chang, and Zhang]{li2022statistical}
X.~Li, J.~Liang, X.~Chang, and Z.~Zhang.
\newblock Statistical estimation and online inference via local sgd.
\newblock In \emph{COLT}, 2022.

\bibitem[Liu et~al.(2020)Liu, Liu, Liu, Liang, Meng, Zhang, and
  Zheng]{liu2020federated}
Y.~Liu, Y.~Liu, Z.~Liu, Y.~Liang, C.~Meng, J.~Zhang, and Y.~Zheng.
\newblock Federated forest.
\newblock \emph{IEEE Trans. Big Data}, 8\penalty0 (3):\penalty0 843--854, 2020.

\bibitem[Liu et~al.(2022)Liu, Zhang, Kang, Li, Chen, Hong, and
  Yang]{liu2022fedbcd}
Y.~Liu, X.~Zhang, Y.~Kang, L.~Li, T.~Chen, M.~Hong, and Q.~Yang.
\newblock Fedbcd: A communication-efficient collaborative learning framework
  for distributed features.
\newblock \emph{IEEE Trans. Signal Process.}, 70:\penalty0 4277--4290, 2022.

\bibitem[Liu et~al.(2024)Liu, Kang, Zou, Pu, He, Ye, Ouyang, Zhang, and
  Yang]{liu2024vertical}
Y.~Liu, Y.~Kang, T.~Zou, Y.~Pu, Y.~He, X.~Ye, Y.~Ouyang, Y.-Q. Zhang, and
  Q.~Yang.
\newblock Vertical federated learning: Concepts, advances, and challenges.
\newblock \emph{IEEE Trans. Knowl. Data Eng.}, 36\penalty0 (7):\penalty0
  3615--3634, 2024.

\bibitem[Mueller et~al.(2005)Mueller, Weiner, Thal, Petersen, Jack, Jagust,
  Trojanowski, Toga, and Beckett]{mueller2005alzheimer}
S.~G. Mueller, M.~W. Weiner, L.~J. Thal, R.~C. Petersen, C.~Jack, W.~Jagust,
  J.~Q. Trojanowski, A.~W. Toga, and L.~Beckett.
\newblock The alzheimer's disease neuroimaging initiative.
\newblock \emph{Neuroimaging Clin. N. Am.}, 15\penalty0 (4):\penalty0 869--877,
  2005.

\bibitem[Nestor et~al.(2008)Nestor, Rupsingh, Borrie, Smith, Accomazzi, Wells,
  Fogarty, and Bartha]{nestor2008ventricular}
S.~M. Nestor, R.~Rupsingh, M.~Borrie, M.~Smith, V.~Accomazzi, J.~L. Wells,
  J.~Fogarty, and R.~Bartha.
\newblock Ventricular enlargement as a possible measure of alzheimer's disease
  progression validated using the alzheimer's disease neuroimaging initiative
  database.
\newblock \emph{Brain}, 131\penalty0 (9):\penalty0 2443--2454, 2008.

\bibitem[Ren et~al.(2022)Ren, Yang, and Chen]{ren2022improving}
Z.~Ren, L.~Yang, and K.~Chen.
\newblock Improving availability of vertical federated learning: Relaxing
  inference on non-overlapping data.
\newblock \emph{ACM Trans. Intell. Syst. Technol.}, 13\penalty0 (4):\penalty0
  58:1--58:20, 2022.

\bibitem[Rizvi et~al.(2021)Rizvi, Lao, Chesebro, Dworkin, Amarante, Beato,
  Gutierrez, Zahodne, Schupf, Manly, Mayeux, and
  Brickman]{rizvi2021association}
B.~Rizvi, P.~J. Lao, A.~G. Chesebro, J.~D. Dworkin, E.~Amarante, J.~M. Beato,
  J.~Gutierrez, L.~B. Zahodne, N.~Schupf, J.~J. Manly, R.~Mayeux, and A.~M.
  Brickman.
\newblock Association of regional white matter hyperintensities with
  longitudinal alzheimer-like pattern of neurodegeneration in older adults.
\newblock \emph{JAMA Netw. Open}, 4\penalty0 (10), 2021.

\bibitem[Song et~al.(2024)Song, Lin, and Zhou]{song2024semi}
S.~Song, Y.~Lin, and Y.~Zhou.
\newblock Semi-supervised inference for block-wise missing data without
  imputation.
\newblock \emph{J. Mach. Learn. Res.}, 25\penalty0 (99):\penalty0 1--36, 2024.

\bibitem[Sun et~al.(2023)Sun, Xu, Yang, Nath, Li, Zhao, Xu, Chen, and
  Roth]{sun2023communication}
J.~Sun, Z.~Xu, D.~Yang, V.~Nath, W.~Li, C.~Zhao, D.~Xu, Y.~Chen, and H.~R.
  Roth.
\newblock Communication-efficient vertical federated learning with limited
  overlapping samples.
\newblock In \emph{ICCV}, pages 5203--5212, 2023.

\bibitem[Sun and Xia(2024)]{sun2024optimal}
Y.~Sun and Y.~Xia.
\newblock Optimal integrative estimation for distributed precision matrices
  with heterogeneity adjustment.
\newblock \emph{arXiv preprint arXiv:2408.06263}, 2024.

\bibitem[Tang et~al.(2025)Tang, Feng, Li, and Wang]{tang2025data}
H.~Tang, L.~Feng, Y.~Li, and F.~Wang.
\newblock Data privatization in vertical federated learning with client-wise
  missing problem.
\newblock \emph{arXiv preprint arXiv:2511.20876}, 2025.

\bibitem[Tu et~al.(2023)Tu, Liu, Mao, and Xu]{tu2023distributed}
J.~Tu, W.~Liu, X.~Mao, and M.~Xu.
\newblock Distributed semi-supervised sparse statistical inference.
\newblock \emph{IEEE Trans. Inf. Theory}, 70\penalty0 (6):\penalty0 4197--4217,
  2023.

\bibitem[Valdeira et~al.(2025)Valdeira, Wang, and Chi]{valdeira2025vertical}
P.~Valdeira, S.~Wang, and Y.~Chi.
\newblock Vertical federated learning with missing features during training and
  inference.
\newblock In \emph{ICLR}, 2025.

\bibitem[van~de Geer et~al.(2014)van~de Geer, B{\"u}hlmann, Ritov, and
  Dezeure]{van2014asymptotically}
S.~van~de Geer, P.~B{\"u}hlmann, Y.~Ritov, and R.~Dezeure.
\newblock On asymptotically optimal confidence regions and tests for
  high-dimensional models.
\newblock \emph{Ann. Stat.}, 42\penalty0 (3):\penalty0 1166--1202, 2014.

\bibitem[Wang et~al.(2023)Wang, Gu, Zhang, Li, Wang, and Ling]{wang2023unified}
G.~Wang, B.~Gu, Q.~Zhang, X.~Li, B.~Wang, and C.~X. Ling.
\newblock A unified solution for privacy and communication efficiency in
  vertical federated learning.
\newblock In \emph{NeurIPS}, 2023.

\bibitem[Wang et~al.(2022)Wang, Zhang, Hong, Yang, and Ding]{wang2022parallel}
X.~Wang, J.~Zhang, M.~Hong, Y.~Yang, and J.~Ding.
\newblock Parallel assisted learning.
\newblock \emph{IEEE Trans. Signal Process.}, 70:\penalty0 5848--5858, 2022.

\bibitem[Xian et~al.(2020)Xian, Wang, Ding, and Ghanadan]{xian2020assisted}
X.~Xian, X.~Wang, J.~Ding, and R.~Ghanadan.
\newblock Assisted learning: A framework for multi-organization learning.
\newblock In \emph{NeurIPS}, 2020.

\bibitem[Xiang et~al.(2013)Xiang, Yuan, Fan, Wang, Thompson, and
  Ye]{xiang2013multi}
S.~Xiang, L.~Yuan, W.~Fan, Y.~Wang, P.~M. Thompson, and J.~Ye.
\newblock Multi-source learning with block-wise missing data for alzheimer's
  disease prediction.
\newblock In \emph{KDD}, 2013.

\bibitem[Xiang et~al.(2014)Xiang, Yuan, Fan, Wang, Thompson, and
  Ye]{xiang2014bilevel}
S.~Xiang, L.~Yuan, W.~Fan, Y.~Wang, P.~M. Thompson, and J.~Ye.
\newblock Bi-level multi-source learning for heterogeneous block-wise missing
  data.
\newblock \emph{NeuroImage}, 102:\penalty0 192--206, 2014.

\bibitem[Xue and Qu(2021)]{xue2021integrating}
F.~Xue and A.~Qu.
\newblock Integrating multisource block-wise missing data in model selection.
\newblock \emph{J. Am. Stat. Assoc.}, 116\penalty0 (536):\penalty0 1914--1927,
  2021.

\bibitem[Xue et~al.(2025)Xue, Ma, and Li]{xue2025statistical}
F.~Xue, R.~Ma, and H.~Li.
\newblock Statistical inference for high-dimensional linear regression with
  blockwise missing data.
\newblock \emph{Stat. Sin.}, 35\penalty0 (1):\penalty0 431--456, 2025.

\bibitem[Yang et~al.(2022)Yang, Du, Gao, Liu, Chen, Wang, Liu, Lv, Zhang, Xia,
  et~al.]{yang2022associations}
A.~Yang, L.~Du, W.~Gao, B.~Liu, Y.~Chen, Y.~Wang, X.~Liu, K.~Lv, W.~Zhang,
  H.~Xia, et~al.
\newblock Associations of cortical iron accumulation with cognition and
  cerebral atrophy in alzheimer’s disease.
\newblock \emph{Quant. Imaging Med. Surg.}, 12\penalty0 (9):\penalty0 4570,
  2022.

\bibitem[Yin et~al.(2026)Yin, Wang, Zhu, Zhu, and Huang]{yin2026vertical}
H.~Yin, L.~Wang, Y.~Zhu, L.~Zhu, and D.~Huang.
\newblock Vertical federated feature screening.
\newblock In \emph{NeurIPS}, 2026.

\bibitem[Yu et~al.(2020)Yu, Li, Shen, and Liu]{yu2020optimal}
G.~Yu, Q.~Li, D.~Shen, and Y.~Liu.
\newblock Optimal sparse linear prediction for block-missing multi-modality
  data without imputation.
\newblock \emph{J. Am. Stat. Assoc.}, 115\penalty0 (531):\penalty0 1406--1419,
  2020.

\bibitem[Yuan et~al.(2012)Yuan, Wang, Thompson, Narayan, Ye, Initiative,
  et~al.]{yuan2012multi}
L.~Yuan, Y.~Wang, P.~M. Thompson, V.~A. Narayan, J.~Ye, A.~D.~N. Initiative,
  et~al.
\newblock Multi-source feature learning for joint analysis of incomplete
  multiple heterogeneous neuroimaging data.
\newblock \emph{NeuroImage}, 61\penalty0 (3):\penalty0 622--632, 2012.

\bibitem[Zhang and Zhang(2014)]{zhang2014confidence}
C.-H. Zhang and S.~S. Zhang.
\newblock Confidence intervals for low dimensional parameters in high
  dimensional linear models.
\newblock \emph{J. R. Stat. Soc. Ser. B}, 76\penalty0 (1):\penalty0 217--242,
  2014.

\bibitem[Zhang et~al.(2026)Zhang, Yang, and Ding]{zhang2026additive}
J.~Zhang, Y.~Yang, and J.~Ding.
\newblock Additive-effect assisted learning.
\newblock \emph{J. R. Stat. Soc. Ser. B}, 88\penalty0 (2):\penalty0 657--676,
  2026.

\end{thebibliography}


\begin{thebibliography}{9}
\providecommand{\natexlab}[1]{#1}
\providecommand{\url}[1]{\texttt{#1}}
\expandafter\ifx\csname urlstyle\endcsname\relax
  \providecommand{\doi}[1]{doi: #1}\else
  \providecommand{\doi}{doi: \begingroup \urlstyle{rm}\Url}\fi

\bibitem[Bickel and Levina(2008)]{bickel2008regularized}
P.~J. Bickel and E.~Levina.
\newblock Regularized estimation of large covariance matrice.
\newblock \emph{Ann. Stat.}, 36\penalty0 (1):\penalty0 199--227, 2008.

\bibitem[Duchi et~al.(2013)Duchi, Jordan, and Wainwright]{duchi2013local}
J.~C. Duchi, M.~I. Jordan, and M.~J. Wainwright.
\newblock Local privacy and statistical minimax rates.
\newblock In \emph{FOCS}, 2013.

\bibitem[Dwork and Roth(2014)]{dwork2014algorithmic}
C.~Dwork and A.~Roth.
\newblock The algorithmic foundations of differential privacy.
\newblock \emph{Found. Trends Theor. Comput. Sci.}, 9\penalty0 (3-4):\penalty0
  211--487, 2014.

\bibitem[Kasiviswanathan et~al.(2011)Kasiviswanathan, Lee, Nissim,
  Raskhodnikova, and Smith]{kasiviswanathan2011can}
S.~P. Kasiviswanathan, H.~K. Lee, K.~Nissim, S.~Raskhodnikova, and A.~Smith.
\newblock What can we learn privately?
\newblock \emph{SIAM. J. Comp.}, 40\penalty0 (3):\penalty0 793--826, 2011.

\bibitem[Mf(1975)]{mf1975mini}
F.~Mf.
\newblock " mini-mental state". a practical method for grading the cognitive
  state of patients for the clinician.
\newblock \emph{J. Psychiatr. Res.}, 12:\penalty0 189--198, 1975.

\bibitem[Ravikumar et~al.(2011)Ravikumar, Wainwright, Raskutti, and
  Yu]{ravikumar2011high}
P.~Ravikumar, M.~J. Wainwright, G.~Raskutti, and B.~Yu.
\newblock High-dimensional covariance estimation by minimizing
  $\ell_1$-penalized log-determinant divergence.
\newblock \emph{Electron. J. Stat.}, 5:\penalty0 935, 2011.

\bibitem[Tombaugh and McIntyre(1992)]{tombaugh1992mini}
T.~N. Tombaugh and N.~J. McIntyre.
\newblock The mini-mental state examination: a comprehensive review.
\newblock \emph{J. Am. Geriatr. Soc.}, 40\penalty0 (9):\penalty0 922--935,
  1992.

\bibitem[Yu et~al.(2020)Yu, Li, Shen, and Liu]{yu2020optimal}
G.~Yu, Q.~Li, D.~Shen, and Y.~Liu.
\newblock Optimal sparse linear prediction for block-missing multi-modality
  data without imputation.
\newblock \emph{J. Am. Stat. Assoc.}, 115\penalty0 (531):\penalty0 1406--1419,
  2020.

\bibitem[Zhou(2024)]{zhou2018fencheldualitystrongconvexity}
X.~Zhou.
\newblock On the fenchel duality between strong convexity and lipschitz
  continuous gradient.
\newblock \emph{arXiv preprint arXiv:1803.06573}, 2024.

\end{thebibliography}
